\documentclass{article}
\usepackage[a4paper]{geometry}

\usepackage{algorithm}
\usepackage{algpseudocode}
\usepackage{amsfonts}
\usepackage{amsmath}
\usepackage{amssymb}
\usepackage{amsthm} 
\usepackage[english]{babel}
\usepackage[numbers]{natbib}
\usepackage{csquotes}
\usepackage{enumitem}
\usepackage{environ}
\usepackage{etoolbox}
\usepackage[T1]{fontenc}
\usepackage{graphicx}
\usepackage[hidelinks]{hyperref}
\usepackage{mathrsfs}
\usepackage{mathtools}
\usepackage{refcount}
\usepackage{subcaption}
\usepackage{thmtools} 
\usepackage{thm-restate}
\usepackage{xspace}

\usepackage{cleveref}

\usepackage{tikz}
\usetikzlibrary{graphs, calc}

\algdef{SE}[FUNC]{Function}{EndFunction}[1]{\textbf{Function} #1:}{}
\algtext*{EndFunction}

\algdef{SE}[FE]{ForEach}{EndForEach}[1]{\textbf{foreach} #1 \textbf{do}}{}
\algtext*{EndForEach}

\algtext*{EndWhile}
\algtext*{EndFor}
\algtext*{EndProcedure}
\algtext*{EndIf}

\allowdisplaybreaks

\Crefname{subsection}{Section}{Sections}
\Crefname{enumi}{}{}

\declaretheorem[name=Theorem, numberwithin=section]{theorem}
\declaretheorem[sibling=theorem, name=Lemma]{lemma}

\declaretheorem[sibling=theorem, name=Proposition]{proposition}
\declaretheorem[sibling=theorem, name=Observation]{observation}

\declaretheorem[sibling=theorem, name=Definition, style=definition]{definition}

\let\cref\Cref
\let\textcite\citet

\newcommand{\oracle}{\operatorname{oracle}}
\newcommand{\heavy}{\textsf{heavy}\xspace}
\newcommand{\light}{\textsf{light}\xspace}
\newcommand{\invalid}{\textsf{invalid}\xspace}

\newcommand{\CCdetect}{\CC_\text{detect}}
\newcommand{\isLLM}{\texttt{IsLLM}\xspace}
\newcommand{\IsLLM}{\isLLM}
\newcommand{\isHeavy}{\texttt{IsHeavy}\xspace}

\newcommand{\CC}{\mathcal{C}}

\newcommand{\IP}{\mathbb{P}}
\newcommand{\IE}{\mathbb{E}}
\newcommand{\IV}{\operatorname{Var}}
\newcommand{\CB}{\mathcal{B}}

\newcommand{\CI}{\mathcal{I}}
\renewcommand{\CI}{\mathcal{K}}

\newcommand{\CT}{\mathcal{T}}

\newcommand{\CS}{\mathcal{S}}

\newcommand{\poly}{\mathrm{poly}}

\renewcommand{\wedge}{w} 
\newcommand{\substr}{z} 
\newcommand{\os}{\operatorname{os}}
\newcommand{\ow}{\operatorname{ow}}

\newcommand{\FC}{\mathfrak{C}}
\newcommand{\FCvalid}{\FC_{\text{valid}}}

\renewcommand{\epsilon}{\varepsilon}
\newcommand{\eps}{\varepsilon}

\newcommand{\tildeO}{\widetilde{O}}
\renewcommand{\tilde}{\widetilde}

\newcommand{\epsoracleW}{\varepsilon_W}
\newcommand{\epsoracleE}{\varepsilon_E}
\newcommand{\epsoracleV}{\varepsilon_V}

\newcommand{\squares}{\CT}

\newcommand{\Var}[1]{\IV\left[ #1 \right]}
\newcommand{\Cov}[1]{\mathrm{Cov}\left[ #1 \right]}
\newcommand{\Expectation}[1]{\IE\left[ #1 \right]}
\newcommand{\Prob}[1]{\IP\left[ #1 \right]}

\date{}
\title{\Large Near-Optimal Four-Cycle Counting in Graph Streams}

\newcommand{\email}[1]{\text{#1}}
\author{Sebastian Lüderssen\thanks{TU Wien, Vienna, Austria (\email{sebastian.luederssen@tuwien.ac.at}, \email{stefan.neumann@tuwien.ac.at}). This research has been funded by the Vienna Science and Technology Fund (WWTF) [Grant ID: 10.47379/VRG23013].} \and Stefan Neumann\footnotemark[1]
\and Pan Peng\thanks{School of Computer Science and Technology, University of Science and Technology of China, Hefei, China (\email{ppeng@ustc.edu.cn}). Supported in part by NSFC Grant 62272431 and the Innovation Program for Quantum Science and Technology (Grant No. 2021ZD0302901).}}

\begin{document}

\maketitle

\begin{abstract}
    \noindent
    We study four-cycle counting in arbitrary order graph streams.
    We present a 3-pass algorithm for $(1+\varepsilon)$-approximating the number of four-cycles using $\tildeO(m/\sqrt{T})$ space, where $m$ is the number of edges and $T$ the number of four-cycles in the graph.
    This improves upon 
    a 3-pass algorithm by Vorotnikova using space $\tildeO(m/T^{1/3})$ and matches a multi-pass lower bound of $\Omega(m/\sqrt{T})$  by McGregor and Vorotnikova.
\end{abstract}

\section{Introduction}
Subgraph detection and counting in large graphs is a fundamental problem in computer science. Given an undirected graph $G$ and a small pattern graph $H$ (e.g., a triangle or a four-cycle), the goal is to detect or count the number of subgraphs (not necessarily induced) of $G$ that are isomorphic to $H$. Such small pattern graphs, or motifs \cite{milo2002network}, play a crucial role in understanding structural properties of networks. Subgraph counting has diverse applications in different areas, including social network analysis \cite{pavan2013counting}, biological systems \cite{grochow2007network}, and database theory \cite{atserias2013size}.

This problem has been extensively studied in the \emph{graph streaming model}, which is relevant for massive graphs where storing the entire graph in main memory is infeasible. In this model, the edges of the graph arrive sequentially in an arbitrary (potentially adversarial) order, and the algorithm has limited memory and may make a constant number of passes over the stream. Depending on whether a single pass or multiple passes are used, we refer to the corresponding settings as the \emph{single-pass} or \emph{multi-pass} streaming models. While single-pass algorithms are necessary in time-critical settings like network monitoring, multi-pass algorithms are highly effective in cases where the data is too large to fit into main memory but multiple passes over the data are feasible, such as in external memory or I/O-bound computations. 

Formally, let $G=(V,E)$ be an undirected graph with $n$ vertices and $m$ edges that is represented as a graph stream, and let $H$ be a small pattern graph. Denote by $T$ the number of subgraphs in~$G$ isomorphic to $H$. The objective of subgraph counting in the streaming model is to solve one of the following tasks using sublinear space and a small number of passes:
\begin{itemize}
\item \textbf{Detection:} Distinguish between the case where $G$ contains no copy of $H$ and the case where it contains at least $T$ copies.

\item \textbf{Counting:} Estimate $T$ within a factor of $1\pm\varepsilon$.
\end{itemize}

While the majority of streaming algorithms for detecting/counting small pattern graphs have focused on triangles \cite{bar2002reductions,jowhari2005new,buriol2006counting,ahn2012graph,braverman2013hard,kolountzakis2012efficient,pavan2013counting,cormode2017second,mcgregor2016better,bulteau2016triangle,bera2017towards,tsourakakis2009doulion,pagh2012colorful,jayaram2021optimal,kallaugher2017hybrid,manjunath2011approximate,Bera2020how}, other more complex pattern graphs including cycles \cite{manjunath2011approximate,bera2017towards,mcgregor2020triangle} and cliques \cite{pavan2013counting,bera2017towards,fichtenbergerCounting2022} have also received significant attention. There has also been work on counting general subgraphs \cite{kane2012counting,bera2017towards,kallaugher2018sketching,assadi2019simple,fichtenbergerCounting2022}, but, due to the generality of these approaches, they often yield suboptimal space bounds for specific pattern graphs when compared to more tailored approaches.

\bigskip
In this paper, we focus on four-cycles, i.e., cycles of length four, which are the next simplest class of pattern graphs after triangles.
In bipartite graphs, four-cycles (also known as butterflies) are the simplest non-trivial motifs, playing a similar role to triangles in general graphs. 
They are used to cluster users in online recommendation systems \cite{abuoda2019link}  and for computing the butterfly clustering coefficient, a standard measure of cohesiveness in bipartite networks \cite{lind2005cycles,aksoy2017measuring}.

Despite their simplicity, the streaming complexity of detecting and counting four-cycles is, perhaps surprisingly, not yet fully understood. 
This is in stark contrast to the space complexity of triangle counting, which has been well-characterized (see, e.g., \cite{mcgregor2016better,bera2017towards}). 
Indeed, \citet{mcgregor2020triangle} observed: ``\emph{Estimating the number of 4-cycles in the arbitrary order model is significantly more challenging than estimating the number of triangles}''.

Due to these challenges, there remains a significant gap in the current state of
the art for multi-pass streaming algorithms for four-cycle detection and
counting.  Most significantly, \textcite{mcgregor2020triangle} provided a $3$-pass counting
algorithm using $\tildeO(m/T^{1/4})$ space\footnote{We write $\tildeO(\cdot)$ to hide $\poly(\log n, 1/\eps)$ factors.}, and a $2$-pass detection algorithm using
$\tilde{O}(m^{3/2}/T^{3/4})$ space. \textcite{vorotnikova2020improved} gave a
3-pass counting algorithm with $\tilde{O}(m/T^{1/3})$ space and conjectured that
this space usage is optimal. However, the best-known lower bound remains
$\Omega(m/\sqrt{T})$ by \citet{mcgregor2020triangle}, yielding a polynomial gap
in the dependency on the number of four-cycles~$T$, which in practice is typically polynomial in~$m$.
\smallskip

\paragraph{Our contributions.}  
In this work, we present an algorithm that matches the known space lower bound
of $\Omega(m/\sqrt{T})$ by \citet{mcgregor2020triangle} up to polylogarithmic
factors and dependencies on $1/\eps$. Our main result is as follows.
   
\begin{restatable}{theorem}{thmcounting}
\label{thm:counting}
	There exists a $3$-pass algorithm that obtains a graph $G$ as an
	arbitrary-order edge stream and returns a $(1\pm\eps)$-approximation of the
	number of four-cycles in $G$ using space $\tildeO(m/\sqrt{T})$ with high
	probability.
\end{restatable}

Furthermore, for four-cycle \emph{detection} (as opposed to counting), our
algorithm requires only two passes while also achieving a space
complexity of $\tildeO(m/\sqrt{T})$ (see \Cref{thm:detect}). 

Even though our focus is on insertion-only streams (i.e., streams containing only edge
insertions), we note that our algorithm can be readily extended to the turnstile
model (which allows both insertions and deletions), as its core subroutines ---
sampling nodes and edges --- can be efficiently implemented even in the presence
of deletions. We also note that our algorithm remains correct even if the edge order changes across different stream passes.
\smallskip

\paragraph{Technical contributions.}
Unlike many prior works, our algorithm relies almost entirely on node sampling. It samples multiple sets of nodes using different probabilities, and then, during multiple passes, collects a subset of edges from their induced subgraphs. We then use the number of four-cycles with certain properties in the resulting subgraph to perform detection or to define our estimator for counting. This allows us to obtain a space usage of just $\tildeO(m/\sqrt{T})$ in expectation, but it comes at significant challenges for our algorithm analysis, especially in the variance analysis, due to the complex dependencies among four-cycles introduced by our sampling scheme. 

To control the variance, we need to identify nodes (or edges and wedges, i.e., paths of length two) that participate in ``many'' four-cycles and exclude those four-cycles from our estimators. To this end, we enhance the commonly used concept of \emph{heavy} nodes (edges or wedges).  
The \emph{heaviness} of a node  $v$ is the number of four-cycles in which~$v$ appears, and $v$ is considered \emph{heavy} if its heaviness is above a \emph{heaviness threshold} that we set in our analysis. (The heaviness of an edge or wedge is defined analogously.) 
In most previous works on four-cycle (or other subgraph) counting, this heaviness threshold was uniform across all nodes, making \emph{being heavy} a property that depends only on the graph itself. In contrast, in our analysis, for each node, edge or wedge \emph{and each possible way it might get sampled}, we must define an individual heaviness threshold to properly bound the variance of our estimator. That is, whether a node is heavy depends on its sampling probability.

Moreover, we observe that a straightforward definition of node heaviness cannot be efficiently estimated within our target space bounds. To overcome this, we introduce \emph{a refined notion of heaviness}, which can be checked efficiently while remaining sufficiently accurate for our variance analysis. This refinement, however, introduces additional technical challenges. For example, it may cause certain four-cycles to be double-counted. We show that such overcounting affects only an $\varepsilon$-fraction of all four-cycles and therefore does not impact our final bounds.

We believe our techniques are of independent interest and, due to the generality
of our vertex sampling approach, they will also be applicable to other subgraph counting
problems in the streaming model. We present a more detailed overview of our
algorithm and our analysis in \Cref{sec:overview,sec:strategy}.
\smallskip

\paragraph{Other related work.} 
For counting four-cycles using single-pass algorithms,
\citet{manjunath2011approximate} provided a $\tilde{O}(m^4/T^2)$-space algorithm.
From a lower bound perspective, we observe that the lower bounds by \citet{bera2017towards}, \citet{kallaugher2019complexity} (for the stronger adjacency list model) and by \citet{braverman2013hard} (after extending their bound from triangle-counting to counting four-cycles)
imply the following: There is no single-pass algorithm with space complexity $O(m/T^{1-\delta})$ for constant $\delta>0$ for general $T=\operatorname{poly}(m)$. This rules out single-pass algorithms with space complexities as achieved by our 3-pass algorithm.
We note that for single-pass algorithms for counting four-cycles there remains a gap between the upper and lower bounds.

In the more powerful adjacency list streaming model, where edges are
grouped by endpoint and appear twice, four-cycle counting has been studied with
different bounds \cite{kallaugher2019complexity,mcgregor2020triangle}.
\citet{kallaugher2019complexity} provided an $\tildeO(m/T^{3/8})$-space algorithm
and a lower bound of $\Omega(m/T^{2/3})$. \citet{mcgregor2020triangle}
improved this to $\tildeO(m/\sqrt{T})$ space. The optimality in this model remains open.

Further, we note that the previously mentioned multi-pass algorithm by \citet{mcgregor2020triangle}
for counting four-cycles in the adjacency list streaming model looks superficially similar to ours.
They also use $\log(T)$~pairs of sampling probabilities to count the number of $K_{2,\kappa}$-subgraphs (which we will later call onions) for specific sizes~$\kappa$. However, this setting is significantly simpler than ours, since in this more powerful model node heaviness can be checked easily for \emph{all} nodes upon receiving a node together with all of its incident edges during the second stream pass.
Additionally, while for this algorithm it is enough to just consider node heaviness, here we must take into account the heaviness of \emph{all} substructures, which poses considerable challenges.

For sketching algorithms in the streaming setting, \citet{kallaugher2018sketching} showed that even for bounded-degree graphs, any algorithm for counting four-cycles requires $\tilde{\Omega}(m/\sqrt{T})$ space. They also presented an upper bound of $\tildeO(m/\sqrt{T})$ for graphs of maximum degree $\tildeO(T^{1/4})$. Their analysis required the degree bound to ensure that the variance of their algorithm is low, but we note that this bound does not hold in any of the hard instances we discuss below (see \Cref{sec:challenges-node-sampling}) and thus their algorithm does not apply in our setting. 
Interestingly, their algorithm also samples nodes, but it samples all nodes with the same probability, which is substantially different from our algorithm.

\citet{chen2022triangle-predictions} studied algorithms with
predictions and considered four-cycle counting when they have access to an
oracle that estimates the heaviness of a given substructure; this is motivated
by the integration of machine learning techniques into streaming models.
Given such oracle access, they presented a single-pass algorithm with space usage $\tildeO(T^{1/3}+m/T^{1/3})$.

\citet{luederssen2026four} presented a two-pass algorithm for counting four-cycles in low degeneracy graphs. Their edge-sampling based algorithm uses $O(m\kappa/\sqrt{T})$ space on graphs of degeneracy $\kappa$ and additionally performs well in practice.
For bipartite graphs, \citet{sanei-mehri2019Fleet}, \citet{li2022Approximately} and \citet{papadias2024counting}
presented practical algorithms for four-cycle counting. They evaluated their algorithms empirically and showed that, in practice, even single-pass streaming algorithms can achieve highly accurate predictions.

We also note that \citet{kallaugher2017hybrid} proposed a single-pass node-sampling algorithm for triangle counting, which, similar to our approach, samples nodes with different probabilities determined by their heaviness values. However, their assignment is performed only for nodes, which suffices for single-pass triangle counting. In contrast, our algorithm requires analogous assignments for edges and wedges, which must be carefully combined to effectively control the variance of our estimator. Moreover, unlike the single-pass setting, we employ multiple passes to collect and identify the relevant four-cycles.

Finally, we note that our sampling approach differs significantly from previous
node- or edge-sampling techniques that were used in triangle counting
\cite{pagh2012colorful,jayaram2021optimal}, offering a new
path for managing variance and enabling accurate estimation with sublinear
space. \smallskip

\paragraph{Outline of our paper.}
The paper is structured as follows. \cref{sec:prelim} introduces basic definitions. In \cref{sec:overview} we present a high-level description of our algorithm and provide intuition on how we designed it. \cref{sec:strategy} contains a technical overview of the analysis. \cref{sec:detection,sec:counting,sec:oracles} contain the main proof of \cref{thm:counting}: First we solve the \emph{detection} problem in \cref{sec:detection}, then we extend the detection algorithm to solve the \emph{counting} problem in \cref{sec:counting}. Our counting algorithm calls several oracles as subroutines, which are presented in \cref{sec:oracles}.

\section{Preliminaries}
\label{sec:prelim}
Throughout the paper, we consider undirected unweighted graphs $G=(V,E)$ with $n$~nodes and $m$~edges. 
For notational convenience, we assume that $V=\{1,\dots,n\}$ such that for $u,v\in V$ we have a natural ordering on the nodes given by $<$.
Further, we assume that we obtain the edge set~$E$ as an arbitrary order stream.

We let $T$ denote the number of four-cycles in $G$. In the \emph{detection} version of our problem, our goal is to distinguish whether $G$ has $0$~four-cycles or at least $T$~four-cycles. In the \emph{counting} version, our goal is to output a value $\tilde{T}$ such that $(1-\varepsilon)T \leq \tilde{T} \leq (1+\varepsilon)T$ for given $\varepsilon>0$.
As is standard in this line of work, we assume that the algorithm is given an $\alpha$-approximate lower bound of $T$ as input for some constant $\alpha\geq1$, and that an upper bound on $\alpha$ is known in advance. This assumption is necessary for setting the sampling probabilities (see \cite[Section 1]{braverman2013hard} and \cite[Section 1.2]{mcgregor2016better} for extensive discussions on this assumption). 
In the counting problem, we assume (without loss of generality) that $\eps = 2^{-k}$ for some integer $k$.

For two vertex sets $A,B\subseteq V$, we let $E[A,B] = \{ (u,v) \in E \colon u\in A, v\in B\}$ denote the edges with one endpoint in $A$ and one in $B$. 
Further, we set $E[A]=E[A,A]$ to denote edges induced by $A$.

We denote the set of all four-cycles by $\squares$.
For a four-cycle $A\in\squares$, $V(A)$ denotes its set of nodes, $E(A)$ its set of edges and $W(A)$ its set of wedges (a wedge is a path of two edges). If $A$ consists of edges $(a,b)$, $(b,c)$, $(c,d)$ and $(d,a)$, then we say that $a$ and $c$, as well as $b$ and $d$ are \emph{opposite nodes}, and for notational convenience we will simply refer to $A$ as $(a,b,c,d)$. Note that with the same node set, we could also form the four-cycle $(a,c,b,d)$, so the order matters in the tuple notation.

Given a four-cycle~$A$, we will often write $\substr$ to denote a \emph{substructure} of $A$, where a substructure is either a node, an edge or a wedge. 
The \emph{heaviness} $t(\substr)$ of a substructure~$\substr$ is the number of four-cycles containing $\substr$. For instance, 
for any vertex $v$, $t(v)$ denotes the number of four-cycles containing $v$. 
Further, we say that a substructure $\substr$ is \emph{$\theta$-heavy} if $t(\substr)\geq\theta$. For instance, we say that $v$ is a $\theta$-heavy node or that $e$ is a $\theta$-heavy edge.

We call the complete bipartite graph $K_{2,\kappa}$ an \emph{onion\footnote{\citet{mcgregor2020triangle} and \citet{vorotnikova2020improved} called this graph a \emph{diamond}.} of width $\kappa$}.
Furthermore, for a node pair $(u,v)$ we define its \emph{onion width $\ow(u,v)$} as the number of wedges with endpoints~$u$ and~$v$, and its \emph{onion size} $\os(u,v)=\binom{\ow(u,v)}{2}$ as the number of four-cycles containing~$u$ and~$v$ as opposite nodes.
Note that each four-cycle $(a,b,c,d)$ is part of two onions (namely those with opposite side nodes $a$ and $c$, and $b$ and $d$).

\section{Sampling Paradigms and Algorithmic Overview}
\label{sec:overview}
In this section, we describe our algorithm and present the main challenges and intuition for designing it.

\subsection{Previous Approaches Based on Edge Sampling}
We start by briefly discussing the previous algorithms by
\citet{mcgregor2020triangle} and by \citet{vorotnikova2020improved}, which were heavily based on edge sampling. In a nutshell, their algorithms works as follows: In the first stream pass, sample each edge with probability $p =\frac{1}{T^{1/3}}$ to obtain a set of edges~$S$. Then perform a second pass over the stream and, for each arriving edge~$e=(u,v)$, count how many 3-paths in~$S$ start at $u$ and end at $v$, and thus complete a four-cycle with~$e$. At the end, return the number of such four-cycles after rescaling appropriately.

While this is a very natural algorithm, it is inherently limited to using space $\Omega(m/T^{1/3})$: When we sample edges uniformly at random from the stream with probability~$p$, we must pick $p$ such that for at least $\Omega(1)$~four-cycles we sample a 3-path of their edges. For each four-cycle, this event happens with probability $4p^3$. Thus, to sample $\Omega(1)$ four-cycles in expectation, we need to set $p\geq \frac{1}{T^{1/3}}$ and thus the sampled edge set~$S$ contains $\Omega(m/T^{1/3})$~edges in expectation.

The main work of the papers by \citet{mcgregor2020triangle} and by \citet{vorotnikova2020improved} then goes into the analysis of their oracles, which they need to classify edges and wedges as heavy or light. This is done to ensure that their estimator's variance does not get too large.

As mentioned above, this algorithm is very natural and this is perhaps the reason why \citet{vorotnikova2020improved} conjectured that space $O(m/T^{1/3})$ is optimal for counting four-cycles. Indeed, observe that it is not at all obvious how this algorithm should be improved using edge sampling.

\subsection{Challenges for Node Sampling}
\label{sec:challenges-node-sampling}
As a consequence of the discussion above, our novel algorithms are almost entirely based on node sampling (we only use edge sampling for some of our heaviness oracles later). Specifically, we show that by storing an induced subgraph of size $\tildeO(m/\sqrt{T})$, we can successfully detect $\Omega(1)$ four-cycles with high probability (each of which is sampled with probability $\tildeO(\frac{1}{T})$), and this will also allow us to estimate the number of four-cycles.

To build some intuition for our node sampling approach, we start by discussing some basic sampling approaches and concrete hard instances that break them. For each instance, we will briefly argue that it contains $T$~four-cycles and give a bound on the desired space usage~$O(m/\sqrt{T})$.
We give an illustration of the instances in \Cref{fig:hardinstances}.

\begin{figure}[t]
    \centering
    \begin{subfigure}[t]{0.25\textwidth}
\centering
    \begin{tikzpicture}[every node/.style={circle, draw}]
        \node[circle, draw] (a) at (2,0) {};
        \node[circle, draw] (b) at (-2,0) {};
        \node[circle, draw] (c) at (0,1) {};
        \node[circle, draw] (d) at (0,1.5) {};
        \node[circle, draw] (e) at (0,-1.5) {};
        \node[circle, draw] (f) at (0,-1) {};
        \node[circle, draw] (g) at (0,0.5) {};
        \node[circle, draw] (h) at (0,-0.5) {};
        \node[circle, draw] (i) at (0,0) {};
        \draw (a) -- (c) -- (b);
        \draw (a) -- (e) -- (b) -- (f) -- (a) -- (g) -- (b) -- (h) -- (a) -- (d) -- (b) -- (i) -- (a);
    \end{tikzpicture}
    \caption{Onion $K_{2,\kappa}$.}
    \label{fig:onion}  
\end{subfigure}%
\hfill
\begin{subfigure}[t]{0.25\textwidth}
\centering
\begin{tikzpicture}[every node/.style={circle, draw}]
    
    \foreach \i in {1,...,4} {
        \node (a\i) at (0, {0.05 + (\i - 1) * 0.9}) {};
    }

    \foreach \j in {1,...,7} {
        \node (b\j) at (2, {0 + (\j - 1) * 0.5}) {};
    }

    \foreach \i in {1,...,4}
        \foreach \j in {1,...,7}
            \draw (a\i) -- (b\j);

\end{tikzpicture}
    \caption{Multiple overlapping onions $K_{\sqrt{T}/\kappa,\kappa}$.}
    \label{fig:complete}  
\end{subfigure}%
\hfill
\begin{subfigure}[t]{0.25\textwidth}
\centering
    \begin{tikzpicture}[]
        \node[circle, draw] (a) at (0,0) {};
        \node[circle, draw] (b) at (0,1.5) {};
        \node[circle, draw] (c) at (1,1.25) {};
        \node[circle, draw] (d) at (1,0.25) {};
        \node[circle, draw] (e) at (1.5,1.5) {};
        \node[circle, draw] (f) at (1.5,0) {};
        \node[circle, draw] (g) at (2.5,-0.5) {};
        \node[circle, draw] (h) at (2.5,2) {};
        \draw[line width=2] (a) -- (b);
        \draw (b) -- (c) -- (d) -- (a) -- (f) -- (e) -- (b) -- (h) -- (g) -- (a);
        \node[] (x) at (2,0.75) {$\dots$};
    \end{tikzpicture}
    \caption{Heavy edge instance.}
    \label{fig:heavyedge}
\end{subfigure}
    \caption{Illustrations of the three hard instances we consider in \Cref{sec:challenges-node-sampling}.}
    \label{fig:hardinstances}
\end{figure}
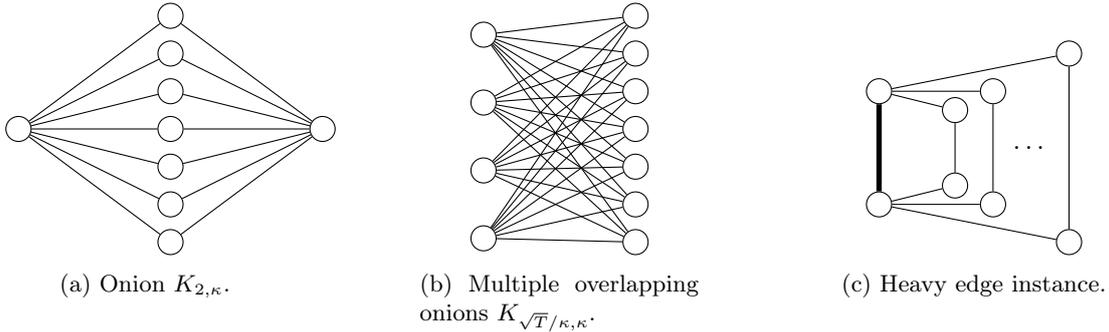

Perhaps the easiest way to employ node sampling would be to sample a set of nodes~$S$ which contains each node uniformly at random with probability~$p$. Then, during a stream pass, we store all edges of the induced subgraph~$E[S]$, and after that pass finishes, we count the four-cycles in $E[S]$.
\smallskip

\emph{Hard instance 1: A single onion.} 
Unfortunately, the simple sampling approach above fails horribly if our input graph is given by a single onion of width~$\kappa$ (i.e., $G=K_{2,\kappa}$). Note that if we set $\kappa=O(\sqrt{T})$ then this induces $T=\binom{\kappa}{2}$~four-cycles (since the two nodes on the one side of the bipartite graph form a four-cycle with all $\binom{\kappa}{2}$~pairs of vertices on the other side). Further, as this graph has $m=O(\kappa)$ edges, we can only use space $O(m/\sqrt{T}) = O(1)$.

Now observe that the simple algorithm above only samples the two special nodes with probability~$p^2$. Thus, to ensure that we even detect a four-cycle, we would need to set $p = \Omega(1)$, resulting in space usage~$O(m)$, which is too much.
\smallskip

\emph{Hard instance 2: Multiple overlapping onions.} Next, let us generalize the instance from above to consist of multiple overlapping onions. Concretely, consider the complete bipartite graph $K_{\sqrt{T}/\kappa,\kappa}$. Note that this graph has $\Theta(T)$ four-cycles since for each of the $\binom{\sqrt{T}/\kappa}{2}$ node pairs on the left side of the graph we obtain $\binom{\kappa}{2}$ four-cycles by picking two neighbors on the right side of the graph. Thus, the total number of four-cycles is given by $\binom{\sqrt{T}/\kappa}{2} \cdot \binom{\kappa}{2} = \Theta(T)$. Further note that this graph has $m=(\sqrt{T}/\kappa)\cdot \kappa=\sqrt{T}$ edges and thus an algorithm achieving our desired space bound can only use space $O(m/\sqrt{T}) = O(1)$. Finally, observe that this graph can be viewed as $\binom{\sqrt{T}/\kappa}{2}$~overlapping onions of width $\kappa$ by picking pairs of nodes on the left-hand side of the graph and keeping all right-hand side vertices.

Now, as before, the simple algorithm of sampling a node set~$S$ and storing $E[S]$ cannot succeed in achieving our desired space usage. 
However, suppose that we sample each vertex from the left-hand side with probability $p_1 = \frac{1}{\sqrt{T}/\kappa}$ to obtain a node set~$S_1$ and each vertex from the right-hand side with probability $p_2 = \frac{1}{\kappa}$ to obtain a node set~$S_2$. Now, during the stream pass, we only store the edges in $E[S_1,S_2]$. Crucially, observe that we sample only $\Theta(1)$ nodes in total. This ensures that we sample a four-cycle with constant probability, and further, we only need space~$O(1)$, matching our desired space bound.

Of course, this approach has two obvious drawbacks: First, it requires knowledge about the two sides of the bipartite graph. Second, it also requires knowledge about the parameter~$\kappa$ to define our probabilities. However, we will see later that we can overcome these obstacles.
\smallskip

\emph{Hard instance 3: Heavy edge instances.}
While above we sketched an approach if our input graph contains many large overlapping onions, we now consider the other extreme. That is, 
we consider instances with heavy edges, i.e., edges that appear in many four-cycles. Concretely, consider a graph with a single edge $e^*=(u^*,v^*)$ that is contained in all $T$~four-cycles. Now, if we still want to store an induced subgraph $E[S_1,S_2]$ as in the previous example, then we have to ensure that (say) $S_1$ contains $u^*$ and $S_2$ contains $v^*$ with high probability, as otherwise we never store $e^*$ and thus have no opportunity to detect any four-cycle. However, as we do not know $u^*$ and $v^*$ beforehand, this means that $S_1$ and $S_2$ must include \emph{all} vertices, which is, of course, infeasible.

Thus, in such a case, we adapt our algorithm from above as follows. We consider an algorithm which samples two node sets $R_1$ and $R_2$, where we set $p_1=1$ and $p_2=\frac{1}{\sqrt{T}}$. Then, during a first stream pass, we store all edges in $E[R_1,R_2] \cup E[R_2]$.
After that, during a second stream pass, we check for each stream edge~$e=(u,v)$ whether it forms a four-cycle with a 3-path of the form $u$--$a$--$b$--$v$, where $u,v\in R_1$, and $a,b\in R_2$. 

Coming back to the concrete example above, observe that $u^*$ and $v^*$ are indeed in $R_1$ since $R_1$ contains all vertices, and $b$ and $c$ are in $R_2$ with probability~$p_2^2 = \frac{1}{T}$. Thus, as $e^*=(u^*,v^*)$ is contained in $T$~four-cycles, in expectation we find $\Omega(1)$ four-cycles containing $e^*$. Below, we will generalize this algorithm for more general probabilities $p_1$ and $p_2$ and we will show that it only uses space $O(m/\sqrt{T})$.

Note that this approach is somewhat similar to those of \citet{mcgregor2020triangle} and \citet{vorotnikova2020improved}, but here we only need to sample two nodes ($a$ and $b$ above) to complete the four-cycle with $e^*$, compared to their three edges. This is crucial to obtain space $\tildeO(m/\sqrt{T})$.

\subsection{Our Algorithm Based on Node Sampling} \label{sec:techoveralgorithm}
Next, we describe our detection algorithm and state its pseudocode in \Cref{algo:detect} below.
Our algorithm combines the insights obtained from the hard instances above and uses a set of parameters $\CI=\{2^k T^{\frac{1}{4}}\;|\; k=0,\dots,\lceil\log(T^{\frac{1}{4}})\rceil)\}$. 
For all $\kappa\in\CI$ it defines probabilities 
\begin{align*}
    p_{1,\kappa}=\frac{c_1\kappa}{\delta^{1.5}\sqrt{T}}
    \quad
    \text{ and }
    \quad
    p_{2,\kappa}=\frac{c_1}{\delta^{1.5} \kappa},
\end{align*}
where $\delta=\frac{1}{2098\log(T)}$ and $c_1$ is some constant to be set later. Note that both of these probabilities crucially depend on the parameter~$\kappa$. 

Next, before the first stream pass, for all $\kappa\in\CI$, the algorithm samples vertex sets $S_{1,\kappa}$, $R_{1a,\kappa}$ and $R_{1b,\kappa}$
by sampling each node independently and uniformly at random with probability $p_{1,\kappa}$, and it further samples vertex sets $S_{2,\kappa}$, $R_{2a,\kappa}$ and $R_{2b,\kappa}$ by sampling each vertex independently and uniformly at random with probability\footnote{
We note that we do not necessarily have to store these sets explicitly. It is enough if we use a hash function $h \colon V \to [0,1]$ and we say that a vertex~$v$ is sampled with probability~$p$ if $h(v)\leq p$.}~$p_{2,\kappa}$.

In a first pass over the stream, the algorithm collects all edges in $E[S_{1,\kappa},S_{2,\kappa}]$, as well as all edges in $E[R_{1b,\kappa},R_{2a,\kappa}]$, $E[R_{2a,\kappa},R_{2b,\kappa}]$ and $E[R_{2b,\kappa},R_{1a,\kappa}]$ for all $\kappa\in\CI$.

In a second pass over the stream, upon arrival of an edge $e=(u,v)$, the algorithm checks whether there exists a four-cycle $A=(u,v,a,b)$ such that $u\in R_{1a,\kappa}$, $v \in R_{1b,\kappa}$, $a\in R_{2a,\kappa}$ and $b\in R_{2b,\kappa}$ for any $w\in\CI$. Note that we can check whether this four-cycle exists since we currently see $(u,v)$ in the stream and the edges $(v,a)$, $(a,b)$ and $(b,u)$ were collected during the first stream pass since they are contained in the sets $E[R_{1b,\kappa},R_{2a,\kappa}]$, $E[R_{2a,\kappa}, R_{2b,\kappa}]$ and $E[R_{2b,\kappa}, R_{1a,\kappa}]$, respectively. If the algorithm identifies such a four-cycle~$A$, it stores the edge $(u,v)$.

After both stream passes finished, the algorithm returns whether the set of all stored edges contains a four-cycle. 
We note that our algorithm remains correct even if the edge order changes in different stream passes.
\smallskip

\paragraph{Intuition for the algorithm and its analysis.}
To understand the intuition behind the algorithm, note that it combines and generalizes the two sampling approaches we discussed in \Cref{sec:challenges-node-sampling} in the second and third hard instances.
In a nutshell, we use the sets $S_{1,\kappa}$ and $S_{2,\kappa}$ to detect four-cycles which are part of large, potentially overlapping onions (corresponding to the second hard instance) and we use the other sets $R_{1a,\kappa}$, $R_{1b,\kappa}$, $R_{2a,\kappa}$ and $R_{2b,\kappa}$ to detect four-cycles which contain heavy edges (corresponding to the third hard instance).

The reason why we use different sets $S_{\cdot,\kappa}$ and $R_{\cdot,\kappa}$ is that we use these sets to detect (and later count) different types of four-cycle, and we want these estimators to be independent from each other.
We note that with the parameter choices above, we always have that $p_{1,\kappa}\geq p_{2,\kappa}$ and thus it is instructive to think of the sets that were sampled with probability $p_{1,\kappa}$ as substantially larger than those that were sampled with probability $p_{2,\kappa}$.

Now let us briefly argue why the algorithm achieves its guarantees \emph{in expectation}.

First, consider any four-cycle $A=(a,b,c,d)$. 
Note that the algorithm detects $A$ if $a,c\in S_{1,\kappa}$ and $b,d\in S_{2,\kappa}$ since in this case all edges of $A$ are in $\Expectation{S_{1,\kappa},S_{2,\kappa}}$. 
Thus, $A$ is detected with probability at least $p_{1,\kappa}^2 p_{2,\kappa}^2 = \left(\frac{c_1^2 \kappa}{\delta^3 \sqrt{T} \kappa}\right)^2 = \tildeO(\frac{1}{T})$. 
Hence, after summing over all $T$~four-cycles, we obtain that the algorithm detects $\Omega(1)$ four-cycles in expectation.
Observe that for this simple expectation analysis, we do not even require the sets $R_{\cdot,\kappa}$ or even different $\kappa$-values, but they will be crucial later to obtain our guarantees with high probability.

Second, let us bound the space usage. Note 
that the sets $E[S_{1,\kappa}, S_{2,\kappa}]$, $E[R_{1b,\kappa}, R_{2a,\kappa}]$ and $E[R_{2b,\kappa}, R_{1a,\kappa}]$ only contain edges for which one of its endpoints was sampled with probability~$p_{1,\kappa}$ and the other one was sampled with probability~$p_{2,\kappa}$. Thus, each edge is contained in this set with probability $2 p_{1,\kappa} p_{2,\kappa} = \tildeO(\frac{1}{\sqrt{T}})$, and therefore we obtain an 
expected space usage of $\tildeO(m/\sqrt{T})$.
Similarly, the set  $E[R_{2a,\kappa},R_{2b,\kappa}]$ contains each edge with probability $p_{2,\kappa}^2 = \tildeO(\frac{1}{\kappa^2}) \in \tildeO(\frac{1}{T^{1/2}})$ since $\kappa\in[T^{1/4},2T^{1/2}]$. Thus, this set also contains at most $\tildeO(m/\sqrt{T})$~edges in expectation.

This shows that the algorithm achieves its guarantees for detection in expectation. We sketch the extension to a four-cycle counting algorithm in \cref{sec:counting_sketch}.

\section{Proof Strategy and Technical Overview} 
\label{sec:strategy}
In this section, we give an overview of our technical contributions and outline several proof ideas.

\subsection{Detecting Four-Cycles Using Space \texorpdfstring{$\tilde{O}(m/\sqrt{T})$}{Õ(m/√T)} With High Probability}

While above we already argued that we find $\Omega(1)$ four-cycles in expectation, showing that this holds with high probability is substantially more difficult.

To bound the probability that our algorithm detects a four-cycle, we have to bound the variance of the random variable counting the number of four-cycles in the sample. As in previous works, the variance of this variable is high if a substructure appearing in many four-cycles is sampled with small probability, where a \emph{substructure} is either a node, an edge or a wedge. 
Thus, we would like to exclude all four-cycles containing substructures that appear in ``too many'' four-cycles.

Now, one of our key technical challenges for excluding four-cycles is the following. Given the algorithm above, each four-cycle~$A$ can be sampled in $\Theta(\log T)$ different ways: $A$ can be sampled for $\Theta(\log T)$~values of $\kappa$ and there are $\Theta(1)$~many different ways to assign the nodes of $A$ to the different sets we sampled. For our analysis, we will pick exactly one such way per four-cycle.

To select how a four-cycle should be sampled, we will intuitively consider two types of four-cycles, which correspond to the two sampling approaches that we are using simultaneously. 
Here, we use the notion of \emph{heaviness} $t(\substr)$ of a substructure~$\substr$, which counts the number of four-cycles containing $\substr$, that we defined in the preliminaries.
These two types of four-cycles are as follows:
\begin{enumerate}
    \item Consider a four-cycle~$A=(a,b,c,d)$ in which all edges~$e$ have heaviness $t(e)\leq \sqrt{T}$. Then we would like to sample $A$ in such a way that $a,c\in S_{1,\kappa}$ and $b,d\in S_{2,\kappa}$, as in this case we collected all edges of $A$ in $E[S_{1,\kappa},S_{2,\kappa}]$.
    This case corresponds to the second hard instance in \Cref{sec:challenges-node-sampling}.
    Crucially, observe that here the nodes $a,c\in S_{1,\kappa}$ are \emph{opposite} nodes in $A$.
    \item Next, consider a four-cycle~$A=(a,b,c,d)$ containing exactly one edge~$e=(a,b)$ with heaviness $t(e)\geq \sqrt{T}$.
    Then we want to detect~$A$ if $a\in R_{1a,\kappa}$, $b\in R_{1b,\kappa}$, $c\in R_{2a,\kappa}$ and $d\in R_{2b,\kappa}$.
    This corresponds to the third hard instance in \Cref{sec:challenges-node-sampling}.
    Crucially, observe that the nodes $a\in R_{1a,\kappa}$ and $b\in R_{1b,\kappa}$ are \emph{adjacent} nodes in~$A$.
\end{enumerate}

Note that above we did not consider the case of a four-cycle containing two or more heavy edges. This is on purpose since there are only a few such four-cycles and we can safely ignore them (as we will see below).

Next, given the two types of four-cycles and the sampling ``rules'' that we sketched, observe that they have the following in common: We want to sample both four-cycle types in such a way that they contain two nodes that were sampled with probability~$p_{1,\kappa}$ and two nodes that were sampled with probability~$p_{2,\kappa}$. 
The main difference between the two cases is whether the two nodes that were sampled with probability~$p_{1,\kappa}$ are opposite nodes (the first case) or if they are adjacent (the second case).
This will be the basis for our definition of \emph{configurations} below.
\bigskip

\paragraph{Configurations, labeled substructures, heaviness thresholds.} 
To formally argue how a four-cycle should be sampled, we introduce configurations. Concretely, given a four-cycle $A\in\squares$ and $\kappa\in \CI$, a \emph{configuration} is a tuple $(A,\kappa,x,y)\in \squares\times \CI\times V^2$ where $x,y\in V(A)$ and $x<y$. 

Note that for each four-cycle in $A$ there are $|\CI| \times \binom{4}{2} = \Theta(\log T)$ configurations. 
A key for our analysis is to \emph{assign} at most one configuration to each four-cycle. In this assignment step, the nodes $x$ and $y$ are picked based on the two sampling strategies above, where we either have $x,y\in S_{1,\kappa}$ or $x\in R_{1a,\kappa}$ and $y\in R_{1b,\kappa}$;
see \Cref{fig:configurations} for an illustration. 
Thus, a configuration intuitively encodes that the vertices of~$A$ are sampled in the correct sets by our algorithm.
Further, the value $\kappa$ has to be picked carefully based on the heaviness of the four-cycle's substructures (see below).

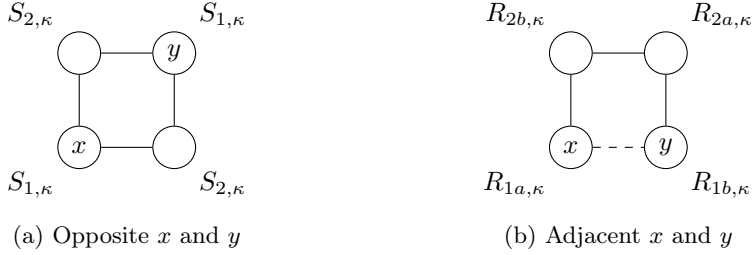
\begin{figure}
    \centering
    \tikzset{customnode/.style={circle, draw, minimum size=16pt, inner sep=0pt}}

\hspace{1cm}
\begin{subfigure}[t]{0.4\textwidth}
\centering
    \begin{tikzpicture}[scale=1.25]
        \node[customnode, label=below left:{$S_{1,\kappa}$}] (a) at (0,0) {$x$};
        \node[customnode, label=above left:{$S_{2,\kappa}$}] (b) at (0,1) {};
        \node[customnode, label=above right:{$S_{1,\kappa}$}] (c) at (1,1) {$y$};
        \node[customnode, label=below right:{$S_{2,\kappa}$}] (d) at (1,0) {};
        
        \draw (a) -- (b) -- (c) -- (d) -- (a);
    \end{tikzpicture}
    \caption{Opposite $x$ and $y$}
    \label{fig:opposite}
\end{subfigure}
\hfill
\begin{subfigure}[t]{0.4\textwidth}
\centering
    \begin{tikzpicture}[scale=1.25]
        \node[customnode, label=below left:{$R_{1a,\kappa}$}] (a) at (0,0) {$x$};
        \node[customnode, label=below right:{$R_{1b,\kappa}$}] (b) at (1,0) {$y$};
        \node[customnode, label=above right:{$R_{2a,\kappa}$}] (c) at (1,1) {};
        \node[customnode, label=above left:{$R_{2b,\kappa}$}] (d) at (0,1) {};
        
        \draw (b) -- (c) -- (d) -- (a);
        \draw[dashed] (a) -- (b);
    \end{tikzpicture}
    \caption{Adjacent $x$ and $y$}
    \label{fig:adjacent}
\end{subfigure}
\hspace{1cm}
    \caption{ Illustration of our assigned configurations $(A,\kappa,x,y)$.
    If $x$ and $y$ are opposite nodes in a four-cycle~$A$, we will only consider $A$ if $x,y\in S_{1,\kappa}$ and its other two nodes are in $S_{2,\kappa}$ \emph{and} it does not contain an edge~$e$ of heaviness $t(e)>\sqrt{T}$. Otherwise, if $x$ and $y$ are adjacent, we only consider $A$ if $x\in R_{1a,\kappa}$ and $y\in R_{1b,\kappa}$ \emph{and} it has exactly one edge~$e$ of heaviness $t(e)\geq \sqrt{T}$.
    The dashed line represents the heavy edge~$e$ and we note that our algorithm collects it in the second pass.
    }
    \label{fig:configurations}
\end{figure}

Given that our configurations encode how the vertices of a four-cycle should be sampled into different sets, we also have to take this into account when controlling the correlation between different configurations. Thus, we define \emph{labeled substructures $(\substr,\kappa,\ell)$} consisting of a substructure~$\substr$, a value $\kappa\in\CI$ and a labeling function~$\ell$ indicating the sets into which each of the vertices in $\substr$ may be sampled (see \Cref{def:label_substructure} for the formal definition). 
For instance, given an edge $e=(u,v)$ it will make a difference whether $u\in S_{1,\kappa}$ and $v\in S_{2,\kappa}$, or whether $u,v\in S_{1,\kappa}$, since in the first case $e$ was part of our sample with probability $p_{1,\kappa}p_{2,\kappa}$ and in the latter case with probability $p_{1,\kappa}^2$; the labelings allow us to distinguish such cases.
Given this notation, we observe that two configurations are dependent/correlated \emph{iff} the labels and $\kappa$-values of one of their substructures agree.

Furthermore, for each labeled substructure~$(\substr,\kappa,\ell)$ we define a \emph{heaviness threshold~$\theta_{(\substr,\kappa,\ell)}$}, which depends on the probability that $\substr$ is sampled with the given labeling~$\ell$. Formally, we will set $\theta_{(\substr,\kappa,\ell)} = \tildeO(T p_{(\substr,\kappa,\ell)})$ where $p_{(\substr,\kappa,\ell)}$ is the probability that $\substr$ is sampled for probabilities $p_{1,\kappa}$ and $p_{2,\kappa}$ with the labeling given by the labeling~$\ell$.
See \Cref{fig:heavinessthresholds} for an illustration.

We define a labeled substructure~$(\substr,\kappa,\ell)$ as \emph{heavy} if $t(\substr) > \theta_{(\substr,\kappa,\ell)}$, i.e., if it appears in more four-cycles than allowed by its heaviness threshold. A configuration is deemed \emph{heavy} if any of its labeled substructures is heavy; otherwise, it is \emph{light}. Furthermore, a four-cycle is called \emph{heavy} if all of its configurations are heavy; otherwise, it is \emph{light}.

We establish two key properties: 
\begin{enumerate}[label=(\alph*)]
    \item We \emph{assign} a light configuration to each light four-cycle, and let $\CCdetect$ denote the set of these assigned configurations. We show that configurations in $\CCdetect$ have low variance.
    \item The number of light four-cycles is at least $|\CCdetect| \geq (1 - \Theta(\delta \log T))T.$
\end{enumerate}

Above, Property~(a) is achieved by \emph{assigning a unique configuration $(A,\kappa_A,x_A,y_A)$} to each light four-cycle~$A$. We intuitively pick $\kappa_A\in\CI$ as the smallest $\kappa\in\CI$ for which~$A$ has a light configuration (technically, we will use a slightly different condition, see \Cref{eq:w_A}). Further, we choose vertices $x_A$ and $y_A$ based on the sampling rules described above (see also \Cref{fig:configurations}).
Then the lightness of the substructures and the choice of the heaviness thresholds allow us to establish the low variance, as sketched below.
Moreover, Property~(b) follows from a technical \emph{combination lemma} (\Cref{lemma:combination}), which shows that few four-cycles contain multiple heavy substructures.

Our analysis of the detection algorithm then focuses on identifying whether any configuration in $\mathcal{C}_\mathrm{detect}$ appears in the sampled subgraphs. Since 
$|\CCdetect| \geq (1 - \Theta(\delta \log T))T$, the expected number of such detected configurations is 
$\Omega(1)$ for sufficiently small $\delta$.
Moreover, because all configurations in $\CCdetect$ are light, their dependencies are weak, ensuring low variance. This guarantees that, with high probability, the algorithm detects a light configuration --- and thus a four-cycle.

We highlight that the introduction of configurations, labeled substructures, and sampling-dependent heaviness thresholds forms a key conceptual contribution of this work. In particular, we assign different thresholds to substructures based on how they are sampled.

\bigskip
\paragraph{Sketch of the variance analysis.}
Next, we provide a sketch of our variance analysis.
We aim to show that our algorithm detects a configuration from $\CCdetect$ with high probability. As pairs of configurations sharing labeled substructures are not independent, we need to carefully bound the covariance of each such pair. Let $X_c$ be the binary random variable indicating whether configuration $c$ is sampled.

For simplification, consider a single sampled configuration $c=(A,\kappa,x,y)\in\CCdetect$ and only one of its labeled substructures $(v,\kappa,\ell)$, where $v$ is a node sampled into $S_{1,\kappa}$.
Consider all four-cycles that share exactly the node~$v$ with $A$. Each such four-cycle corresponds to exactly one configuration~$b$ that contains $(v,\kappa,\ell)$ and we let $\CB$ denote the set of all such configurations~$b$.

Crucially, as $(v,\kappa,\ell)$ is light, there are at most $\theta_{(v,\kappa,\ell)} = \tildeO( Tp_{1,\kappa})$ such configurations $b$ and thus $|\CB|=\tildeO( Tp_{1,\kappa})$.
The probability that $b$ is sampled, given that $v$ is sampled, is $p_{1,\kappa}p_{2,\kappa}^2$. This allows us to bound the covariances as
\[\sum_{b\in \CB}\Cov{X_c,X_b}\leq\sum_{b\in \CB}\Expectation{X_b\mid X_c=1}\Prob{X_c=1}=\sum_{b\in \CB}p_{1,\kappa}p_{2,\kappa}^2\cdot p_{1,\kappa}^2p_{2,\kappa}^2 = \tildeO\left( T p_{1,\kappa}^4p_{2,\kappa}^4\right) =  \tildeO\left(\frac{1}{T}\right),\]
where in the penultimate step we used that $|\CB| = \tildeO( T p_{1,\kappa})$ by the argument above.

The same bound applies to all $O(1)$ substructures and all $\tildeO(T)$ configurations $c$, resulting in a total covariance summed over all configuration pairs of $\tildeO(1)$.
This is sufficient to show that a four-cycle will be detected with constant probability, which can be boosted to high probability, by running our algorithm $\log(n)$ times in parallel.

For our formal variance analysis, see \Cref{sec:variance}.

\subsection{Extension to Counting}
\label{sec:counting_sketch}
Next, we extend our detection algorithm to counting. Recall that we previously showed how to detect (i.e., find) $\Theta(1)$ four-cycles. To extend this to the counting setting, we use the following approach: for each of the $\Theta(1)$ detected four-cycles~$A$, check whether~$A$ was sampled with its assigned configuration~$(A,\kappa_A,x_A,y_A)$, and then scale the counts appropriately (see \Cref{algo:count} for the pseudocode of our approach). This enables us to filter out configurations that may lead to high variance while still obtaining a good estimate of the number of light four-cycles, which in turn serves as a good approximation of $T$.
However, efficiently verifying whether a sampled four-cycle conforms to its assigned configuration --- within $\tilde{O}(m/\sqrt{T})$ space --- presents us with significant technical challenges.
\smallskip

\emph{Challenge~1:}
Suppose a four-cycle~$A$ was sampled with a configuration~$c=(A,\kappa,x,y)$. To check whether $c=(A,\kappa_A,x_A,y_A)$, i.e., whether $A$ was sampled with its assigned configuration, we have to check (among others) that all labeled substructures~$\substr$ of~$c$ are light. Thus, we need to check $t(\substr) \leq \theta_{(\substr,w,\ell)}$ for all $\substr\in c$. However, we find that this cannot be done using space $\tildeO(m/\sqrt{T})$.

To resolve this issue, we introduce a notion of \emph{refined heaviness}~$t(\substr,\kappa,\ell)$ that replaces~$t(\substr)$. Here, recall that $t(\substr)$ counts all four-cycles containing the substructure~$\substr$, which is completely independent of our sampling algorithm. In contrast, $t(\substr,w,\ell)$ takes the sampling probabilities and labelings of substructures into account and ensures that we only count configurations towards $t(\substr,\kappa,\ell)$ in which important substructures are light. For instance, for labeled nodes~$(\substr,\kappa,\ell)$ we define:
 \begin{align}
 \label{eq:refined-heaviness-nodes}
 \begin{split}
    t(\substr,\kappa,\ell)=|\{c\in \FCvalid\mid (\substr,\kappa,\ell)\in c,\; t(\wedge,\kappa,\ell')&\leq s_1\theta_{(\wedge,\kappa,\ell')} \text{ for all wedges } (\wedge,\kappa,\ell')\in c,\\\quad t(e,\kappa,\ell')&\leq s_2\theta_{(e,\kappa,\ell')} \text{ for all edges } (e,\kappa,\ell')\in c\}|,
\end{split}
\end{align}
where $\FCvalid$ is the set of so-called valid configurations (defined in \Cref{sec:counting}) and $s_1$ and $s_2$ are some random shifts required for technical reasons.

Now we define labeled substructures as \emph{heavy} if $t(\substr,\kappa,\ell) > s \theta_{(\substr,\kappa,\ell)}$, where the heaviness threshold $\theta_{(\substr,\kappa,\ell)}$ is the same as before and again $s$ is some random shift for technical reasons. The main advantage of this definition is that it allows us
to check whether
$t(\substr,\kappa,\ell) \leq s\theta_{(\substr,\kappa,\ell)}$ using space just $\tildeO(m/\sqrt{T})$.
However, we note that now we might have that $t(\substr,\kappa,\ell)\ll t(\substr)$ since we only count four-cycles in the definition of $t(\substr,\kappa,\ell)$ with certain good properties. Despite this relaxation, we show that at most $O(\varepsilon T)$ four-cycles are lost and that the variance remains bounded. 
\smallskip

\emph{Challenge~2:}
Next, to check whether $c=(A,\kappa,x,y)$ satisfies $c=(A,\kappa_A,x_A,y_A)$, we also have to check that $\kappa = \kappa_A$. However, whether a substructure is heavy is no longer a monotone property with respect to the parameter~$\kappa$, causing significant challenges.

To understand this, let us briefly reconsider our detection algorithm. There, a substructure is heavy if $t(\substr)>\theta_{(\substr,\kappa,\ell)}$ (i.e., using our old heaviness notion $t(\substr)$). This is a very convenient definition because $t(\substr)$ is fixed and only $\theta_{(\substr,\kappa,\ell)}$ depends on $\kappa$. Since $\theta_{(\substr,\kappa,\ell)}$ is a monotone function in~$\kappa$, whether a substructure is heavy is a monotone property in~$\kappa$. Now, 
since we picked $\kappa_A$ as the smallest $\kappa\in\CI$ such that $A$ has a light configuration,
this suggests the following natural approach to check whether $\kappa=\kappa_A$: If $(A,\kappa,x,y)$ is light and $(A,\kappa/2,x,y)$ is heavy, then return that $\kappa=\kappa_A$, otherwise return $\kappa\neq\kappa_A$.

Unfortunately, for our new notion of refined heaviness $t(\substr,\kappa,\ell)$, this approach no longer works. Now, a substructure is heavy if $t(\substr,\kappa,\ell) > \theta_{(\substr,\kappa,\ell)}$ --- where both sides depend on~$\kappa$. Indeed, note that in \Cref{eq:refined-heaviness-nodes} the configurations that we are counting heavily depend on the choice of $\kappa$. Now, we can construct examples (see \Cref{fig:nonmonotone}) where the property $t(\substr,\kappa,\ell)>\theta_{(\substr,\kappa,\ell)}$ is no longer monotone in $\kappa$.
Thus, our simple approach for checking $\kappa=\kappa_A$ from before (for $t(\substr)$) no longer works. Further, we also cannot check this property for all $\kappa'\in\CI$ with $\kappa'<\kappa$ as this will require too much space if $\kappa'\ll\kappa$. 

We resolve these issues, by only considering configurations $c=(A,\kappa,x,y)$ which are \emph{light and locally minimial}, i.e., they satisfy that all $(A,\kappa',x,y)$ with $\varepsilon^2\kappa\leq \kappa' < \kappa$ are heavy and that $c$ is light.
This can be checked efficiently within our desired space $\tildeO(m/\sqrt{T})$ since we only have to check $\kappa'$-values in a small, local neighborhood around $\kappa$.
However, the problem is that with this definition some four-cycles may now have \emph{multiple} light and locally minimal configurations, i.e., we introduce some double-counting. 
Fortunately, we manage to show that this is only the case for at most $O(\varepsilon T)$~four-cycles, which still allows us to obtain our desired approximation guarantees.
\smallskip

\emph{Challenge~3:}
The third challenge arises from our refined heaviness definition~$t(\substr,\kappa,\ell)$.
To decide whether $t(\substr,\kappa,\ell)\leq \theta_{(\substr,\kappa,\ell)}$ we have to estimate $t(\substr,\kappa,\ell)$.  However, in the estimation, we have to ensure that we only count configurations~$c'$ such that certain labeled substructures in $c'$ are light. 
For instance, if $\substr$ is a node, then for each sampled configuration~$c'$ we must ensure that its wedges and edges are light (as per \Cref{eq:refined-heaviness-nodes}). The latter can only be done by again checking the heaviness/lightness of these substructures. Thus, we have to build multiple so-called \emph{oracles} for checking the heaviness/lightness of configurations, nodes, edges and wedges that call each other. For instance, the node oracle will call the wedge oracle and the edge oracle; the edge oracle itself will also call another wedge oracle. However, we make sure that we only perform $\tildeO(1)$~oracle calls in total, which is crucial to obtain our space bound.

\section{Four-Cycle Detection}
\label{sec:detection}
In this section, we present an algorithm for the detection version of our problem, i.e., the algorithm decides whether a graph has $0$~four-cycles or at least $T$~four-cycles. It performs two passes over the stream and uses space $\tildeO(m/\sqrt{T})$, and it is the basis for our counting algorithm below.
We summarize its guarantees in the following theorem.

\begin{theorem}
\label{thm:detect}
    There exists a 2-pass algorithm which obtains a graph~$G$ as an arbitrary-order edge stream and correctly decides whether $G$ has 0 or at least $T$ four-cycles using  $\tilde{O}(m/\sqrt{T})$ space with high probability.
\end{theorem}

In the following, we describe the detection algorithm in \Cref{sec:detection-algorithm-description}. In \Cref{sec:configurations}, we formalize the notions of configurations, their labelings, and labeled substructures. Next, in \Cref{sec:heaviness}, we define the concept of heaviness for (labeled) substructures, configurations, and four-cycles, and show how to assign each light four-cycle to a unique configuration. We establish a lower bound on the number of light four-cycles in \Cref{sec:bound-light-squares}, and analyze the variance of the random variable counting four-cycles in \Cref{sec:variance}. Finally, we complete the proof of \Cref{thm:detect} in \Cref{sec:proofofthmdetect}.

\subsection{Algorithm Description}
\label{sec:detection-algorithm-description}

We start by defining our detection algorithm and we summarize its pseudocode in
\Cref{algo:detect}. 

As we described in \Cref{sec:techoveralgorithm}, the algorithm uses a set of parameters $\CI=\{2^k T^{\frac{1}{4}}\;|\; k=0,\dots,\lceil\log(T^{\frac{1}{4}})\rceil)\}$, and 
iterates over all $\kappa\in\CI$. Then for each $\kappa$ it proceeds as follows.

Initially, the algorithm defines probabilities 
\begin{align*}
    p_{1,\kappa}=\frac{c_1\kappa}{\delta^{1.5}\sqrt{T}}
    \quad
    \text{ and }
    \quad
    p_{2,\kappa}=\frac{c_1}{\delta^{1.5} \kappa},
\end{align*}
where $\delta=\frac{1}{2098\log(T)}$ and $c_1$ is some constant to be set later. Note that both of these probabilities crucially depend on the parameter~$\kappa$.

Next, before the first stream pass, for all $\kappa\in\CI$, the algorithm samples vertex sets $S_{1,\kappa}$, $R_{1a,\kappa}$ and $R_{1b,\kappa}$
by sampling each node independently and uniformly at random with probability $p_{1,\kappa}$, and it further samples vertex sets $S_{2,\kappa}$, $R_{2a,\kappa}$ and $R_{2b,\kappa}$ by sampling each vertex independently and uniformly at random with probability~$p_{2,\kappa}$.

In a first pass over the stream, the algorithm collects all edges in $E[S_{1,\kappa},S_{2,\kappa}]$, as well as all edges in $E[R_{1b,\kappa},R_{2a,\kappa}]$, $E[R_{2a,\kappa},R_{2b,\kappa}]$ and $E[R_{2b,\kappa},R_{1a,\kappa}]$ for all $\kappa\in\CI$.

In a second pass over the stream, upon arrival of an edge $e=(u,v)$, the algorithm checks whether there exists a four-cycle $A=(u,v,a,b)$ such that $u\in R_{1a,\kappa}$, $v \in R_{1b,\kappa}$, $a\in R_{2a,\kappa}$ and $b\in R_{2b,\kappa}$ for any $w\in\CI$. Note that we can check whether this four-cycle exists since we currently see $(u,v)$ in the stream and the edges $(v,a)$, $(a,b)$ and $(b,u)$ were collected during the first stream pass since they are contained in the sets $E[R_{1b,\kappa},R_{2a,\kappa}]$, $E[R_{2a,\kappa}, R_{2b,\kappa}]$ and $E[R_{2b,\kappa}, R_{1a,\kappa}]$, respectively. If the algorithm identifies such a four-cycle~$A$, it stores the edge $(u,v)$.

After both stream passes finished, the algorithm returns whether the set of all stored edges contains a four-cycle.

\begin{algorithm}[t]
\caption{Node Sampling for Detection}
\label{algo:detect}
\begin{algorithmic}[1]
\State $\CI = \{ 2^k T^{1/4} \mid k = 0, \dots, \lceil \log(T^{1/4}) \rceil \}$
\For{$\kappa \in \CI$}
    \State $p_{1,\kappa} \gets \frac{c_1 \kappa}{\delta^{1.5} \sqrt{T}}$
    \State $p_{2,\kappa} \gets \frac{c_1}{\delta^{1.5} \kappa}$
    \State Sample $S_{1,\kappa}$, $R_{1a,\kappa}$, $R_{1b,\kappa} \subset V$ by sampling each node independently and u.a.r with probability $p_{1,\kappa}$
    \State Sample $S_{2,\kappa}$, $R_{2a,\kappa}$, $R_{2b,\kappa} \subset V$ by sampling each node independently and u.a.r with probability $p_{2,\kappa}$
\EndFor
\State \textbf{Pass 1:} Collect edges in 
$E[S_{1,\kappa}, S_{2,\kappa}] \cup E[R_{1b,\kappa}, R_{2a,\kappa}] \cup E[R_{2a,\kappa}, R_{2b,\kappa}] \cup E[R_{2b,\kappa}, R_{1a,\kappa}]$ 
for all $\kappa \in \CI$
\State \textbf{Pass 2:} Collect all edges $(u,v) \in E[R_{1a,\kappa}, R_{1b,\kappa}]$ completing a four-cycle $(u, v, a, b)$ such that 
\Statex \hspace{0.5em} $u\in R_{1a,\kappa}$, $v\in R_{1b,\kappa}$, $a\in R_{2a,\kappa}$, $b\in R_{2b,\kappa}$
\State \Return whether any four-cycle exists in the collected subgraph
\end{algorithmic}
\end{algorithm}

\subsection{Configurations and Labeled Substructures}
\label{sec:configurations}

As we already discussed above, each four-cycle~$A$ can be detected by our algorithm in different ways: it can be sampled for each value $\kappa\in\CI$, and it can also be detected in multiple ways based on how its nodes were sampled into the sets $S_{1,\kappa}$, $S_{2,\kappa}$, $R_{1a,\kappa}$, $R_{1b,\kappa}$, $R_{2a,\kappa}$ and $R_{2b,\kappa}$. For this reason, we now introduce \emph{configurations} and \emph{labeled substructures}, which let us reason about how the algorithm detects a four-cycle~$A$.

\smallskip
\paragraph{Configurations.} We start by defining configurations.

\begin{definition}
\label{def:config}
    Given a four-cycle $A\in\squares$ and $\kappa\in \CI$, a \emph{configuration} is a tuple $(A,\kappa,x,y)\in \squares\times \CI\times V^2$ where $x,y\in V(A)$ and $x<y$. 
\end{definition}

Note that there are $|\squares|\times |\CI|\times \binom{4}{2} = \Theta(T \log(T))$ configurations in total. 

Further, observe that the definition of a configuration is completely deterministic and currently completely independent from our sampling. To connect this definition to our sampling, we are interested in whether certain well-behaved configurations are \emph{realized} by our random sample. Thus whether a configuration is realized is a random event and, in a nutshell, all realized configuration satisfy that either $x,y\in S_{1,\kappa}$ or that $x\in R_{1a,\kappa}$ and $y\in R_{1b,\kappa}$, i.e., they contain two nodes that were sampled with probability $p_{1,\kappa}$.

More specifically, consider a four-cycle $A$ with node set $V(A) = \{x,y,u,v\}$ and a configuration $c=(A,\kappa,x,y)$. If $x$ and $y$ are opposite nodes of $A$, then we say that $c$ is \emph{realized} if
\begin{align}
\label{eq:realized-S}
    x,y\in S_{1,\kappa}\quad  \text{ and } \quad u,v \in S_{2,\kappa}.
\end{align}
Otherwise, if $x$ and $y$ are adjacent in $A$, $c$ is \emph{realized} if 
\begin{align}
\label{eq:realized-R}
    x\in R_{1a,\kappa},\quad y\in R_{1b,\kappa},\quad u\in R_{2a,\kappa} \quad \text{ and } \quad v\in R_{2b,\kappa},
\end{align}
where $v$ and $y$ are the neighbors of $x$ in $A$.

Observe that both of these cases are in line with what we discussed above: Realized configurations contain exactly two nodes that were sampled with probability $p_{1,\kappa}$ (namely $x$ and $y$) and two nodes that were sampled with probability $p_{2,\kappa}$ (namely $u$ and $v$). 
We illustrate this in \Cref{fig:configurations}.

Now observe that if a configuration $(A,\kappa,x,y)$ is realized, then it is much more likely that another configuration $(A',\kappa,x,y')$, which also contains node~$x$ and uses the same $\kappa$-value, is also realized. 
Thus, we say that two configurations $c$ and $c'$ are \emph{dependent} if the random events that $c$ and $c'$ are realized are dependent.
To argue about this dependence more precisely, we will now have to define labeled substructures and when they are realized.

\smallskip
\paragraph{Labelings of Configurations and Labeled Substructures.}
Now, before we define labeled substructures, we say that the \emph{labeling $\ell_c$ of a configuration} $c=(A,\kappa,x,y)$ is a mapping $\ell_c \colon V(A)\to\{S_1,S_2,R_{1a},R_{1b},R_{2a},R_{2b}\}$, which encodes in which sets the nodes $v\in V(A)$ must be contained such that $c$ is realized. 
Note that here we omit the subscript $\kappa$ of the sets on purpose, because we only want $S_1$, $S_2$, etc. to denote \emph{labels} and not concrete random sets.
For instance, if $x$ and $y$ are opposite nodes in $A=(x,u,y,v)$, then the labeling $\ell_c$ of $c=(A,\kappa,x,y)$ is given by  $\ell_c(x) = S_1$, $\ell_c(y) = S_1$, $\ell_c(u) = S_2$, $\ell_c(v) = S_2$.

Next, recall from the preliminaries that a \emph{substructure} is a node, edge or wedge of a four-cycle. Further, we will briefly also consider a pair of opposite nodes of a four-cycle as a substructure, but we will see below that we tacitly ignore it. For convenience, given a substructure $\substr$, we denote its set of nodes by $V(\substr)$.

We will call such mappings \emph{labelings} and we extend this idea to \emph{labeled substructures}.

\begin{definition}\label{def:label_substructure}
    A \emph{labeled substructure} $(\substr,\kappa,\ell)$ is a substructure $\substr$ together with a value $\kappa\in\CI$ and a labeling function $\ell: V(\substr) \rightarrow \{S_1,S_2,R_{1a},R_{1b},R_{2a},R_{2b}\}$.
\end{definition}

Again, observe that the definition of a labeled substructure is completely deterministic, and we stress again that here $S_1,S_2,R_{1a},R_{1b},R_{2a},R_{2b}$ are \emph{labels}, rather than (random) sets.

Now we say that a labeled substructure $(\substr,\kappa,\ell)$ is \emph{realized} if each node $v\in V(\substr)$ is sampled into the node set defined by $\kappa$ and the labeling $\ell(v)$. 

For instance, if $\substr$ consists of a single node $v$, then the labeled substructure $(\{v\},\kappa,v\mapsto S_1)$ is realized if $v\in S_{1,\kappa}$. As another example, consider a labeled substructure $(e,\kappa,\ell_e)$, where $e=(u,v)$ and $\ell_e(u) = S_1$ and $\ell_e(v) = S_2$. Then this labeled substructure is realized if $u\in S_{1,\kappa}$ and $v\in S_{2,\kappa}$.

Now observe that whether a labeled substructure $(\substr,\kappa,\ell)$ is realized is a random event and we denote the probability of this event by $p_{(\substr,\kappa,\ell)}$. Note that
\begin{align*}
    p_{(\substr,\kappa,\ell)}=p_{1,\kappa}^i p_{2,\kappa}^j,
\end{align*} 
where $i=|\ell^{-1}(\{S_1,R_{1a},R_{1b}\})|$ is the number of nodes of $\substr$ with a label of a set into which nodes are sampled with probability $p_{1,\kappa}$, and $j=|\ell^{-1}(\{S_2,R_{2a},R_{2b}\})|$ is the number of nodes with a label of a set into which nodes are sampled with probability $p_{2,\kappa}$.

Now we can connect configurations and labeled substructures as follows. Given a configuration $c=(A,\kappa',x',y')$ and a labeled substructure $(\substr,\kappa,\ell)$ we write $(\substr,\kappa,\ell)\in c$ if $\kappa=\kappa'$ and $V(\substr)\subset V(A)$ and $\ell=\ell_c\big|_{V(\substr)}$, i.e., when restricting the domain of the labeling~$\ell_c$ of the configuration~$c$ to the vertices in $\substr$ then the two labelings must agree.
Observe that if $c$ is realized, then any labeled substructure $(\substr,\kappa,\ell)\in c$ is realized.

Note that each configuration~$c$ contains 14 labeled substructures: 4 nodes, 4 edges, 4 wedges, and 2 opposite node pairs, each labeled according to $\ell_c$.
We obtain the following observation, which characterizes when two configurations are dependent.

\begin{observation}
    Two configurations $c$ and $c'$ are dependent if and only if there exists a labeled substructure $(\substr,\kappa,\ell)$ with $(\substr,\kappa,\ell)\in c$ and $(\substr,\kappa,\ell)\in c'$.
    \label{observ:dependent_configurations}
\end{observation}

\subsection{Heaviness and Assigned Configurations}
\label{sec:heaviness}

To bound the variance of our estimator below, we must ensure that our sample does not contain substructures of too high heaviness compared to the probability with which they were sampled. Thus, below we categorize each labeled substructure $(\substr,\kappa,\ell)$ as \emph{heavy} or \emph{light} based on whether the heaviness $t(\substr)$ exceeds a heaviness threshold $\theta_{(\substr,\kappa,\ell)}=\widetilde{\Theta}(Tp_{(\substr,\kappa,\ell)})$ which depends on the probability with which this labeled substructure is realized.

First, recall from the preliminaries that the \emph{heaviness} $t(\substr)$ of a node, edge or wedge $\substr$ is defined as the number of four-cycles containing $\substr$. Next, the heaviness of an opposite node pair $\substr$ with $V(\substr)=\{u,v\}$ is the number of four-cycles containing the node as an opposite pair, i.e.,  $t(\substr)=\os(u,v)$ (note that this does not count four-cycles in which $u$ and $v$ are adjacent).

Now, the thresholds $\theta_{(\substr,\kappa,\ell)}$ for labeled substructures are given below. Note that in the case distinctions we use the sampling probabilities $p_{(\substr,\kappa,\ell)}$ of the labeled substructures to define the thresholds. This notation is relatively succinct and highlights the connection between thresholds and sampling probabilities, but we note that we could as well define the thresholds based on the labelings of $\substr$ (see below).

A labeled node $(v,\kappa,\ell)$ is \emph{heavy} if
\begin{align*}
    t(v)>
    \theta_{(v,\kappa,\ell)}
    :=
    \begin{cases}
    \kappa\sqrt{T}/\delta^{1.5} & \text{ if }\; p_{(v,\kappa,\ell)} = p_{1,\kappa}, \\
    T/(\kappa\delta^{1.5}) & \text{ if }\; p_{(v,\kappa,\ell)} = p_{2,\kappa}.
    \end{cases}
\end{align*}
A labeled edge $(e,\kappa,\ell)$ is \emph{heavy} if
\begin{align*}
    t(e)
    >\theta_{(e,\kappa,\ell)}
    :=
    \begin{cases}
    \kappa^2/\delta^{2} & \text{ if }\; p_{(e,\kappa,\ell)}=p_{1,\kappa}^2, \\
    \sqrt{T}/\delta^{2} & \text{ if }\; p_{(e,\kappa,\ell)}=p_{1,\kappa} p_{2,\kappa},\\
    T/(\kappa^2\delta^{2}) & \text{ if }\; p_{(e,\kappa,\ell)}=p_{2,\kappa}^2.
    \end{cases}
\end{align*}
A labeled wedge $(\wedge,\kappa,\ell)$ is \emph{heavy} if 
\begin{align*}
    t(\wedge)>\theta_{(\wedge,\kappa,\ell)}
    :=
    \begin{cases}
    \kappa/\delta & \text{ if }\; p_{(\wedge,\kappa,\ell)}=p_{1,\kappa}^2 p_{2,\kappa}, \\
    \sqrt{T}/(\kappa\delta) & \text{ if }\; p_{(\wedge,\kappa,\ell)}=p_{1,\kappa} p_{2,\kappa}^2.
    \end{cases}
\end{align*}  
A labeled opposite node pair $(\substr,\kappa,\ell)$ is \emph{heavy} if
\begin{align*}
    t(\substr)
    >\theta_{(\substr,\kappa,\ell)}
    :=
    \begin{cases}
    \kappa^2/\delta^{2} & \text{ if }\; p_{(\substr,\kappa,\ell)}=p_{1,\kappa}^2, \\
    \sqrt{T}/\delta^{2} & \text{ if }\; p_{(\substr,\kappa,\ell)}=p_{1,\kappa} p_{2,\kappa},\\
    T/(\kappa^2\delta^{2}) & \text{ if }\; p_{(\substr,\kappa,\ell)}=p_{2,\kappa}^2.
    \end{cases}
\end{align*}
Otherwise, the labeled substructure is called \emph{light}.

\begin{figure}[t]
    \centering
    \tikzset{customnode/.style={circle, draw, minimum size=20pt, inner sep=0pt}}

\definecolor{color1}{HTML}{1b9e77}
\definecolor{color2}{HTML}{7570b3}
\colorlet{w1}{color1}
\colorlet{w2}{color2}

\hspace{1cm}

\begin{subfigure}[t]{0.45\textwidth}
\centering
    \begin{tikzpicture}[scale=1.8]
    \footnotesize
        \node[customnode, label=below left:{$p_{1}$}] (a) at (0,0) {\scriptsize $S_1$};
        \node[customnode, label=above left:{$p_{2}$}] (b) at (0,1) {\scriptsize $S_2$};
        \node[customnode, label=above right:{$p_{1}$}] (c) at (1,1) {\scriptsize $S_1$};
        \node[customnode, label=below right:{$p_{2}$}] (d) at (1,0) {\scriptsize $S_2$};
        
        
      \draw[] (a) -- node[midway, left] {$p_{1}p_{2}$} (b);
      \draw (b) -- node[midway, above] {$p_{1}p_{2}$} (c);
      \draw (c) -- node[midway, right] {$p_{1}p_{2}$} (d);
      \draw (d) -- node[midway, below] {$p_{1}p_{2}$} (a);

      \draw[w1,very thick, rounded corners=6pt]  ($ (b) + (5pt, -3pt) $) -|  ($ (d) + (-3pt, 5pt) $);
      \draw[w2,very thick, rounded corners=6pt]  ($ (a) + (3pt, 5pt) $) |-  ($ (c) + (-5.5pt, -5pt) $);
      \node[] (w1) at (0.7,0.6){\textcolor{w1}{$p_{1}p_{2}^2$}};
      \node[] (w2) at (0.3,0.4){\textcolor{w2}{$p_{1}^2p_{2}$}};

    \end{tikzpicture}
    \caption{Sampling probabilities for opposite nodes in $S_1$.}
    \label{fig:opposite_prob}
\end{subfigure}
\hfill
\begin{subfigure}[t]{0.45\textwidth}
\centering
    \begin{tikzpicture}[scale=1.8]
    \footnotesize
        \node[customnode, label=below left:{$p_1$}] (a) at (0,0) {\scriptsize $R_{1a}$};
        \node[customnode, label=above left:{$p_2$}] (b) at (0,1) {\scriptsize $R_{2b}$};
        \node[customnode, label=above right:{$p_2$}] (c) at (1,1) {\scriptsize $R_{2a}$};
        \node[customnode, label=below right:{$p_1$}] (d) at (1,0) {\scriptsize $R_{1b}$};
        
        
      \draw[] (a) -- node[midway, left] {$p_1p_2$} (b);
      \draw (b) -- node[midway, above] {$p_2^2$} (c);
      \draw (c) -- node[midway, right] {$p_1p_2$} (d);
      \draw[dashed] (d) -- node[midway, below] {$p_1^2$} (a);

      \draw[w1,very thick, rounded corners=6pt]  ($ (b) + (5pt, -3pt) $) -|  ($ (d) + (-3pt, 5pt) $);
      \draw[w2,very thick, rounded corners=6pt]  ($ (a) + (3pt, 5pt) $) |-  ($ (c) + (-5.5pt, -5pt) $);
      \node[] (w1) at (0.7,0.6){\textcolor{w1}{$p_{1}p_{2}^2$}};
      \node[] (w2) at (0.3,0.4){\textcolor{w2}{$p_{1}p_{2}^2$}};
    \end{tikzpicture}
    \caption{Sampling probabilities for adjacent nodes in $R_{1a}$ and $R_{1b}$.}
    \label{fig:adjacent_prob}
\end{subfigure}
\hspace{1cm}

\vspace{0.5cm}

\hspace{1cm}

\begin{subfigure}[t]{0.45\textwidth}
\centering
    \begin{tikzpicture}[scale=1.8]
        \node[customnode, label=below left:{$\frac{\kappa\sqrt{T}}{\delta^{1.5}}$}] (a) at (0,0) {\tiny $S_{1,\kappa}$};
        \node[customnode, label=above left:{$\frac{T}{\kappa\delta^{1.5}}$}] (b) at (0,1) {\tiny $S_{2,\kappa}$};
        \node[customnode, label=above right:{$\frac{\kappa\sqrt{T}}{\delta^{1.5}}$}] (c) at (1,1) {\tiny $S_{1,\kappa}$};
        \node[customnode, label=below right:{$\frac{T}{\kappa\delta^{1.5}}$}] (d) at (1,0) {\tiny $S_{2,\kappa}$};
        
        
      \draw[] (a) -- node[midway, left] {$\frac{\sqrt{T}}{\delta^2}$} (b);
      \draw (b) -- node[midway, above] {$\frac{\sqrt{T}}{\delta^2}$} (c);
      \draw (c) -- node[midway, right] {$\frac{\sqrt{T}}{\delta^2}$} (d);
      \draw (d) -- node[midway, below] {$\frac{\sqrt{T}}{\delta^2}$} (a);

      \draw[w1,very thick, rounded corners=6pt]  ($ (b) + (5pt, -3pt) $) -|  ($ (d) + (-3pt, 5pt) $);
      \draw[w2,very thick, rounded corners=6pt]  ($ (a) + (3pt, 5pt) $) |-  ($ (c) + (-5.5pt, -5pt) $);
      
      \node[] (w1) at (0.71,0.5){\textcolor{w1}{$\frac{\sqrt{T}}{\kappa\delta}$}};
      \node[] (w2) at (0.2,0.5){\textcolor{w2}{$\frac{\kappa}{\delta}$}};

    \end{tikzpicture}
    \caption{Heaviness thresholds for opposite nodes in $S_{1,\kappa}$}
    \label{fig:opposite_theta}
\end{subfigure}
\hfill
\begin{subfigure}[t]{0.45\textwidth}
\centering
    \begin{tikzpicture}[scale=1.8]
        \node[customnode, label=below left:{$\frac{\kappa\sqrt{T}}{\delta^{1.5}}$}] (a) at (0,0) {\tiny $R_{1a,\kappa}$};
        \node[customnode, label=above left:{$\frac{T}{\kappa\delta^{1.5}}$}] (b) at (0,1) {\tiny $R_{2b,\kappa}$};
        \node[customnode, label=above right:{$\frac{T}{\kappa\delta^{1.5}}$}] (c) at (1,1) {\tiny $R_{2a,\kappa}$};
        \node[customnode, label=below right:{$\frac{\kappa\sqrt{T}}{\delta^{1.5}}$}] (d) at (1,0) {\tiny $R_{1b,\kappa}$};
        
        
      \draw[] (a) -- node[midway, left] {$\frac{\sqrt{T}}{\delta^2}$} (b);
      \draw (b) -- node[midway, above] {$\frac{T}{\kappa^2\delta^2}$} (c);
      \draw (c) -- node[midway, right] {$\frac{\sqrt{T}}{\delta^2}$} (d);
      \draw[dashed] (d) -- node[midway, below] {$\frac{\kappa^2}{\delta^2}$} (a);

      \draw[w1,very thick, rounded corners=6pt]  ($ (b) + (5pt, -3pt) $) -|  ($ (d) + (-3pt, 5pt) $);
      \draw[w2,very thick, rounded corners=6pt]  ($ (a) + (3pt, 5pt) $) |-  ($ (c) + (-5.5pt, -5pt) $);
      
      \node[] (w1) at (0.71,0.5){\textcolor{w1}{$\frac{\sqrt{T}}{\kappa\delta}$}};
      \node[] (w2) at (0.29,0.5){\textcolor{w2}{$\frac{\sqrt{T}}{\kappa\delta}$}};
    \end{tikzpicture}
    \caption{Heaviness thresholds for adjacent nodes in $R_{1a,\kappa}$ and $R_{1b,\kappa}$.}
    \label{fig:adjacent_theta}
\end{subfigure}
\hspace{1cm}
    \caption{Visualization of the probabilities $p_{(\substr,\kappa,\ell)}$ and the heaviness thresholds $\theta_{(\substr,\kappa,\ell)}$ for given labeled substructures $(\substr,\kappa,\ell)$. In the first two figures, we annotated each vertex~$v$ with its label~$\ell(v)$; for the probabilities, we omitted the dependency on~$\kappa$ for better readability. The third and fourth figure include the dependency on~$\kappa$ and list the heaviness thresholds for each node, each edge, and two of the four wedges.}
    \label{fig:heavinessthresholds}
\end{figure}
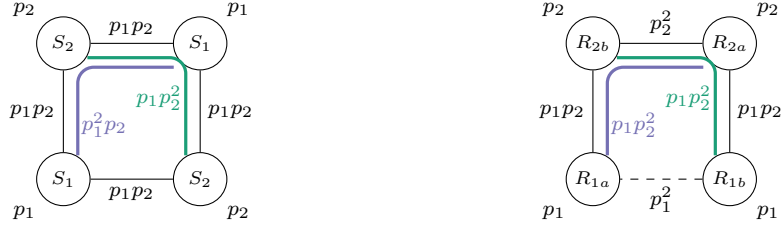
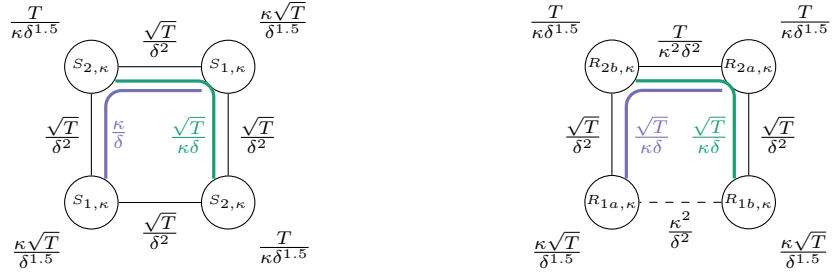

We provide a visualization of the heaviness thresholds in \cref{fig:heavinessthresholds}, and note that the thresholds for edges and for opposite node pairs are the same if they are realized with the same probabilities.

Given the definitions above, note that we could equivalently have defined the thresholds in terms of the labelings of the substructures. For instance, above we write $p_{(v,\kappa,\ell)} = p_{1,\kappa}$ which is equivalent to writing $\ell(v)\in\{S_1,R_{1a},R_{1b}\}$. As another example, consider an edge $e=(u,v)$. Above we write $p_{(e,\kappa,\ell)} = p_{1,\kappa}^2$ which is equivalent to writing $(\ell(u),\ell(v)) \in \{ S_1, R_{1a}, R_{1b} \}^2$, i.e., both endpoints of the edge are contained in a set into which nodes are sampled with probability $p_{1,\kappa}$.

We note that for each substructure, the thresholds are listed from high to low. For instance, for labeled nodes one can check that $\kappa\sqrt{T}/\delta^{1.5} \geq T/(\kappa\delta^{1.5})$ for all $\kappa\in\CI$, and for labeled edges one can check that $\kappa^2/\delta^2 \geq \sqrt{T}/\delta^2 \geq T/(\kappa^2\delta^2)$ for all $\kappa\in\CI$, and so on. This observation will be useful later in the proof of \Cref{thm:expectation}.

Next, we note that the thresholds are not set arbitrarily. They are all set such that they are in $\Theta(T p_{(\substr,\kappa,\ell)})$ but later it will be crucial that they are a large constant factor smaller than $T p_{(\substr,\kappa,\ell)}$.
\begin{observation}
\label{obs:relationship-theta-t-p}
    For each labeled substructure $(\substr,\kappa,\ell)$, the heaviness threshold is bounded by $\theta_{(\substr,\kappa,\ell)}\leq {Tp_{(\substr,\kappa,\ell)}}/{c_1}$ where $c_1 > 1$ is a constant.
\end{observation}

Using the substructure classification, we also define the heaviness for configurations and four-cycles:
\begin{itemize}
    \item A configuration is \emph{heavy} if any of its labeled substructures is heavy, and otherwise it is \emph{light}.
    \item A four-cycle is \emph{heavy} if all of its configurations are heavy, and otherwise it is \emph{light}.
\end{itemize}
Note that a light four-cycle can have $1$ to $6|\CI|$ light configurations. 
Thus, in such cases, we will select one unique light configuration per four-cycle.
\smallskip
\paragraph{Wedges Dominate Opposite Node Pairs.}
Next, we argue that 
the heaviness constraints on opposite node pairs are redundant. In the following, this will allow us to only consider nodes, edges and wedges as substructures, since we show that a configuration is already light if its labeled nodes, edges and wedges are light.

First, observe that the heaviness of an opposite node pair $\{u,v\}$ and a wedge $\wedge$ connecting it are closely related since $t(\{u,v\})=\os(u,v)=\binom{t(\wedge)+1}{2}$.
Second, for any configuration $c$ with light wedges, consider a labeled opposite node pair $(\substr,\kappa,\ell_\substr)$ with $V(\substr)=\{u',v'\}$ and a wedge $(\wedge,\kappa,\ell_\wedge)\in c$ with endpoints $u'$ and $v'$. As $(\wedge,\kappa,\ell_\wedge)$ is light we have 
$t(\wedge)<\theta_{(\wedge,\kappa,\ell_\wedge)}$ and thus 
\begin{align*}
    t(\substr)
    =\os(u',v')
    =\binom{t(\wedge)+1}{2}
    \leq \binom{\theta_{(\wedge,\kappa,\ell_{\wedge})}+1}{2} 
    \leq\theta_{(\wedge,\kappa,\ell_{\wedge})}^2.
\end{align*}
Now note that if $p_{(\substr,\kappa,\ell_\substr)}=p_{1,\kappa}^2$ or $p_{(\substr,\kappa,\ell_\substr)}=p_{2,\kappa}^2$ we have $\theta_{(\substr,\kappa,\ell_\substr)}=\theta_{(\wedge,\kappa,\ell_\wedge)}^2$ and thus the opposite node pair is light. If $p_{(\substr,\kappa,\ell_\substr)}=p_{1,\kappa}p_{2,\kappa}$ consider the unique wedge $(\wedge,\kappa,\ell_\wedge)\in c$ with endpoints $u',v'$ and $ p_{(\wedge,\kappa,\ell_\wedge)}=p_{1,\kappa}p_{2,\kappa}^2$. For this wedge we have $\theta_{(\wedge,\kappa,\ell_\substr)}^2\leq \theta_{(\substr,\kappa,\ell_\substr)}$ and thus $\substr$ is light.

We conclude that the lightness of wedges implies the lightness of opposite node pairs, and thus, for the rest of this section, we only consider nodes, edges and wedges as substructures.
\smallskip
\paragraph{Assigning Configurations to Four-Cycles.}
We now proceed by \emph{assigning} a unique configuration to each four-cycle. This is crucial to bound the variance of our estimator.
Below, in \cref{thm:expectation}, we show that nearly all four-cycles have a light configuration and thereby prove that nearly all four-cycles are light.

Specifically, the \emph{configuration assigned to a four-cycle}~$A$ is the configuration $\sigma_A:=(A,\kappa_A,x_A,y_A)$, where $\kappa_A$, $x_A$ and $y_A$ are set as described below.
We will prove that this configuration $\sigma_A$ is light for almost all four-cycles (see \Cref{thm:expectation}).

\emph{Choice of $\kappa_A$:} 
Intuitively, we want to choose $\kappa_A$ as the smallest $\kappa\in\CI$ over all light configurations of~$A$. However, for technical reasons, it will be convenient to work with the following definition, where we set
\begin{equation}
    \kappa_A=\Bigg\lceil \max_{\substack{v\in V(A)\\e\in E(A)\\ \wedge\in W(A)}}\left\{T^{\frac{1}{4}},\delta^{1.5} \frac{t(v)}{\sqrt{T}},\delta \sqrt{t(e)},\delta t(\wedge)\right\}\Bigg\rceil_\CI, \label{eq:w_A}
\end{equation}
where the maximum is taken over all nodes, edges or wedges in $A$ and then rounded up to the next value in $\CI$ (by at most a factor of 2).

Indeed, for this definition we note that picking a smaller value for~$\kappa$ directly violates the heaviness threshold for the substructure attaining the maximum above. Thus, the smallest value of $\kappa$ such that $A$ is light is at least $\kappa_A$, but could potentially be higher. This implies that there might be four-cycles which have a light configuration, but for which the assigned configuration $\sigma_A$ is not light. However, in \Cref{thm:expectation} we show that this is not a problem since with this choice of $\kappa_A$ and the choice of $x_A$ and $y_A$ below, a vast majority of the configurations $\sigma_A$ are light.

\emph{Choice of $x_A$ and $y_A$:} 
To finish the description of $\sigma_A$, we assign the nodes $x_A$ and $y_A$ for light four-cycles $A$ as follows:
\begin{itemize}
\item If $A$ contains a $\sqrt{T}/\delta^2$-heavy edge then 
we assign $x_A$ and $y_A$ to be the endpoints of this edge, i.e., in this case $x_A$ and $y_A$ are adjacent.
\item Otherwise, we assign an opposite node pair of $A$ based on the value $\kappa_A$ (where ties are broken arbitrarily) and the substructure $\substr$ attaining the maximum in \cref{eq:w_A} as follows:
    \begin{itemize}
        \item If $\kappa_A=T^{1/4}$ we assign the lexicographically smaller opposite node pair of $A$ as $x_A$ and $y_A$.
        \item If $\substr$ is a wedge, we assign its endpoints as $x_A$ and $y_A$.
        \item If $\substr$ is a node, we assign $\substr$ and its opposite node in $A$ as $x_A$ and $y_A$.
        \item Note that in this case $\substr$ cannot be an edge: Since for all edges in $A$, we have that $t(e) < \sqrt{T}/\delta^2$ and thus we get that $\delta\sqrt{t(e)} < T^{1/4}$, implying that none of these edges can attain the maximum in \Cref{eq:w_A}.
    \end{itemize}
\end{itemize}

Observe that this choice of $x_A$ and $y_A$ is exactly corresponding to which sampling ``rule'' should be used such that the algorithm detects~$A$: If the heaviest substructure in $A$ is a heavy edge, then we want to pick $x_A$ and $y_A$ as adjacent nodes, and in all other cases we want to pick $x_A$ and $y_A$ as opposite nodes in $A$.

\subsection{Bounding the Number of Light Four-Cycles}
\label{sec:bound-light-squares}

Based on the assignment rule $\sigma_A$ for assigning unique light configurations to light four-cycles, we can now prove the following proposition, showing that the vast majority of four-cycles are light.

\begin{proposition}
    \label{thm:expectation}
    Let $\delta \in (0, \min\{\frac{1}{2}, \frac{1}{1049\log(T)}\})$.
    For at least $(1-1049\delta\log(T))T$~four-cycles
    the configuration $\sigma_A$ is light.
\end{proposition}

We note that these excluded four-cycles are also all heavy four-cycles, but we do not need this and thus also do not prove it.

To prove \cref{thm:expectation} we use the next three lemmas to exclude at most $(328+145+576)\delta\log(T)T$ four-cycles.

We start by restating a known result by \citet{mcgregor2020triangle}, showing that only a few four-cycles contain two or more heavy edges.
Recall the definition of $\theta$-heavy nodes, edges and wedges from \cref{sec:prelim}.
\begin{lemma}[\citet{mcgregor2020triangle}]
    \label{lemma:no2heavyedges}
    At most $328\delta^2 T$ four-cycles contain two or more $\frac{\sqrt{T}}{4\delta^2}$-heavy edges.
\end{lemma}

Next, we show that only a few four-cycles contain a heavy edge and also a heavy wedge.
\begin{lemma}
    \label{lemma:heavyedge_cant_be_in_onions}
    At most $145\delta T$ four-cycles contain at least one $\frac{\sqrt{T}}{4\delta^2}$-heavy edge and a $\frac{1}{\delta}$-heavy wedge.
\end{lemma}
\begin{proof}
    In this proof, we call four-cycles with the property stated in the lemma \emph{bad}.
    
    Now let $u$ and $v$ be endpoints of a $\frac{1}{\delta}$-wedge. This implies that $u$ and $v$ form an onion of size 
    \begin{align*}
        \os(u,v) = \binom{\ow(u,v)}{2} \geq \binom{1/\delta}{2} >\frac{1}{4\delta^2},
    \end{align*}
    where we used that $\delta\leq\frac{1}{2}$. This implies that, since each four-cycle is part of two onions, there are at most $8\delta^2T$ onions of size at least $\frac{1}{4\delta^2}$, as otherwise there would be more than $T$ four-cycles in total.
    
    Next, let $h(u,v)$ be the number of wedges between $u$ and $v$ which contain at least one $\frac{\sqrt{T}}{4\delta^2}$-heavy edge.
    Then any bad four-cycle in the onion with endpoints $u$ and $v$ consists of one of these $h(u,v)$ wedges and one other wedge connecting $ u$ and $ v$.
    Hence, the number of bad four-cycles in this onion is at most $h(u,v)\ow(u,v)$.
    
    Thus, the total number of bad four-cycles in the entire graph is bounded by
    \begin{align*}
        \sum_{\substack{u,v\in V\\ \ow(u,v)>\frac{1}{\delta}}} h(u,v)\ow(u,v) &\leq \sum_{\substack{u,v\in V\\ \ow(u,v)>\frac{1}{\delta}\\ h(u,v)\geq 2}} h(u,v)\ow(u,v) + \sum_{\substack{u,v\in V\\ \ow(u,v)>\frac{1}{\delta}}} \ow(u,v).
    \end{align*}
    
    Now let $D_\kappa$ denote the number of onions of width $\kappa$. Note that our result above implies that $\sum_{\kappa\geq 1/\delta} D_\kappa \leq 8\delta^2 T$.
    Now we get
    \begin{align*}
        \sum_{\substack{u,v\in V\\ \ow(u,v)>\frac{1}{\delta}}} \ow(u,v) &\leq \sum_{\kappa=\frac{1}{\delta}}^{\sqrt{T}} D_\kappa \kappa\\
        &\leq \sqrt{\sum_{\kappa=\frac{1}{\delta}}^{\sqrt{T}} D_\kappa} \sqrt{\sum_{\kappa=\frac{1}{\delta}}^{\sqrt{T}} D_\kappa \kappa^2}\\
        &\leq \sqrt{8\delta^2 T\cdot 8T} \\
        &\leq 8\delta T,
    \end{align*}
    where we used the Cauchy--Schwarz inequality. Further, we used the identity $\sum_{\kappa=1}^{\sqrt{T}} D_\kappa \binom{\kappa}{2} = 2T$, where the factor of 2 stems from the fact that each four-cycle is part of two onions. Finally, we used that $\kappa^2\leq 4\binom{\kappa}{2}$. This bounds the second summand above.

    Next, we bound the first summand and get that
    \begin{align*}
        \sum_{\substack{u,v\in V\\ \ow(u,v)>\frac{1}{\delta}\\ h(u,v)\geq 2}} h(u,v)\ow(u,v) &\leq  \sqrt{\sum_{\substack{u,v\in V\\ \ow(u,v)>\frac{1}{\delta}\\ h(u,v)\geq 2}} h(u,v)^2} \cdot \sqrt{\sum_{\substack{u,v\in V\\ \ow(u,v)>\frac{1}{\delta}}} \ow(u,v)^2}\\
        &\leq  \sqrt{\left( \sum_{\substack{u,v\in V\\ \ow(u,v)>\frac{1}{\delta}}} 4\binom{h(u,v)}{2}\right) \cdot  8T}\\
        &\leq  \sqrt{8\cdot328\delta^2 T \cdot 8T} \\
        &\leq 145\delta T,
    \end{align*}
    where we again used the Cauchy--Schwarz inequality in the first step. In the second step we again used that $h(u,v)^2 \leq 4 \binom{h(u,v)}{2}$ and that $\sum_{u,v}\ow(u,v)^2 \leq \sum_{u,v} 4 \binom{\ow(u,v)}{2} \leq 8T$. Then in the final step we applied \cref{lemma:no2heavyedges} since $\binom{h(u,v)}{2}$ counts pairs of wedges from $u$ to $v$ which both contain a $\sqrt{T}/(4\delta^2)$-heavy edge, and thus form a four-cycle containing two $\sqrt{T}/(4\delta^2)$-heavy edges.
\end{proof}

Next, we provide a combination lemma, which bounds the number of four-cycles containing multiple heavy configurations.
\begin{lemma}[Combination lemma]
    \label{lemma:combination}
    Across all possible choices for $\kappa\in\CI$, in total, there are at most
    $576\delta\log(T)T$ four-cycles with one of the following properties:
    \begin{enumerate}[label=(\roman*),ref=(\roman*)]
        \item \label{comb:1a} One $\kappa/(2\delta)$-heavy wedge and one $\sqrt{T}/(\kappa\delta)$-heavy wedge sharing exactly one edge.
        \item \label{comb:1b} One $\kappa/(2\delta)$-heavy wedge and one disjoint $T/(\kappa\delta^{1.5})$-heavy node.
        \item \label{comb:2a} One $\kappa^2/(4\delta^2)$-heavy edge and one disjoint $T/(\kappa^2\delta^2)$-heavy edge.
        \item \label{comb:2b} One $\kappa^2/(4\delta^2)$-heavy edge $e$ and one disjoint $T/(\kappa\delta^{1.5})$-heavy node $v$ and the wedge $\wedge$ using $v$ and $e$ fulfilling $t(\wedge)\leq \kappa/\delta$.
        \item \label{comb:3a4a} One $\kappa\sqrt{T}/(2\delta^{1.5})$-heavy node $v$ and one adjacent $T/(\kappa\delta^{1.5})$-heavy node $y$ and the edge $e=(v,y)$ fulfilling $t(e)\leq\sqrt{T}/\delta^2$.
        \item \label{comb:3b} One $\kappa\sqrt{T}/(2\delta^{1.5})$-heavy node $v$ and one opposite $T/(\kappa\delta^{1.5})$-heavy node $y$ and $\os(v,y)\leq \sqrt{T}/\delta^2$.
        \item \label{comb:3c} One $\kappa\sqrt{T}/(2\delta^{1.5})$-node $v$ and one disjoint $T/(\kappa^2\delta^2)$-heavy edge $f$ and the wedge $\wedge$ using $v$ and $f$ fulfilling $t(\wedge)\leq\sqrt{T}/(\kappa\delta)$.
        \item \label{comb:4b} One $\kappa\sqrt{T}/(2\delta^{1.5})$-heavy node and one disjoint $\sqrt{T}/(\kappa\delta)$-heavy wedge.
        \item \label{comb:4c} One $\kappa\sqrt{T}/(2\delta^{1.5})$-heavy node $v$ and one disjoint $\sqrt{T}/\delta^2$-heavy edge $e$ and the wedge $\wedge$ using $v$ and $e$ fulfilling $t(\wedge)\leq \kappa/\delta$.
    \end{enumerate}
\end{lemma}
\begin{proof}
    We give a proof for \cref{comb:3a4a}. All other cases follow analogously.

    First, observe that for any heaviness threshold $\theta>1$, there can be at most $4T/\theta$ nodes, edges,  wedges and node pairs of heaviness at least $\theta$, as otherwise there would be more than $T$ four-cycles in total. 

    Next, consider any $\kappa\in\CI$. For each such choice of $\kappa$, we will exclude at most $288 \delta T$~four-cycles.
    
    Let $\theta_1=\kappa\sqrt{T}/(2\delta^{1.5})$ and $\theta_2=T/(\kappa\delta^{1.5})$ and $r=\sqrt{T}/\delta^2$. 
    Based on the above observation, there are at most $16T^2/(\theta_1\theta_2)$ pairs of one $\theta_1$-heavy-node and one $\theta_2$-heavy-node, and thus also at most this many adjacent pairs.
    For each such pair, we get at most $r$ four-cycles containing both nodes of the pair and fulfilling the conditions of \cref{comb:3a4a}.
    Thus, we get at most 
    \begin{equation*}
        \frac{4T}{\theta_1}\frac{4T}{\theta_2}r\leq \delta^3\frac{8\sqrt{T}}{\kappa}\cdot 4\kappa\cdot \frac{\sqrt{T}}{\delta^2} \leq 32\delta T
    \end{equation*}
    such four-cycles.

    All other cases except (i) follow by adjusting the values for $\theta_1,\theta_2$ and $r$ and observing $\frac{4T}{\theta_1}\frac{4T}{\theta_2}r\leq 32\delta T$.
    Note that $r=1$ in cases (ii), (iii) and (viii).

    For case (i), we observe that if a wedge is $\theta$-heavy its endpoints form a $\theta^2/2$-heavy node pair. Thus, any four-cycle satisfying (i) contains a $\kappa^2/(8\delta^2)$-heavy node pair and a disjoint $T/(2\kappa^2\delta^2)$-heavy node pair. Now applying the same argument as before proves that at most  $\frac{4T}{\kappa^2/(8\delta^2)}\frac{4T}{T/(2\kappa^2\delta^2)}\leq 256\delta^4 T\leq 32\delta T$ four-cycles have this property, where we used $\delta\leq\frac{1}{2}$.

    In total we exclude for each $\kappa\in\CI$ and each of the nine cases at most $64\delta T$ four-cycles, giving at most $\log(T) \cdot 9 \cdot 64 \delta T = 576 \delta \log(T) T$ excluded four-cycles.
\end{proof}

We now prove \cref{thm:expectation} by showing that $\sigma_A=(A,\kappa_A,x_A,y_A)$ is light for all four-cycles $A$ which are not excluded by the previous three lemmas.

\begin{proof}[Proof of \cref{thm:expectation}]
Recall that in our assignment $\sigma_A = (A,\kappa_A,x_A,y_A)$ of light configurations to light four-cycles $A$, we set
\begin{equation*}
    \kappa_A=\Bigg\lceil \max_{\substack{v\in V(A)\\e\in E(A)\\ \wedge\in W(A)}}\left\{T^{\frac{1}{4}},\delta^{1.5} \frac{t(v)}{\sqrt{T}},\delta \sqrt{t(e)},\delta t(\wedge)\right\}\Bigg\rceil_\CI.
\end{equation*}

For the rest of the proof, we exclude all four-cycles that satisfy the conditions of  \cref{lemma:no2heavyedges}, \cref{lemma:heavyedge_cant_be_in_onions} and \cref{lemma:combination}.

Next, we do a case distinction on which substructure attains the maximum in \cref{eq:w_A}.
For each case, we show that $\sigma_A=(A,\kappa_A,x_A,y_A)$ is a light configuration by showing that each substructure is light unless $A$ is excluded by one of the three previous lemmas.

Before we start the case distinctions, recall that the configuration $\sigma_A$ gives us a labeling function $\ell_{\sigma_A}$ that maps the vertices of $A$ to labels $S_1$, $S_2$, etc.
Further recall that for each labeled substructure $\substr$ of $A$ we now know the probability $p_{(\substr,\kappa_A,\ell_\substr)}$ that the labeled structure $(\substr,\kappa_A,\ell_\substr)$ is realized, where $\ell_\substr = \ell_c\big|_{V(\substr)}$ is the restriction of $\ell_c$ onto the vertices $V(\substr)$ in $\substr$. Thus, in the following, we will make heavy use of the probabilities $p_{(\substr,\kappa_A,\ell_\substr)}$.

We consider the following cases:
\begin{itemize}
    \item Suppose $\kappa_A=\left\lceil \delta t(\wedge) \right\rceil_\CI$, i.e., the heaviest substructure is a wedge~$\wedge$.
    
        First, note that due to the rounding we have $\kappa_A/(2\delta) \leq t(\wedge)\leq \kappa_A/\delta$.
        Next, recall that $x_A$ and $y_A$ are set to the endpoints of the wedge $\wedge$, implying that $\wedge$ is a wedge with $p_{(\wedge,\kappa_A,\ell_\wedge)} = p_{1,\kappa_A}^2 p_{2,\kappa_A}$.
        We now give a comprehensive list of all labeled substructures of $\sigma_A$ and argue why they are light.
        \begin{itemize}
            \item The two wedges $\wedge'\in A$ with 
                $p_{(\wedge',\kappa_A,\ell_{\wedge'})}=p_{1,\kappa_A}^2 p_{2,\kappa_A}$ 
                are light 
                by choice of $\kappa_A$ since $t(\wedge') \leq t(\wedge)$ and above we already showed argued that $t(\wedge) \leq \kappa_A/\delta$.
            \item The two wedges~$\wedge'$ with 
                $p_{(\wedge',\kappa_A,\ell_{\wedge'})}=p_{2,\kappa_A}^2 p_{1,\kappa_A}$ 
                are light 
                because if they were heavy, then they would be $\sqrt{T}/(\kappa_A\delta)$-heavy which implies that $A$ satisfies \cref{lemma:combination} \cref{comb:1a} (since the wedge and $\wedge$ share exactly one edge), but we excluded all such four-cycles above. 
            \item The two nodes~$v$ with 
                $p_{(v,\kappa_A,\ell_v)}=p_{1,\kappa_A}$ 
                are light by choice of $\kappa_A$ since 
                $\kappa_A \geq \left\lceil \delta^{1.5}t(v)/\sqrt{T}\right\rceil_\CI
                \geq \delta^{1.5}t(v)/\sqrt{T}$ and thus after rearranging we get that $t(v) \leq \kappa_A\sqrt{T} / \delta^{1.5}$.
            \item The two nodes~$v$ with 
                $p_{(v,\kappa_A,\ell_v)}=p_{2,\kappa_A}$ 
                are light because if they were heavy, then $A$ would satisfy 
                \cref{lemma:combination} \cref{comb:1b} where we also use the heaviness of~$\wedge$, but we excluded all such four-cycles above.
            \item The four edges $e$ with
                $p_{(e,\kappa_A,\ell_v)}=p_{1,\kappa_A}p_{2,\kappa_A}$ 
                are light because if they were heavy, then $A$ would satisfy \cref{lemma:heavyedge_cant_be_in_onions} where we also use the heaviness of $\wedge$, but we excluded all such four-cycles above.
        \end{itemize}
    \item Suppose $\kappa_A=\left\lceil \delta \sqrt{t(e)} \right\rceil_\CI$, i.e., the heaviest substructure is an edge~$e$.

        First, note that due to the rounding we get that $\kappa_A^2 /(4\delta^2) \leq t(e) \leq \kappa_A^2 / \delta^2$.
        Next, recall that $x_A$ and $y_A$ are the endpoints of the edge $e$ and thus $e$ is an edge with $p_{(e,\kappa_A,\ell_e)} = p_{1,\kappa_A}^2$. Now we get that:
        \begin{itemize}
            \item  The unique edge $e$ with 
                $p_{(e,\kappa_A,\ell_e)} = p_{1,\kappa_A}^2$
                is light by choice of $\kappa_A$ since above we showed that
                $t(e) \leq \kappa_A^2/\delta^2$.
            \item  The unique edge $e'$ with
                $p_{(e',\kappa_A,\ell_{e'})} = p_{\kappa_A,2}^2$
                is light
                because if it were heavy, then $A$ would satisfy
                \cref{lemma:combination} \cref{comb:2a} where we also include the heaviness of $e$, but we excluded all such four-cycles.
                \item  The two edges $e'$ with
                $p_{(e',\kappa_A,\ell_{e'})} = p_{1,\kappa_A} p_{2,\kappa_A}$
                are light because if it were heavy, then $A$ would satisfy
                \cref{lemma:no2heavyedges}, but we excluded all such four-cycles.
            \item  The four wedges $\wedge$ with
                $p_{(\wedge,\kappa_A,\ell_\wedge)} = p_{1,\kappa_A}^2 p_{2,\kappa_A}$
                and
                $p_{(\wedge,\kappa_A,\ell_\wedge)} = p_{1,\kappa_A} p_{2,\kappa_A}^2$, respectively,
                are light
                because if they were heavy then they would satisfy $t(\wedge') > \min\{\kappa_A/\delta, \sqrt{T}/(\kappa_A\delta)\} \geq 1/\delta$ and thus $A$ would satisfy
                \cref{lemma:heavyedge_cant_be_in_onions}, but we excluded all such four-cycles.
            \item  The two nodes with
                $p_{(v,\kappa_A,\ell_v)} = p_{1,\kappa_A}$
                are light by choice of $\kappa_A$ since $\kappa_A\geq \left\lceil \delta^{1.5} \frac{t(v)}{\sqrt{T}} \right\rceil_\CI\geq\delta^{1.5} \frac{t(v)}{\sqrt{T}}$ and thus $t(v)\leq \kappa_A\sqrt{T}/\delta^{1.5}$.
            \item  The two nodes $v$ with
                $p_{(v,\kappa_A,\ell_v)} = p_{2,\kappa_A}$
                are light because if they were heavy, then $A$ would satisfy
                \cref{lemma:combination} \cref{comb:2b}, but we excluded all such four-cycles. Here, we used that the wedge $\wedge$ consisting of $e$ and $v$ is a wedge with $p_{(\wedge,\kappa_A,\ell_\wedge)}=p_{1,\kappa_A}^2 p_{2,\kappa}$ and thus $t(\wedge)\leq \kappa/\delta$.
        \end{itemize}
    \item Suppose $\kappa_A=\left\lceil \delta^{1.5} \frac{t(v)}{\sqrt{T}} \right\rceil_\CI$, i.e., the heaviest substructure is a vertex~$v$, and an there is a $\sqrt{T}/\delta^2$-heavy edge $e\in A$ incident to $v$.
    
        First, note that due to the rounding we have $\kappa_A\sqrt{T}/(2\delta^{1.5}) \leq t(v)\leq \kappa_A\sqrt{T}/\delta^{1.5}$.
        Next, recall that $x_A$ and $y_A$ are the endpoints of the edge $e$ and thus $e$ is an edge with $p_{e,\kappa_A,\ell_e} = p_{1,\kappa_A}^2$.
        Now we get that:
        \begin{itemize}
            \item  The two nodes $v'$ with
                $p_{(v',\kappa_A,\ell_{v'})} = p_{1,\kappa_A}$
                are light since $t(v') \leq t(v)$ and above we showed that $t(v) \leq \kappa_A \sqrt{T} /\delta^{1.5}$.
            \item  The unique edge $e$ with
                $p_{(e,\kappa_A,\ell_e)} = p_{1,\kappa_A}^2$
                is light by choice of $\kappa_A$
                since $\kappa_A\geq \left\lceil \delta \sqrt{t(e)} \right\rceil_\CI\geq \delta \sqrt{t(e)}$ and thus $t(e)\leq \kappa_A^2/\delta^2$.
            \item  The two edges $e'$ with
                $p_{(e',\kappa_A,\ell_{e'})} = p_{1,\kappa_A} p_{2,\kappa_A}$
                are light because if they were heavy, then $A$ would satisfy
                \cref{lemma:no2heavyedges} but we excluded all such four-cycles.
            \item  All four wedges~$\wedge$ with 
                $p_{(\wedge,\kappa_A,\ell_{e})} = p_{1,\kappa_A}^2 p_{2,\kappa_A}$
                and $p_{(\wedge,\kappa_A,\ell_{e})} = p_{1,\kappa_A} p_{2,\kappa_A}^2$, respectively,
                must be light because if they were heavy, then
                then they would satisfy $t(\wedge') > \min\{\kappa_A/\delta, \sqrt{T}/(\kappa_A\delta)\} \geq 1/\delta$
                and thus
                $A$ would satisfy \cref{lemma:heavyedge_cant_be_in_onions}, but we excluded all such four-cycles. 
            \item  The node $v'$ which is adjacent to $v$ with
                $p_{({v'},\kappa_A,\ell_{v'})} = p_{2,\kappa_A}$
                is light                
		since if it were heavy, then $A$ would satisfy
                 \cref{lemma:combination} \cref{comb:3a4a}, but we excluded all such four-cycles. Here, we used that the edge $e'=(v,v')$ with $p_{(e',\kappa_A,\ell_{e'})}=p_{1,\kappa_A}p_{2,\kappa_A}$ is light as showed above.
            \item  The unique node $v'$ which is opposite to $v$ in $A$ with
                $p_{({v'},\kappa_A,\ell_{v'})} = p_{2,\kappa_A}$
                is light
                since if it were heavy, then $A$ would satisfy
                \cref{lemma:combination} \cref{comb:3b}: Here, we used that the wedges are light and thus the onion size is at most $\os(v,v')\leq T/(\kappa_A^2\delta^2)\leq \sqrt{T}/\delta^2$. 
            \item  The unique edge $e'$ with
                $p_{(e',\kappa_A,\ell_{e'})} = p_{2,\kappa_A}^2$
                is light since if it were heavy, then $A$ would satisfy 
                \cref{lemma:combination} \cref{comb:3c}, but we excluded all such cases.
                Here, we used that the wedge formed by $v$ and $e'$ is a wedge $\wedge$ with $p_{(\wedge,\kappa_A,\ell_\wedge)} = p_{1,\kappa_A} p_{2,\kappa_A}^2$ for which above we showed that it is light and thus $t(\wedge) \leq \sqrt{T}/(\kappa_A\delta)$. 
        \end{itemize}
    \item Suppose $\kappa_A=\left\lceil \delta^{1.5} \frac{t(v)}{\sqrt{T}} \right\rceil_\CI$, i.e., the heaviest substructure is a vertex ~$v$, and an there is no $\sqrt{T}/\delta^2$-heavy edge $e\in A$ incident to $v$.

        First, note that due to the rounding we have $\kappa_A\sqrt{T}/(2\delta^{1.5}) \leq t(v)\leq \kappa_A\sqrt{T}/\delta^{1.5}$.
        Let $v'$ be the node opposite to $v$ in $A$.
        Recall that $x_A$ and $y_A$ are the nodes $v$ and $v'$ and thus $v$ and $v'$ are nodes with $p_{(v,\kappa_A,\ell_v)}=p_{(v',\kappa_A,\ell_{v'})}=p_{1,\kappa_A}$.
        Now we get that:
        \begin{itemize}
            \item  The two nodes $v$ and $v'$ with
                $p_{(v',\kappa_A,\ell_{v'})} = p_{1,\kappa_A}$
                are light since $t(v') \leq t(v)$ and above we showed that $t(v) \leq \kappa_A \sqrt{T} /\delta^{1.5}$.
            \item  The nodes $v''$ with
                $p_{(v'',\kappa_A,\ell_{v''})} = p_{2,\kappa_A}$
                are light since if they were heavy, then $A$ would satisfy 
                \cref{lemma:combination} \cref{comb:3a4a}, but we excluded all such four-cycles.
                Here, we used the assumption that there is no $\sqrt{T}/\delta^2$-heavy edge incident upon~$v$.
            \item The two wedges $\wedge'\in A$ with 
                $p_{(\wedge',\kappa_A,\ell_{\wedge'})}=p_{1,\kappa_A}^2 p_{2,\kappa_A}$ 
                are light 
                by choice of $\kappa_A$ since $t(\wedge') \leq t(\wedge)$ and above we already showed argued that $t(\wedge) \leq \kappa_A/\delta$.
            \item The two wedges $\wedge'\in A$ with 
                $p_{(\wedge',\kappa_A,\ell_{\wedge'})}=p_{1,\kappa_A} p_{2,\kappa_A}^2$ 
                are light since if they were heavy, then $A$ would satisfy \cref{lemma:combination} \cref{comb:4b}, but we excluded all such four-cycles above.
            \item The two edges $e'$ with
                $p_{(e',\kappa_A,\ell_{e'})} = p_{1,\kappa_A} p_{2,\kappa_A}$ which are incident to $v$ are light by the case assumption that no $\sqrt{T}/\delta^2$-heavy edge is incident to $v$.
            \item The two edges $e'$ with
                $p_{(e',\kappa_A,\ell_{e'})} = p_{1,\kappa_A} p_{2,\kappa_A}$ which are disjoint to $v$
                are light because if they were heavy, then $A$ would satisfy
                \cref{lemma:combination} \cref{comb:4c}, but we excluded all such four-cycles. Here, we used that the wedge $\wedge$ in $A$ containing $v$ and $e$ is light as shown above.
        \end{itemize}
    \item Suppose $\kappa_A=\left\lceil T^{\frac{1}{4}} \right\rceil$.
                As $p_{1,\kappa_A} = p_{2,\kappa_A}$, the choice of $x_A$ and $y_A$ does not matter, as the sampling probabilities and thus heaviness thresholds do not change.
                Thus, all substructures are light by choice of $\kappa_A$. 
\end{itemize}

This proves that for all four cycles that were not excluded by the three lemmas above, the configuration $\sigma_A$ is light.
\end{proof}

\subsection{Variance Analysis}
\label{sec:variance}

Next, we prove the following lemma, showing that if the input graph contains $T$ four-cycles, then the algorithm finds such a four-cycle with probability at least $9/10$. The main work to prove the lemma will be to bound the variance of the random variable counting the number of four-cycles~$A$ that were sampled with the assigned configurations~$\sigma_A$.

\begin{lemma}
\label{lem:detect}
    If $G$ contains at least $T$ four-cycles, then running \cref{algo:detect} with $\delta\leq \frac{1}{2098\log(T)}$ detects a four-cycle with probability at least $9/10$.
\end{lemma}
\begin{proof}
Let $Z$ be the number of configurations detected by the algorithm. The goal of our proof is to show that $\Prob{Z\geq 1} \geq 9/10$. However, due to the high variance of $Z$, we cannot show this directly, and instead define a random variable $X$ such that (deterministically) $Z\geq X$, and then we show that $\Prob{X\geq 1} \geq 9/10$.

More concretely,  set
\begin{align*}
    \CCdetect=\{\sigma_A\mid A\in\squares,\; \sigma_A \text{ is light}\},
\end{align*}
i.e., $\CCdetect$ contains all assigned configurations~$\sigma_A$ for four-cycles~$A$ such that $\sigma_A$ is light. As per \cref{thm:expectation}, this is the vast majority of the four-cycles as long as $\delta$ is small enough.

Now we set $X_c$ to the binary random variable indicating whether the configuration $c\in\CCdetect$ was realized by \Cref{algo:detect}, and we let $X=\sum_{c\in\CCdetect} X_c$ denote the number of realized configurations from $\CCdetect$.  Observe that $Z\geq X$ since $Z$ might as well count four-cycles that were sampled with configurations that do not count towards $X$. 

Thus, if $G$ contains at least $T$~four-cycles, we can apply Chebyshev's inequality and get that
\begin{align}
\label{eq:proof-detection}
    \Prob{Z\leq 0} &\leq \Prob{X\leq0}=\Prob{\Expectation{X}-X\geq\Expectation{X}}\leq\frac{\Var{X}}{\Expectation{X}^2}.
\end{align} 
Now, the main work will be to bound $\Expectation{X}$ and $\Var{X}$. 

First, let us bound $\Expectation{X}$. Note that each configuration $\sigma_A\in\CCdetect$ contains two nodes that are sampled with probability $p_{1,\kappa}$ and two nodes that are sampled with probability $p_{2,\kappa}$, for a suitable choice of $\kappa$. Thus, each configuration $\sigma_A$ is sampled with probability $p_{1,\kappa}^2 p_{2,\kappa}^2 = \frac{c_1^2 \kappa^2}{\delta^3 T} \frac{c_1^2}{\delta^3 \kappa^2} = \frac{c_1^4}{\delta^6 T} = \tildeO(\frac{1}{T})$, which is independent of $\kappa$. Thus, we set $p=\frac{c_1^4}{\delta^6 T}$ and observe that
 \begin{align*}
     \Expectation{X}=|\CCdetect|p.
 \end{align*}

Next, we have to bound the variance of~$X$, i.e., we have to bound the variance of sampling configurations from $\CCdetect$.
To do this, we need to consider pairs of dependent configurations $c$ and $d$ since
\begin{align*}
    \Var{X} \leq \sum_{c\in\CCdetect}\Var{X_c} + \sum_{c,d\in\CCdetect} \Cov{X_c, X_d}.
\end{align*}

Observe that the first term is bounded by 
\begin{align*}
    \sum_{c\in\CCdetect}\Var{X_c}
    \leq \Expectation{X}
    \leq \frac{1}{20}\Expectation{X}^2,
\end{align*}
since $\Var{X_c}\leq \Expectation{X_c^2}=\Expectation{X_c}$ and $\Expectation{X}\geq 20$
if we pick the constant $c_1$ in $p_{1,\kappa}$ and $p_{2,\kappa}$ large enough.

To bound the second term consisting of the covariances, consider two dependent configurations $c$ and $d$. Note that by \cref{observ:dependent_configurations} this implies that $c$ and $d$ share a labeled substructure. To compute the covariance, we are interested in the largest shared labeled substructure of $c$ and $d$. Therefore we define $c\cap d$ as the unique labeled substructure $(\substr,\kappa,\ell)$ with $(\substr,\kappa,\ell)\in c$, $(\substr,\kappa,\ell)\in d$ and maximum $|V(\substr)|$. Note that $c \cap d$ can also be an opposite node pair.

Assuming that $c$ is realized, $c\cap d$ is realized as well, and thus the probability of realizing $d$ is exactly $p_1^2p_2^2/p_{c\cap d}$. Here, we used that $c\cap d$ is maximal and $p_{c\cap d} = p_{(\substr,\kappa,\ell)}$ if $c\cap d = (\substr,\kappa,\ell)$.

For a labeled substructure $(\substr,\kappa,\ell)$, observe that $(\substr,\kappa,\ell)\in c$ for at most $t(\substr)$ configurations $c$.
We can now bound the covariance as
\begin{align*}
    \sum_{c,d\in\CCdetect} \Cov{X_c, X_d} 
    &=\sum_{c,d\in\CCdetect} \Expectation{X_c X_d} - \Expectation{X_c}\Expectation{X_d} \\
    &\leq \sum_{c,d\in\CCdetect, c\cap d\not=\emptyset} \Expectation{X_c X_d} \\
    &=\sum_{c\in\CCdetect}\Expectation{X_c}\sum_{(\substr,\kappa,\ell)\in c}\sum_{\substack{d\in\CCdetect\\(\substr,\kappa,\ell)=c\cap d}} \Expectation{X_d\mid X_c=1}\\
    &\leq  \sum_{c\in\CCdetect} p \sum_{(\substr,\kappa,\ell)\in c}\sum_{\substack{d\in\CCdetect\\(\substr,\kappa,\ell)=c\cap d}} p/p_{(\substr,\kappa,\ell)}\\
    &\leq \sum_{c\in\CCdetect} p \sum_{(\substr,\kappa,\ell)\in c} t(\substr) p/p_{(\substr,\kappa,\ell)}\\
    &\leq \sum_{c\in\CCdetect} p \sum_{(\substr,\kappa,\ell)\in c} \theta_{(\substr,\kappa,\ell)} p/p_{(\substr,\kappa,\ell)}\\
    &\leq |\CCdetect| p^2   \cdot 14 T/c_1\\
    &\leq \frac{1}{20}|\CCdetect|^2 p^2 \\
    &= \frac{1}{20} \Expectation{X}^2,
\end{align*}
where the sixth step follows from $c$ being light, the seventh step follows from \Cref{obs:relationship-theta-t-p} and the last inequality assumes a large enough constant $c_1$ and $T\leq 2|\CCdetect|$ which we get from \cref{thm:expectation} by choosing $\delta\leq \frac{1}{2098\log(T)}$.

This shows that
\begin{equation*}
    \Var{X}\leq \frac{1}{10}\Expectation{X}^2.
\end{equation*}
 Now, plugging this into \Cref{eq:proof-detection}, we get that
\begin{align*}
    \Prob{Z\leq 0} \leq\frac{\Var{X}}{\Expectation{X}^2}<\frac{1}{10}.
\end{align*} 
This proves that \cref{algo:detect} finds a configuration and thus a four-cycle with probability at least $\frac{9}{10}$.
\end{proof}

\subsection{Proof of \texorpdfstring{\Cref{thm:detect}}{Theorem 5.1}}

\label{sec:proofofthmdetect}
Now we prove \Cref{thm:detect}. We already showed that \Cref{algo:detect} detects a four-cycle with probability at least $9/10$ in \Cref{lem:detect} if we set $\delta=\frac{1}{2098\log(T)}$. 
 Next, we argue the algorithm's space bound.

First, recall that we do not have to store the sampled vertex sets explicitly, but that we can employ hash functions to store them implicitly using space $\tildeO(1)$. Thus, upon seeing a new edge in the stream, we can check whether its endpoints are sampled in one of the sets $S_{1,\kappa}$, $S_{2,\kappa}$, etc.\ by hashing them.

As above we have that $q := p_{1,\kappa}p_{2,\kappa} = \frac{c_1^2}{\delta^3 \sqrt{T}}$ is independent of the choice of $\kappa$. Thus, as we already argued in \Cref{sec:detection-algorithm-description}, storing the collected edges in the first pass takes space $\tildeO(m q)=\tildeO(m/\sqrt{T})$ in expectation.
In the second pass, we expect to find $Tq^2=\tildeO(1)$ four-cycles for each $\kappa\in\CI$. Thus, we detect $\tildeO(1)$~four-cycles in expectation and therefore store another $\tildeO(1)$ edges during the second stream pass.
We conclude that in total the algorithm uses $M := \tilde{O}(m/\sqrt{T})$ expected space. 

Now, observe that by a Markov bound, we get that we use space of more than $10M$, with probability at most $1/10$. Further, the algorithm fails to detect a four-cycle with probability at most $1/10$. Thus, with probability $8/10$, the algorithm detects a four-cycle and uses $\tildeO(m/\sqrt{T})$ space.

To obtain \Cref{thm:detect}, we run the algorithm $O(\log n)$ times in parallel. As soon as any of the instances of the algorithm uses more than $10M$~space, we immediately stop it. Thus, if the input graph contains at least $T$~four-cycles, at least one of the algorithm runs detects a four-cycle with high probability. If the input graph contains no four-cycles, the algorithm would never output a false positive by construction. This proves the theorem.

\section{Counting}
\label{sec:counting}

In this section, we prove the main result of our paper, which we restate from the introduction.

\thmcounting*

Before we describe our algorithm, we give a short overview of how our approach from \Cref{sec:detection} has to be adapted to extend our algorithm to the approximate counting setting.

Recall that in \Cref{sec:heaviness} we described configurations~$\sigma_A$ that we assigned to four-cycles~$A$. Then the crucial part of our analysis was to show that the set 
\begin{align*}
    \CCdetect=\{\sigma_A\mid A\in\squares,\; \sigma_A \text{ is light}\}
\end{align*}
satisfies $(1-\eps)T \leq |\CCdetect| \leq T$. The most difficult part was to show that $\CCdetect$ contains configurations for at least $(1-\eps)T$~four-cycles (see \Cref{thm:expectation}), i.e., to satisfy the first inequality, whereas the second inequality was trivial.
Then, to show that our algorithm detects a four-cycle, we only needed to bound the probability that at least one configuration from $\CCdetect$ is realized.
While the variance of the number of sampled four-cycles is generally high, we proved that the variance of the number of realized configurations from $\CC_{\texttt{Detect}}$ is small.

Thus, the natural approach for extending our detection algorithm to counting four-cycles is to estimate $|\CC_{\texttt{Detect}}|$.
Indeed, we would immediately get a $(1\pm\eps)$-approximate counting algorithm (after adapting our analysis from \cref{sec:detection}), if we could decide the following question during the algorithm: 
\begin{center}
    Given a realized configuration $c$, is $c\in \CCdetect$?
\end{center}

The main problem here is that to check whether $c\in\CCdetect$, we have to verify that the heaviness thresholds that we introduced in \Cref{sec:heaviness} are satisfied for all labeled substructures $(\substr,\kappa,\ell)\in c$. However, this is not possible, as we can show that there exist graphs for which verifying $t(v)\leq\theta_{(v,\kappa,\ell)}$ for some labeled node $(v,\kappa,\ell) \in c$ requires space $\omega(m/\sqrt{T})$. Thus, we have to abandon the idea of checking whether $c\in\CCdetect$.

To fix this, we refine our heaviness notion.
Concretely, now we define a new heaviness notion $t(\substr,\kappa,\ell)$ that replaces the old heaviness notion $t(\substr)$ in the definition of heavy and light substructures. This allows us to check whether substructures are heavy (i.e., whether $t(\substr,\kappa,\ell)\leq \theta_{(\substr,\kappa,\ell)}$) using oracles that only use space $\tildeO(m/\sqrt{T})$.
Now we can define a new set of configurations $\CC$ for which we will be able to answer the following question using one additional stream pass:
\begin{center}
Given a realized configuration $c$, is $c\in \CC$?
\end{center}

Similar to $\CCdetect$, the new set $\CC$ still only contains light configurations, but now it satisfies the relaxed property $(1-\eps)T\leq|\CC|\leq (1+\eps)T$. 
The main obstacle will be to show that our new heaviness notion $t(\substr,\kappa,\ell)$ satisfies all properties that we need to adopt our variance proof from \Cref{sec:variance}. In particular, one of the technical issues we encounter is that a single four-cycle~$A$ might now have multiple configurations~$c\in\CC$ (this is different from $\CCdetect$, where each four-cycle~$A$ had at most one configuration~$\sigma_A\in\CCdetect$), and thus showing that $|\CC|\leq(1+\eps)T$ will be non-trivial.

\smallskip
\paragraph{Outline of the remaining subsections.} We begin by describing our counting algorithm in \Cref{sec:countingsquares}. In \Cref{sec:refinedheaviness}, we introduce the notion of refined heaviness for labeled substructures and analyze its properties. To address the challenge that refined heaviness is no longer monotone, we define the concept of light, locally minimal (LLM) configurations in \Cref{sec:llm}. We then bound the number of LLM configurations in \Cref{sec:boundLLMconfigs}. In \Cref{sec:variance_counting}, we analyze the variance of our estimator and prove \Cref{thm:counting}. Finally, we defer the implementation details of checking whether a sampled configuration is LLM to \Cref{sec:oracles}.

\subsection{Algorithm Description}
\label{sec:countingsquares}
Next, we describe our counting algorithm and present its pseudocode in \Cref{algo:count}.
Due to the high similarity of \Cref{algo:count} to the detection algorithm (\Cref{algo:detect}), here we only describe how the two algorithms differ and refer to \Cref{sec:detection-algorithm-description} for a high-level description of the detection algorithm.

Before we explain \Cref{algo:count}, let us briefly note that the algorithm will make use of \emph{oracles}. In a nutshell (see \Cref{sec:oracles} for details), an oracle takes as input a labeled substructure $(\substr,\kappa,\ell)$ and returns whether $t(\substr,\kappa,\ell) \leq s \theta_{(\substr,\kappa,\ell)}$, where $t(\substr,\kappa,\ell)$ is the new heaviness notion that we will introduce below (replacing $t(\substr)$ from before) and $\theta_{(\substr,\kappa,\ell)}$ is the same heaviness threshold as introduced in \Cref{sec:heaviness}. Here, a minor difference is that for technical reasons we will slightly \emph{shift} the heaviness threshold by some value $s\in[1,2)$ where $s=(1+\frac{1}{L})^i$ for some $L\in\tildeO(1)$ and some random value of $i\in\mathbb{N}$. This allows us to ensure that all our oracle calls return the correct answer with some large constant probability.
We formally describe all guarantees we need from the oracles in \Cref{thm:oracles}, and we describe them in detail in \Cref{sec:oracles}.

Now in \Cref{algo:count}, we note that the overall structure of the algorithm stays the same compared to the detection algorithm. 
Before the stream passes, we sample the same vertex sets as before and some sample additional vertex sets that will be used for constructing our oracles. Further, we pick random shifts $s_1$, $s_2$ and $s_3$ for the heaviness thresholds. 
To obtain a $(1\pm\eps$)-approximate estimate, we run the algorithm for $\delta=\frac{\eps}{2098\log(T)^3}$ and amplify the sampling probabilities by a factor of $1/\delta^2$ compared to \cref{algo:detect}.

Then, during the first two stream passes, the counting algorithm proceeds exactly as the detection algorithm. Additionally, during these two passes, and also in the third stream pass, the algorithm collects additional edges which are necessary for our oracles (see \Cref{sec:oracles} for the details).

After the three stream passes finish, the algorithm iterates over all realized configurations~$c$ (i.e., all four-cycles it detected with the nodes sampled in the correct sets, see \Cref{sec:configurations}). For each such configuration, it checks whether $c\in\CC$. This is done by checking whether $c$ satisfies the LLM property, which we define below (see \Cref{def:LLM}), and this check can be done using the oracles we constructed in the three stream passes. For each $c\in\CC$ the algorithm increments a counter $X$ and eventually returns the estimate $X/p^2$, where $p = \frac{c^2}{\delta^7 \sqrt{T}}$. Here, we note that $p=p_{1,\kappa}p_{2,\kappa}$ for any $\kappa$ and $p^2$ is the probability a four-cycle~$A$ is realized with a configuration $c\in\CC$.

\begin{algorithm}[t]
\caption{Node Sampling for Four-Cycle Counting}
\label{algo:count}
\begin{algorithmic}[1]
\State $\CI = \{ 2^k T^{1/4} \mid k=0,\dots,\lceil \log(T^{1/4}) \rceil \}$
\State Pick random shifts $s_1,s_2,s_3 \in [1,2]$ as described in \Cref{sec:refinedheaviness}
\For{$\kappa \in \CI$}
    \State $p_{1,\kappa} \gets \frac{c \kappa}{\delta^{3.5} \sqrt{T}}$
    \State $p_{2,\kappa} \gets \frac{c}{\delta^{3.5} \kappa}$
    \State Sample $S_{1,\kappa}, R_{1a,\kappa}, R_{1b,\kappa} \subset V$ by sampling each node independently and u.a.r. with probability $p_{1,\kappa}$
    \State Sample $S_{2,\kappa}, R_{2a,\kappa}, R_{2b,\kappa} \subset V$ by sampling each node independently and u.a.r. with probability $p_{2,\kappa}$
    \State Sample node sets for the \IsLLM-oracle (see \Cref{sec:oracles})
\EndFor
\State \textbf{Pass 1:} Collect all edges in 
$E[S_{1,\kappa}, S_{2,\kappa}] \cup E[R_{1b,\kappa}, R_{2a,\kappa}] \cup E[R_{2a,\kappa}, R_{2b,\kappa}] \cup E[R_{2b,\kappa}, R_{1a,\kappa}]$ for all $\kappa$
\State \hspace{0.5em} Sample and collect edge sets for the \IsLLM-oracle (see \Cref{sec:oracles})
\State \textbf{Pass 2:} Collect all edges $(u,v) \in E[R_{1a,\kappa}, R_{1b,\kappa}]$ completing a four-cycle $(u, v, a, b)$ such that
\Statex \hspace{0.5em} $u\in R_{1a,\kappa}$, $v\in R_{1b,\kappa}$, $a\in R_{2a,\kappa}$, $b\in R_{2b,\kappa}$
\State \hspace{0.5em} Collect edge sets for the \IsLLM-oracle
\State \textbf{Pass 3:} Collect edges for the \IsLLM-oracle (see \Cref{sec:oracles})
\State $X \gets 0$
\ForEach{realized configuration $c$}
    \If{\IsLLM($c$), i.e., if oracle says $c \in \mathcal{C}$}
        \State $X \gets X + 1$
    \EndIf
\EndForEach
\State $p \gets \frac{c^2}{\delta^7 \sqrt{T}}$
\State $\hat{T} \gets X / p^2$
\State \Return $\hat{T}$
\end{algorithmic}
\end{algorithm}

\smallskip
\paragraph{Invalid Configurations.}
For technical reasons, we now introduce the notion of valid configurations, as we will need this in our definitions below. 
Formally, a configuration $(A,\kappa,x,y)$ is \emph{invalid} if $e=(x,y)\in E(A)$ and $t(e)\leq\sqrt{T}/\delta^2$, and otherwise it is \emph{valid}. We write $\FC$ to denote the set of all configurations, and we let $\FCvalid$ denote the set of all valid configurations.

Note that configurations $(A,\kappa,x,y)$ with opposite nodes $x,y\in A$ are always valid. Intuitively, valid configurations help us to ensure that if a configuration is realized as described in \Cref{eq:realized-R}, then there exists exactly one $\sqrt{T}/\delta^2$-heavy edge in the configuration. In other words, this ensures that we really only use configurations with two neighboring nodes $x$ and $y$ if the edge connecting the two nodes is $\sqrt{T}/\delta^2$-heavy, akin to the third hard instance in \Cref{sec:challenges-node-sampling}.  

Ensuring the validity of configurations is one of the reasons our algorithm requires three passes over the stream instead of just two. This is due to the fact that we use a special edge-sampling based oracle for this case, which is more accurate than node-sampling based oracles. However, this comes at the expense of one extra stream pass. See \cref{sec:oracleedge} for more details.

\subsection{Refined Heaviness}
\label{sec:refinedheaviness}
Next, we define our new heaviness notion $t(\substr,\kappa,\ell)$, which will enable us to build space-efficient heaviness oracles below.

Let $(\substr,\kappa,\ell)$ be a labeled substructure. By \emph{true heaviness} of $\substr$ we refer to the value $t(\substr)$, i.e., the number of four-cycles containing $\substr$. Note that this value is independent of $\kappa$ and $\ell$.
In contrast, our notion of refined heaviness $t(\substr,\kappa,\ell)$ depends on $\kappa$ and $\ell$.

Before we formally define $t(\substr,\kappa,\ell)$, we remark that the refined heaviness only differs from the true heaviness for nodes and for edges $e$ with $p_{(e,\kappa,\ell)}=p_{2,\kappa}^2$.
Intuitively, in the cases when $t(\substr,\kappa,\ell)$ differs from $t(\substr)$, 
$t(\substr,\kappa,\ell)$ is the number of configurations $c$ such that $(\substr,\kappa,\ell) \in c$ and such that all labeled substructures $(y,\kappa,\ell')\in c$ of larger cardinality than $\substr$ (i.e., $|V(\substr)|<|V(y)|$) are light.

Formally, we define the \emph{refined heaviness} $t(\substr,\kappa,\ell)$ of a labeled substructure $(\substr,\kappa,\ell)$ as follows, where $s_1$ and $s_2$ are the random shifts used in the algorithm and the ordering is picked such that we avoid circular definitions:
\begin{itemize}
\item If $\substr$ is neither a node nor an edge with $p_{(\substr,\kappa,\ell)} = p_{2,\kappa}^2$, then we set $t(\substr,\kappa,\ell)=t(\substr)$.
In other words, if $\substr$ is a wedge or node pair or an edge with $p_{(\substr,\kappa,\ell)} \in \{p_{1,\kappa}^2, p_{1,\kappa}p_{2,\kappa}\}$ then $t(\substr,\kappa,\ell)$ is just its true heaviness $t(\substr)$.

\item If $\substr$ is an edge with $p_{(\substr,\kappa,\ell)} = p_{2,\kappa}^2$, then we define
\begin{align*}
    t(\substr,\kappa,\ell)=|\{c\in \FC\mid (\substr,\kappa,\ell)\in c,\; t(\wedge,\kappa,\ell')\leq s_1\theta_{(\wedge,\kappa,\ell')} \text{ for all wedges } (\wedge,\kappa,\ell')\in c\}|.
\end{align*}

\item If $\substr$ is a node, then we define
 \begin{align*}
    t(\substr,\kappa,\ell)=|\{c\in \FCvalid\mid (\substr,\kappa,\ell)\in c,\; t(\wedge,\kappa,\ell')&\leq s_1\theta_{(\wedge,\kappa,\ell')} \text{ for all wedges } (\wedge,\kappa,\ell')\in c,\\\quad t(e,\kappa,\ell')&\leq s_2\theta_{(e,\kappa,\ell')} \text{ for all edges } (e,\kappa,\ell')\in c\}|.
\end{align*}
\end{itemize}

For this definition, several remarks are in order:
(1)~Although the refined heaviness counts configurations (instead of four-cycles) in the second and third case, it counts at most one configuration per four-cycle (see \Cref{obs:refined-at-most-true}).
(2)~The reason why we were able to set $t(\substr,\kappa,\ell)=t(\substr)$ in first case is that for these substructures it is ``easy'' to build space-efficient oracles that check our heaviness thresholds from \Cref{sec:heaviness}. The main technical challenges arise for the other substructures, where we needed to define $t(\substr,\kappa,\ell)$ differently. 
(3)~In the definition of $t(\substr,\kappa,\ell)$ for nodes, we need to count valid configurations to avoid double-counting four-cycles with two invalid but light configurations.

One of the technical intricacies of this definition is the following: In \Cref{sec:oracles}, we will build oracles that check whether a realized configuration~$c$ is light. This means that we have to check whether $t(\substr,\kappa,\ell) \leq \theta_{(\substr,\kappa,\ell)}$ for all $(\substr,\kappa,\ell)\in c$. Now suppose that $\substr=v$ is a node from $c$. Then our oracle will estimate $t(v,\kappa,\ell)$ by sampling additional configurations~$c'$ that contain $v$. However, now we can only count configurations~$c'$ towards our estimate of $t(v,\kappa,\ell)$ if they satisfy $t(\wedge,\kappa,\ell')\leq s_1 \theta_{(\wedge,\kappa,\ell')}$ for all labeled wedges $(\wedge,\kappa,\ell') \in c'$ and, additionally, $t(e,\kappa,\ell')\leq s_2 \theta_{(e,\kappa,\ell')}$ for all edges $(e,\kappa,\ell')\in c$. However, to check that these conditions are satisfied, we will again need oracles that verify these conditions. Thus, our oracles for checking the lightness of a configuration will (among others) call oracles for node-lightness, which in turn call oracles for edge-lightness, which then further call oracles for wedge-lightness. We present the details of this in \Cref{sec:oracles}.

\smallskip
\paragraph{Heavy Substructures, Configurations and Shifts.}
Next, we define slightly modified notions of heavy substructures and configurations. This is essentially the same as in \Cref{sec:heaviness}, except that now we use the refined heaviness $t(\substr,\kappa,\ell)$ instead of $t(\substr)$ and that we additionally introduce the random shifts $s_1$, $s_2$ and $s_3$ of the algorithm in the thresholds to ensure that our oracles succeed with large constant probability. We note that the values of $\theta_{(\substr,\kappa,\ell)}$ are still exactly the same as defined in \Cref{sec:heaviness}.

Consider a labeled substructure $(\substr,\kappa,\ell)$. Then:
\begin{itemize}
    \item If $\substr$ is a wedge then $(\substr,\kappa,\ell)$ is \emph{heavy} if $t(\substr,\kappa,\ell) > s_1 \theta_{(\substr,\kappa,\ell)}$.
    \item If $\substr$ is an opposite node pair then $(\substr,\kappa,\ell)$ is \emph{heavy} if $t(\substr,\kappa,\ell) > s_1^2 \theta_{(\substr,\kappa,\ell)}$.
    \item If $\substr$ is an edge then $(\substr,\kappa,\ell)$ is \emph{heavy} if $t(\substr,\kappa,\ell) > s_2 \theta_{(\substr,\kappa,\ell)}$.
    \item If $\substr$ is a node then $(\substr,\kappa,\ell)$ is \emph{heavy} if $t(\substr,\kappa,\ell) > s_3 \theta_{(\substr,\kappa,\ell)}$.
\end{itemize}
Otherwise the substructures are called \emph{light}\footnote{
    As in \Cref{sec:heaviness}, we can ignore opposite node pair heaviness as it is dominated by wedge heaviness: Both opposite node pairs of a configuration $c$ are light if its wedges are light as for any node pair $(\substr,\kappa,\ell_\substr)$ we have $t(\substr)\leq \theta_{(\wedge,\kappa,\ell)}^2\leq  \theta_{(\substr,\kappa,\ell_\substr)}$, where $\wedge$ is one of the wedges of $c$ connecting the node pair.
}.

Let us briefly note why we have to introduce the shifts and how we pick them.
While the refined heaviness generally allows us to build small space oracles for checking whether substructures are heavy or not (see \cref{sec:oracles}), these oracles are naturally inaccurate if the heaviness $t(\substr,\kappa,\ell)$ is very close to the heaviness threshold $\theta_{(\substr,\kappa,\ell)}$. For instance, if $t(\substr,\kappa,\ell)$ and $\theta_{(\substr,\kappa,\ell)}$ only differ by some small constant, we have essentially no hope of deciding the heaviness correctly while using small space.
We circumvent this inaccuracy by introducing small shifts to the heaviness thresholds: by multiplicatively shifting all thresholds by a small random real number, we can ensure that with constant probability, all substructures have a refined heaviness $t(\substr,\kappa,\ell)$ that is either at most $(1-\delta')\theta_{(\substr,\kappa,\ell)}$ or at least $(1+\delta')\theta_{(\substr,\kappa,\ell)}$, where $\delta' = \operatorname{poly}(\delta, \frac{1}{\log n})$, i.e., they are ``far away'' from the threshold.

All wedge thresholds are shifted by $s_1=(1+\frac{1}{200L''})^i$ where the exponent $i$ is randomly chosen by the algorithm s.t. $1\leq s_1 < 2$. Here $L''\in \tildeO(1)$ is an upper bound on the number of wedge oracle calls depending on the sampling probabilities and is specified in \cref{sec:oracles}.
All opposite node pair thresholds are shifted by $s_1^2$ accordingly. 
Similarly $1\leq s_2=(1+\frac{1}{200L'})^i<2$ and $1\leq s_3=(1+\frac{1}{200L})^i<2$ shift the edge and node thresholds, where $L,L'\in\tildeO(1)$ are specified in \cref{sec:oracles}.

Finally, a configuration is \emph{heavy} if at least one of its substructures is heavy with respect to the refined heaviness $t(\substr,\kappa,\ell)$, and otherwise it is \emph{light}.

\smallskip
\paragraph{Properties of the Refined Heaviness.}
Next, we state two key properties of the refined heaviness, which will be crucial to obtain our results for counting the number of four-cycles.

First, we observe that $t(\substr)$ is an upper bound on $t(\substr,\kappa,\ell)$.
This observation is crucial for us because it implies that all combinatorial arguments that we used in the proof of \cref{thm:expectation} also hold with respect to the refined heaviness.

\begin{observation}
\label{obs:refined-at-most-true}
    The refined heaviness is always at most the true heaviness, i.e., $$t(\substr,\kappa,\ell)\leq t(\substr).$$
\end{observation}
\begin{proof}
    For wedges and for edges~$e$ with $p_{(e,\kappa,\ell)}\neq p_{2,\kappa}^2$ this follows from the definition because $t(\substr,\kappa,\ell)=t(\substr)$.
    
    Next, consider an edge~$e=(u,v)$ with $p_{(e,\kappa,\ell)} = p_{2,\kappa}^2$.  Let $A=(u,v,a,b)$ be any four-cycle containing $e$.
    Now note that $c=(A,\kappa,a,b)$ is the only configuration of $A$ which contains $(e,\kappa,\ell)$. This is because the $\kappa$-values have to match and $a$ and $b$ need to be labeled $\ell_c(a),\ell_c(b)\in\{R_{1a},R_{2a}\}$ to extend $\ell$ to a labeling of a configuration $\ell_c$.
    Thus, any four-cycle containing $e$ counts at most once towards $t(e,\kappa,\ell)$.

    Finally, consider a node $\substr=v$. Let $A=(v,a,b,d)$ be any four-cycle containing $v$. We again show that at most one configuration of $A$ contains $(v,\kappa,\ell)$ and is valid.
    If $\ell(v)=S_1$ (resp. $\ell(v)=S_2$), there is a unique labeling $\ell_c$ which extends $\ell$ to $V(A)$ and results in the unique configuration $c=(A,\kappa,v,b)$ (resp. $(A,\kappa,a,d)$), which may count towards $t(v,\kappa,\ell)$.
    
    Now consider $\ell(v)\in \{R_{1a},R_{1b},R_{2a},R_{2b}\}$. Let $c$ be any configuration of $A$ which counts towards $t(v,\kappa,\ell)$.
    Then for $c$ to be valid, $A$ needs to contain a $\sqrt{T}/\delta^2$-heavy edge $e$. Furthermore, $A$ can only contain one such edge; otherwise, $c$ would contain a heavy edge. Let $x<y$ be the endpoints of $e$. Then $c$ must be $(A,\kappa,x,y)$ as $\ell_c$ needs to assign $\ell_c(x)=R_{1a}$ and $\ell_c(y)=R_{1b}$ for $c$ to be valid and contain $(v,\kappa,\ell)$. 
    Thus, for each four-cycle containing $v$, at most one of its configuration counts towards $t(v,\kappa,\ell)$, implying $t(v,\kappa,\ell)\leq t(v)$.
\end{proof}

Our second property shows that $t(\substr,\kappa,\ell)$ bounds the number of light configurations containing $(\substr,\kappa,\ell)$. This is crucial to show that \cref{algo:count} has low variance when excluding substructures based on their refined heaviness (see \cref{sec:variance_counting}).
 \begin{observation}
\label{observ:refinedheaviness}
    For any labeled substructure $(\substr,\kappa,\ell)$, at most $t(\substr,\kappa,\ell)$ configurations containing $(\substr,\kappa,\ell)$ are light.
\end{observation}
\begin{proof}
    Consider any configuration $c$ containing $(\substr,\kappa,\ell)$ which does not count towards $t(\substr,\kappa,\ell)$. Then $c$ contains some other substructure $(y,\kappa,\ell')\in c$ of cardinality larger than $\substr$, which is heavy. But $(y,\kappa,\ell')$ being heavy implies that $c$ is heavy.
\end{proof}

\smallskip
\paragraph{(Non-)Monotonicity.}
Next, we discuss a drawback of our refined heaviness notion. For this, recall that in \Cref{sec:heaviness} we defined assigned configurations $\sigma_A=(A,\kappa_A,x_A,y_A)$. We used these assigned configurations to ensure that each four-cycle has at most one assigned configuration and to argue that $|\CCdetect|\leq T$. There, we intuitively chose $\kappa_A$ as the smallest $\kappa\in\CI$ over all light configurations of $A$. Thus, if we wanted to build an oracle for checking whether $c=\sigma_A$ for some configuration $c$ of $A$, this suggests the following strategy: If $c=(A,\kappa,x,y)$ then check if $(A,\kappa/2,x,y)$ is light; if that is the case, then report that $c\neq \sigma_A$, otherwise report that $c=\sigma_A$. Note that for this approach, it is crucial that if we fix $A$, $x$ and $y$, then whether a configuration is light is a monotone function with respect to the parameter~$\kappa$. Note that the same holds for all substructures.

Now the main disadvantage of the refined heaviness definition is that the property of a substructure being heavy is not monotone in $\kappa$ anymore:
Let $(v,\kappa,\ell)$ be a labeled node. The property of $(v,\kappa,\ell)$ being heavy is monotone in $\kappa$ for the true heaviness since the node heaviness $t(v)$ is constant for different $\kappa\in\CI$, while the heaviness threshold~$\theta_{(v,\kappa,\ell)}$ changes monotonically.
However, for the refined heaviness~$t(v,\kappa,\ell)$ we lose this nice property, as now not only the thresholds~$\theta_{(v,\kappa,\ell)}$ but also the heaviness~$t(v,\kappa,\ell)$ itself can vary for different $\kappa$. In particular, the definition of $t(v,\kappa,\ell)$ might include some configurations only for specific values of $\kappa$. An example is presented in \cref{fig:nonmonotone}. 

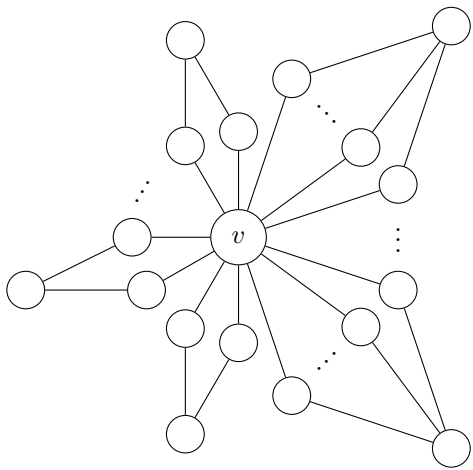
\begin{figure}
    \centering
    \tikzset{customnode/.style={circle, draw, minimum size=13pt, inner sep=5pt}}

\begin{tikzpicture}[scale=0.7]
    \node[customnode] (a) at (0,0) {$v$};
    \node[customnode] (b) at (3,1) {};
    \node[customnode] (c) at (2.3,1.7) {};
    \node[rotate=45] (x3) at (1.55,2.45) {$\vdots$};
    \node[customnode] (d) at (1,3) {};
    \node[customnode] (e) at (4,4) {};
    
    \node[] (x1) at (3,0.1) {$\vdots$};
    
    \node[customnode] (f) at (3,-1) {};
    \node[customnode] (g) at (2.3,-1.7) {};
    \node[rotate=135] (x4) at (1.55,-2.45) {$\vdots$};
    \node[customnode] (h) at (1,-3) {};
    \node[customnode] (i) at (4,-4) {};

    \node[customnode] (j) at (0,2) {};
    \node[customnode] (k) at (-1,3.73) {};
    \node[customnode] (l) at (-1,1.73) {};
    
    \node[customnode] (m) at (0,-2) {};
    \node[customnode] (n) at (-1,-3.73) {};
    \node[customnode] (o) at (-1,-1.73) {};

    \node[rotate=-30] (x2) at (-1.73,1) {$\vdots$};
    
    \node[customnode] (p) at (-1.73,-1) {};
    \node[customnode] (q) at (-4,-1) {};
    \node[customnode] (r) at (-2,0) {};
    
    \draw (a) -- (b) -- (e) -- (c) -- (a) -- (d) -- (e);
    \draw (a) -- (f) -- (i) -- (g) -- (a) -- (h) -- (i);
    \draw (a) -- (j) -- (k) -- (l) -- (a);
    \draw (a) -- (m) -- (n) -- (o) -- (a);
    \draw (a) -- (p) -- (q) -- (r) -- (a);
\end{tikzpicture}
    \caption{Example of a substructure with non-monotone heaviness.
    Here, we consider a node $v$ with label $\ell(v)=S_1$ which is contained in $2T^{3/4}/\delta^{1.5}$ edge-disjoint four-cycles. This makes $v$ heavy for $\kappa=T^{1/4}$ as $t(v,T^{1/4},\ell)>\theta_{(v,T^{1/4},\ell)}=T^{3/4}/\delta^{1.5}$. Additionally, $v$ is endpoint of $T^{1/8}/\delta^{1.5}$ onions of width $\kappa'=T^{3/8}$. Then all four-cycles in the onions do not count towards $t(v,\kappa'/2,\ell)$ but do count towards $t(v,\kappa',\ell)$. As $t(v,T^{3/8},\ell)>\theta_{(v,T^{3/8},\ell)}=T^{7/8}/\delta^{1.5}$, 
    we observe that $v$ is heavy for $\kappa=T^{1/4}$ and $\kappa=\kappa'$ but light for $\kappa=\kappa'/2$.}
    \label{fig:nonmonotone}
\end{figure}

\subsection{LLM Configurations}
\label{sec:llm}

To deal with the drawback that the heaviness of substructures and configurations is not monotone with respect to our refined heaviness, we now introduce the notion of LLM configurations. This will allow us to carefully select the configurations that we count in $\CC$ and in \cref{lemma:nonmonotone}.

Before we give our definition, we need some additional notation. 
Let $A$ be any four-cycle with node set $\{x,y,x',y'\}$.
Now, for any $\kappa\in\CI$, we identify the configuration $(A,\sqrt{T}/\kappa,x,y)$ with the configuration $(A,\kappa,x',y')$. This is useful for the following reasons. Note that $(A,\sqrt{T}/\kappa,x,y)$ is not a configuration according to \cref{def:config} as $\sqrt{T}/\kappa\notin\CI$ since $\sqrt{T}/\kappa < T^{1/4}$. However, if we extend the definitions of the probabilities $p_{1,{\kappa'}}$ and $p_{2,{\kappa'}}$ to arbitrary values $\kappa'$, we observe that
\begin{align*}
    p_{1,\kappa}=p_{2,{\sqrt{T}/\kappa}}
    \quad \text{ and } \quad
     p_{2,\kappa}=p_{1,{\sqrt{T}/\kappa}}.
\end{align*}
Thus a configuration $(A,\kappa,x',y')$ which samples $x'$ and $y'$ with probabilities $p_{1,\kappa}$ and $x$ and $y$ with probability $p_{2,\kappa}$ corresponds exactly to sampling $x$ and $y$ with $p_{1,{\sqrt{T}/\kappa}}$ and $x'$ and $y'$ with $p_{2,{\sqrt{T}/\kappa}}$. Furthermore, $x$ and $y$ are adjacent in $A$ if and only if $x'$ and $y'$ are adjacent in $A$. This shows that this is a natural extension of configurations to values $\sqrt{T}/\kappa\notin\CI$.

Now recall that the detection assignment $\sigma_A$ assigns each light four-cycle a light configuration with minimal~$\kappa$. As we already discussed above for the non-monotonicity of the refined heaviness, we cannot follow this approach any longer. It is possible that some four-cycle $A$ has configurations for which $(A,\kappa,x,y)$ and $(A,\kappa',x,y)$ with $\kappa \neq \kappa'$ are both light. The issue now is that we cannot check whether $(A,\kappa',x,y)$ is light for all $\kappa'<\kappa$ because if $\kappa'\ll \kappa$ then the oracle for checking the heaviness threshold $\theta_{(x,\kappa',\ell)}$ will have a too large space usage and will break our desired space bound. Thus, if we know that $(A,\kappa,x,y)$ is light, we can only check if $(A,\kappa',x,y)$ is light for $\kappa'$-values that are within a factor of $\tildeO(1)$ of $\kappa$ to guarantee that our oracles do not use too much space. This is the intuition for our next definition of locally minimal configurations, where we need to cover several corner cases.

\begin{definition}
\label{def:locally-minimial}
Let $c=(A,\kappa,x,y)$ be a light configuration. Then:
\begin{itemize}
    \item If $\kappa\geq T^{1/4}/\delta^2$, $c$ is \emph{locally minimal} if $(A,\kappa',x,y)$ is heavy for all $\kappa'\in\{\delta^2 \kappa, \dots, \kappa/2\}$.
    \item If $\kappa = T^{1/4}$, then $c$ is \emph{locally minimal} if $(x,y)$ is the lexicographically smallest node pair in the four-cycle.
    \item If $T^{1/4}< \kappa < T^{1/4}/\delta^2$ then we say that $c$ is \emph{locally minimal} if $(A,\kappa',x,y)$ is heavy for all $\kappa'\in\{\max \{\delta^2 \kappa,2\sqrt{T}/\kappa\}, \dots, \kappa/2\}$ and \emph{either} $(A,\sqrt{T}/\kappa,x,y):=(A,\kappa,x',y')$ is heavy as well \emph{or} $(A,\sqrt{T}/\kappa,x,y)$ is light and $(x,y)$ is lexicographically smaller than $(x',y')$.
\end{itemize}
\end{definition}

Note that in the final case of the definition above, we used our extended definition of configurations, where we allow the value $\kappa$ in the configuration to be less than $T^{1/4}$.

Next, we can define LLM configurations, which is the property that four-cycles must satisfy such that they are counted by \Cref{algo:count}.

\begin{definition}
\label{def:LLM}
    A configuration is \emph{LLM} if it is light, locally minimal and valid. We let $\CC$ denote the set of all LLM configurations.
\end{definition}

Note that whether a configuration is LLM is not a deterministic property, since the definition depends on our notion of refined heaviness, which in turn uses the random shifts from the algorithm.

Next, we state a useful observation about light and LLM configurations.
\begin{observation}
    If a four-cycle has a light configuration, then it also has an LLM configuration
\end{observation}
Note that this observation holds because for each light four-cycle $A$, we can pick the light configuration $(A,\kappa,x,y)$ with the smallest $\kappa$-value (exactly as when we defined $\sigma_A$ above). However, now there could also be more LLM configurations for~$A$, i.e., light configurations with non-minimal $\kappa$-values. Fortunately, we show in \Cref{thm:expectation_counting} that this cannot happen for too many four-cycles.

\smallskip
\paragraph{Oracle.}
Given the definitions of LLM configurations and the set~$\CC$, we now consider the guarantees that 
\cref{algo:count} requires for the oracle $\isLLM(\cdot)$, which checks whether $c\in\CC$ for any given configuration $c$.
The proposition below summarizes the oracle's guarantees, and we prove it in \cref{sec:oracles}.
\begin{restatable}{restatableproposition}{thmoracles}
\label{thm:oracles}
    There exists an oracle $\isLLM(\cdot)$ for checking whether a configuration is LLM that, with probability at least~$2/3$,
    answers all $\tildeO(1)$ queries by \cref{algo:count} correctly and uses $\tildeO(m/\sqrt{T})$ space.
\end{restatable}

\subsection{Bounding the Number of LLM Configurations} \label{sec:boundLLMconfigs}
Next, we bound the number of LLM configurations, which will be crucial for the analysis of our estimator below. In the following proposition we show that  $ ||\CC|-T| \leq 1049\delta\log(T)^3 T$.
\begin{proposition}
    \label{thm:expectation_counting}
    The number of LLM configurations is bounded by
    $(1-1049\delta\log(T))T \leq |\CC| \leq (1+ 3\delta\log(T)^3) T$.
\end{proposition}

Now we show that most four-cycles admit an LLM configuration, which proves the first inequality in \cref{thm:expectation_counting}.
\begin{lemma}
    At most $1049\delta\log(T) T$ four-cycles do not admit an LLM configuration.
\end{lemma}
\begin{proof}
    Recall that by \cref{thm:expectation} there are at most $1049\delta\log(T)$ four-cycles that have no configuration in $\CCdetect$.
    Since any configuration in $\CCdetect$ is light with respect to the true heaviness~$t(\substr)$, this also bounds the number of light four-cycles. Now observe that these configurations are also light with respect to the refined heaviness since $t(\substr,\kappa,\ell)\leq t(\substr)$ for all substructures $(\substr,\kappa,\ell)$ by \Cref{obs:refined-at-most-true}.
    Thus, if a four-cycle admits a light configuration, it also admits an LLM configuration. This shows that at most $1049\log(T)\delta T$ four-cycles do not admit an LLM configuration.
\end{proof}

Next, to prove the second inequality in \cref{thm:expectation_counting}, recall that some four-cycles might have multiple LLM configurations in $\CC$ because being heavy is not a monotone property with respect to the refined heaviness. Thus, we say that a four-cycle $A$ is \emph{non-monotone} if there is a LLM configuration $(A,\kappa,x,y)$ and a light configuration $(A,\kappa',x,y)$ with $\kappa'<\kappa$ (and with equal $x,y$). 

Now we state an intermediate result showing that there are few non-monotone four-cycles.
\begin{lemma}
\label{lemma:nonmonotone}
    There exist at most $3\delta \log(T)^3 T$ non-monotone four-cycles.
\end{lemma}
\begin{proof}
We prove the lemma in three steps: 
(1)~We show that any non-monotone four-cycle admits a triple of labeled substructures that induce its non-monotonicity. 
(2)~We show that each triple $Z$ induces at most $\delta\theta_Z$ non-monotone four-cycles for some $\theta_Z$ depending on the triple. 
(3)~We show that there are at most $\log(T)^2T/\theta_Z$ triples for each $\theta_Z$.

\emph{Step 1:} Any non-monotone four-cycle admits three configurations $c_1=(A,\kappa_1,x,y)$, $c_2=(A,\kappa_2,x,y)$, $c_3=(A,\kappa_3,x,y)$ where $\kappa_1\leq\delta \kappa_2\leq\delta \kappa_3$ and $c_1,c_3$ are light and $c_2$ is heavy.
Let $(\substr,\kappa_2,\ell)\in c_2$ be the labeled substructure which is heavy in $c_2$.
We know 
\begin{align*}
        {t}(\substr,\kappa_1,\ell)<\theta_{(\substr,\kappa_1,\ell)},\quad {t}(\substr,\kappa_2,\ell)>\theta_{(\substr,\kappa_2,\ell)},\quad  {t}(\substr,\kappa_3,\ell)<\theta_{(\substr,\kappa_3,\ell)}
\end{align*}

and say the triple $((\substr,\kappa_1,\ell),(\substr,\kappa_2,\ell),(\substr,\kappa_3,\ell))$ \emph{induces} the non-monotonicity of four-cycle $A$.
We know $\substr$ is either a node or a $R_{2a}$-$R_{2b}$-edge as all other substructures are heavy or light based on true heaviness $t(\substr)$, which is independent of $\kappa_1,\kappa_2,\kappa_3$.

\emph{Step 2:} Now observe that a triple $((\substr,\kappa_1,\ell),(\substr,\kappa_2,\ell),(\substr,\kappa_3,\ell))$ induces at most $\min\{\theta_{(\substr,\kappa_1,\ell)}, \theta_{(\substr,\kappa_3,\ell)}\}$ non-monotone four-cycles: 
By \cref{observ:refinedheaviness}, at most $t(\substr,\kappa_1,\ell)$ (resp. $t(\substr,\kappa_3,\ell)$) light configurations contain $(\substr,\kappa_1,\ell)$ (resp. $(\substr,\kappa_3,\ell)$). Thus the triple can induce the non-monotonicity of at most $\min\{t(\substr,\kappa_1,\ell),t(\substr,\kappa_3,\ell)\}\leq \min\{\theta_{(\substr,\kappa_1,\ell)}, \theta_{(\substr,\kappa_3,\ell)}\}$ four-cycles.

Crucially, we have $\min\{\theta_{(\substr,\kappa_1,\ell)}, \theta_{(\substr,\kappa_3,\ell)}\}\leq \delta \theta_{(\substr,\kappa_2,\ell)}$ as $\substr$ is either a node or a $R_{2a}$-$R_{2b}$-edge:
If $\substr$ is a node and $\ell(v)\in\{S_1,R_{1a},R_{1b}\}$, we have 
$$\theta_{(\substr,\kappa_1,\ell)}=\kappa_1\sqrt{T}/\delta^{1.5}\leq\delta \kappa_2\sqrt{T}/\delta^{1.5}=\delta\theta_{(\substr,\kappa_2,\ell)}.$$
If $\substr$ is a node and $\ell(v)\in\{S_1,R_{1a},R_{1b}\}$, we have
$$\theta_{(\substr,\kappa_3,\ell)}=T/(\delta^{1.5}\kappa_3)\leq\delta T/(\delta^{1.5}\kappa_2)=\delta\theta_{(\substr,\kappa_2,\ell)},$$
and if $\substr$ is an edge, we have 
$$\theta_{(\substr,\kappa_3,\ell)}=T/(\delta^{2}\kappa_3^2)\leq\delta^2 T/(\delta^{2}\kappa_2^2)\leq\delta\theta_{(\substr,\kappa_2,\ell)}.$$

\emph{Step 3:} Now we count the number of triples of substructures $((\substr,\kappa_1,\ell),(\substr,\kappa_2,\ell),(\substr,\kappa_3,\ell))$ such that $(\substr,\kappa_2,\ell)$ is heavy. Fix values $\kappa_1,\kappa_2,\kappa_3\in\CI$ and fix some heaviness threshold $\theta\in\{\kappa_2\sqrt{T}/\delta^{1.5},T/(\delta^{1.5}\kappa_2),T/(\delta^2\kappa_2^2)\}$. Observe that there are at most $4T/\theta$ heavy substructures $(\substr,\kappa_2,\ell)$ with $\theta_{(\substr,\kappa_2,\ell)}=\theta$.
Each induces at most $\delta\theta_{(\substr,\kappa_2,\ell)}$ non-monotone four-cycles and thus there are at most $\delta T$ non-monotone four-cycles for this choice of $\kappa_1,\kappa_2,\kappa_3,\theta$. As there are at most $3\log(T)^3$ choices for these parameters, this proves the lemma.
\end{proof}

Using \Cref{lemma:nonmonotone} we can now prove the second inequality of \Cref{thm:expectation_counting}.
\begin{lemma}
\label{lemma:oneLLMconfig}
    At most $3\delta (\log^3 T) T$ four-cycles have more than one LLM configuration.
\end{lemma}
\begin{proof}
    We show that any four-cycle with more than one LLM configuration must also be non-monotone. Then the lemma immediately follows from \Cref{lemma:nonmonotone}. 

    In the proof, we write $\CC_{\text{adj}}\subset \CC$ for the set of all LLM configurations $(A,\kappa,x,y)$ in which $x$ and $y$ are adjacent and $\CC_{\text{opp}}$ for the set of all LLM configurations in which $x$ and $y$ are opposite nodes.

    Let $A$ be a four-cycle with two different LLM configurations $c=(A,\kappa,x,y)$ and $c'=(A,\kappa',x',y')$ with $\kappa\geq \kappa'$. 
    Note that here we make no assumptions on $x$, $y$, $x'$ and $y'$, i.e., it could be that $x=x'$ and $y=y'$, but it could also be that $x=y'$ and $x'=y'$.
    We perform a case distinction based on the values of $\kappa$ and $\kappa'$ and whether $c$ and $c'$ are in $\CC_{\text{adj}}$ or $\CC_{\text{opp}}$:
    \begin{itemize}
        \item Suppose $\kappa=\kappa'=T^{1/4}$. Then only one of $c$ or $c'$ is LLM by definition.
        \item Suppose $\kappa>\kappa'=T^{1/4}$. Note that $c'=(A,T^{1/4},x',y')$ being light implies that $c''=(A,T^{1/4},x,y)$ is light, as corresponding substructures in $c',c''$ have equal heaviness thresholds for $\kappa=T^{1/4}$. 
        Thus $A$ is non-monotone due to $c$ being LLM and $c''$ being light.
        \item Suppose $\kappa,\kappa'>T^{1/4}$. Now consider the following cases:
        \begin{itemize}
            \item  Suppose both $c,c'\in \CC_{\text{adj}}$. Then both $(x,y)$ and $(x',y')$ are $\sqrt{T}/\delta^2$-heavy-edges as otherwise $c$ or $c'$ are not valid and thus not locally minimal. Thus, we have $x=x'$ and $y=y'$ since a light four-cycle can only have one $\sqrt{T}/\delta^2$-heavy edge. As $c$ and $c'$ are LLM configurations with the same node pair $(x,y)$, $A$ is non-monotone.
            \item Suppose $c\in \CC_{\text{opp}}$ and $c'\in\CC_{\text{adj}}$.
            This cannot happen because $c'$ contains a $\sqrt{T}/\delta^2$-heavy edge $(x',y')$ (since $c'\in\CC_{\text{adj}}$), which would imply that $c$ is not light. 
            \item Suppose $c'\in \CC_{\text{opp}}$ and $c\in\CC_{\text{adj}}$. This cannot happen for the same reason as above (but we note that it is a different case since we assume that $\kappa\geq\kappa'$).
            \item Suppose $c,c'\in \CC_{\text{opp}}$. If $x=x'$ and $y=y'$, then this implies that $A$ is non-monotone. Otherwise, $\{x,y\}\cap \{x',y'\}=\emptyset$ and we have $(A,\kappa',x',y')=(A,\sqrt{T}/\kappa',x,y)$.             
            Now if $\sqrt{T}/\kappa'<\delta^2\kappa$, then $A$ is non-monotone.
            Otherwise, observe that we get the series of inequalities $\delta^2 \kappa \leq \sqrt{T}/\kappa' < T^{1/4} \leq \kappa$ and we also have that $\sqrt{T}/\kappa\leq \sqrt{T}/\kappa'$ since we assume that $\kappa\geq \kappa'$. Thus $\max\{\delta^2\kappa,\sqrt{T}/\kappa\}\leq\sqrt{T}/\kappa'$
            implying that $c$ is locally minimal only if $c'$ is heavy or not locally minimal. 
            Thus, only one of $c$ and $c'$ can be LLM, contradicting the assumption that both of them are LLM.
        \end{itemize}
    \end{itemize}
\end{proof}

The previous lemmas together prove \cref{thm:expectation_counting}.

\subsection{Variance Analysis and Proof of \texorpdfstring{\Cref{thm:counting}}{Theorem 1.1}}
\label{sec:variance_counting}

Now we analyze the variance of our estimator similar to \Cref{sec:variance}, but this time using our notion of refined heaviness. Throughout the section, we condition on the fact that all oracle calls to the $\IsLLM(\cdot)$ procedure are answered correctly, which happens with probability at least $2/3$ by \cref{thm:oracles}.

First, note that similar to \Cref{obs:relationship-theta-t-p}, we obtain that the sampling probabilities $p_{1,\kappa}$ and $p_{2,\kappa}$, which we use for \Cref{algo:count}, satisfy the following inequality.
\begin{observation}
\label{obs:relationship-theta-t-p-counting}
    For any labeled substructure $(\substr,\kappa,\ell)$, it holds that
    $\theta_{(\substr,\kappa,\ell)}<\delta^2Tp_{(\substr,\kappa,\ell)}/c_1$.
    \label{observ:heaviness2}
\end{observation}

Next, let $X_c$ be the binary random variable indicating whether configuration $c\in\CC$ was realized.
Let $X=\sum_{c\in\CC} X_c$ and as in the algorithm set $p = \frac{c^2}{\delta^7 \sqrt{T}}$.

Now, as before in \Cref{sec:variance}, we get that 
\[\Expectation{X}=|\CC|p^2,\]
and 
\[\Var{X} \leq \sum_{c\in\CC}\Var{X_c} + \sum_{c,d\in\CC} \Cov{X_c, X_d}.\]
Further, we again have
\begin{align*}
    \sum_{c\in\CC}\Var{X_c}\leq \Expectation{X}\leq \frac{\delta^2}{20}\Expectation{X}^2,
\end{align*}
since $\Expectation{X}>20/\delta^2$. 

Next, we bound the covariances analogously to \cref{sec:variance}, but using refined heaviness instead of true heaviness:
\begin{align*}
    \sum_{c,d\in\CC} \Cov{X_c, X_d} 
    &\leq  \sum_{c\in\CC}p^2\sum_{(\substr,\kappa,\ell)\in c}\sum_{\substack{d\in\CC\\(\substr,\kappa,\ell)=c\cap d}} p^2/p_{(\substr,\kappa,\ell)}\\
    &\leq \sum_{c\in\CC} p^2 \sum_{(\substr,\kappa,\ell)\in c} t(\substr,\kappa,\ell) p^2/p_{(\substr,\kappa,\ell)}\\
    &\leq \sum_{c\in\CC} p^2 \sum_{(\substr,\kappa,\ell)\in c} s_1^2s_2s_3\theta_{(\substr,\kappa,\ell)} p^2/p_{(\substr,\kappa,\ell)}\\
    &\leq |\CC| p^4   \cdot 14 s_1^2s_2s_3 T/c\\
    &\leq \frac{1}{20}|\CC|^2 p^4,
\end{align*} 
where the first step is analogous to \Cref{sec:variance}. The second step follows from \Cref{observ:refinedheaviness}. The third step follows from the fact that all $c\in\CC$ are light and the fact that 
$s_1^2s_2s_3$ is an upper bound on the largest shift $\max\{s_1,s_1^2,s_2,s_3\}$. The fourth step follows from \Cref{obs:relationship-theta-t-p-counting}.
 The final step follows from picking $c$ large enough and \cref{thm:expectation_counting}.

This shows
\begin{equation*}
    \Var{X}\leq \frac{1}{10}\Expectation{X}^2.
\end{equation*}

We can now bound the success probability of \cref{algo:count} using a Chebyshev bound.
\begin{lemma}
    Conditioned on the event that all oracle calls $\IsLLM(\cdot)$ are answered correctly, we have $\Prob{|\hat{T}-T|\geq 1049\delta\log(T)^3 T}\leq \frac{1}{10}$.
    \label{lemma:accuracy}
\end{lemma}
\begin{proof}
Recall that $\hat{T}=X/p^2$.
    By Chebyshev's inequality, we obtain that
    \begin{align*}
        \Prob{|\hat{T}-\Expectation{\hat{T}}|\geq {\delta} \Expectation{\hat{T}}}&\leq \frac{\Var{\hat{T}}}{\delta^2\Expectation{\hat{T}}^2} = \frac{\Var{X}}{\delta^2\Expectation{X}^2} 
        \leq \frac{1}{10}.
    \end{align*}
    Since we have that
    $$\Expectation{\hat{T}}=|\CC|$$ and since
    $$(1-1049\log(T)\delta)T \leq |\CC| \leq (1+ 3\log(T)^3\delta) T$$
    by \cref{thm:expectation_counting}, 
    this proves the claim.
\end{proof}

We can now prove \Cref{thm:counting}.
\begin{proof}[Proof of \Cref{thm:counting}]
    We run \cref{algo:count} for $\delta=\frac{\eps}{2098\log(T)^3}$. Here $\eps$ is the desired approximation guarantee of the algorithm estimate.
    The space bound follows from \cref{thm:detect,thm:oracles} as resetting $\delta$ only adds logarithmic factors.
    The oracle answers all queries correctly with probability $\frac{2}{3}$.
    By \cref{lemma:accuracy} the algorithm computes a $(1\pm{1049\delta\log(T)^3})$-approximate solution with probability at least $\frac{9}{10}$, which gives a success probability of $\frac{3}{5}$ for a single run of \cref{algo:count}. 
    Further, note that above we bounded the algorithm's space usage by $M=\tildeO(m/\sqrt{T})$ in expectation. Observe that by a Markov bound, the algorithm uses space more than $10M$ with probability at most $\frac{1}{10}$.
    
    Now, if we run the algorithm $O(\log n)$ times independently and stop each instance whenever it exceeds space $10M$, we can apply the median trick to obtain a $(1\pm\eps)$-approximate answer with high probability.
\end{proof}

\section{Oracles}
\label{sec:oracles}

In this section, we design an oracle for checking whether a sampled configuration is LLM, and we prove \cref{thm:oracles}, which we restate here for convenience.

\thmoracles*

We note that our oracle for checking whether a configuration is LLM uses heaviness oracles to classify whether substructures are heavy or light as subroutines. We start by explaining the LLM oracle, before we proceed to explain the heaviness oracles for the substructures.

\smallskip
\paragraph{LLM Oracle.}
The procedure to check whether a given configuration $c$ is LLM, i.e., whether $c\in \CC$, is given in \cref{algo:oracle}. It checks whether $c=(A,\kappa,x,y)$ is LLM analogously to the definition of LLM configurations (see \Cref{def:LLM}).
Recall that intuitively a configuration $(A,\kappa,x,y)$ is LLM if it is valid and light, and all $(A,\kappa',x,y)$ with $\kappa> \kappa'\geq\delta^2 \kappa$ are heavy. Thus, the oracle checks all of these conditions from the definition by calling other oracles for these specific properties. We note that the pseudocode in \Cref{algo:oracle} is stated for the case of $\kappa\geq T^{1/4}/\delta^2$, and that for other values of $\kappa$, the concrete intervals for the values of $\kappa'$ differ (see \Cref{def:locally-minimial}). However, the other cases work analogously and we omit this for the sake of brevity.

Note that the LLM oracle uses the following two subroutines:
\begin{itemize}
\item To check whether $(A,\kappa,x,y)$ is light, the function $\isHeavy(\cdot)$ checks whether any substructure of $(A,\kappa,x,y)$ is heavy. For this task, $\isHeavy(\cdot)$ calls the function $\oracle(\substr,\kappa,\ell)$ which returns whether the substructure $(\substr,\kappa,\ell)$ is heavy or light.
\item If $e=(x,y)$ is an edge of $A$, then, to check whether $(A,\kappa,x,y)$ is valid, we employ an oracle $\oracle_{\text{valid}}(\cdot)$ which checks whether $t(e)>\sqrt{T}/\delta^2$ .
\end{itemize}

In the following, we denote the calls of the LLM oracle to other oracles as \emph{heaviness oracle calls} and obtain the following observation.

\begin{observation}
\label{observ:LLMoracle}
    If all heaviness oracle calls during \cref{algo:oracle} are answered correctly, then $\isLLM(c)$ correctly returns whether $c\in\CC$.
\end{observation}

\begin{algorithm}[t]
\caption{LLM Oracle for $\kappa \geq T^{1/4}/\delta^2$}
\label{algo:oracle}
\begin{algorithmic}[1]
\Function{$\isLLM(A,\kappa,x,y)$}
\If{\texttt{IsHeavy}($A, \kappa, x, y$)}
    \State \Return \texttt{false}
\EndIf
\ForEach{$\kappa' \in \{\delta^2 \kappa, \dots, \kappa/2\}$}
    \If{not \texttt{IsHeavy}($A, \kappa', x, y$)}
        \State \Return \texttt{false}
    \EndIf
\EndForEach
\If{$(x, y) \in E(A)$ \textbf{and} $\oracle_{\text{valid}}(\{x, y\}) = \invalid$}
    \State \Return \texttt{false}
\EndIf
\State \Return \texttt{true}
\EndFunction
\Function{$\isHeavy(c = (A, \kappa, x, y))$}
\ForEach{labeled node, edge, or wedge $(\substr, \kappa, \ell) \in c$}
    \If{$\oracle(\substr, \kappa, \ell) = \heavy$}
        \State \Return \texttt{true}
    \EndIf
\EndForEach
\State \Return \texttt{false}
\EndFunction
\end{algorithmic}
\end{algorithm}

\paragraph{Heaviness Oracles.}
In the rest of this section, we build the heaviness oracles which are used by \cref{algo:oracle}.
Given any labeled substructure $(\substr,\kappa,\ell)\in c$ as a query, the oracle should decide whether $t(\substr,\kappa,\ell)>s\theta_{(\substr,\kappa,\ell)}$, where $s\in\{s_1,s_2,s_3\}$ is the appropriate shift from \Cref{sec:refinedheaviness}. We denote the oracle's answer, which is either \light or \heavy, by $\oracle(\substr,\kappa,\ell)$.

All of our heaviness oracles work by sampling a small subgraph and estimating $t(\substr,\kappa,\ell)$ by counting the number of configurations in the sample that count towards $t(\substr,\kappa,\ell)$. We do not sample a new subgraph for each oracle query; rather, we generate one sample to answer all oracle calls. Storing the samples allows us to compute the oracle values after all stream passes finish, which has the advantage that our oracle's answers are consistent: The same query is answered identically each time.

To obtain the sample, we use two different strategies:
For edges $(e,\kappa,\ell)$ with $p_{(v,\kappa,\ell)}=p_{1,\kappa}p_{2,\kappa}$ we use edge-sampling based oracles.
For all other substructures we use node-sampling based oracles which sample nodes similarly to the main algorithm.

Next, recall that by definition of the refined heaviness for nodes, a configuration may or may not count towards the node heaviness depending on the heaviness of its edges and wedges. Thus, naturally, our node heaviness oracles will depend on the edge and wedge heaviness oracles and call them internally to ensure the node heaviness is estimated correctly.

We use the following standard variants of the Chernoff and Chebyshev bounds for our analysis.
\begin{lemma}[Chernoff bound]
    \label{lemma:chernoff}
    Let $X_1,\dots,X_n\in[0,B]$ be independent random variables,
    and $X=\sum_{i=1}^n X_i$. If $M_L\leq \Expectation{X}\leq M_H$, then for all $\delta\in(0,\frac{1}{2})$:
    \begin{align*}
        \IP[X\geq M_H/(1-\delta)]&\leq \exp\left(-\frac{1}{12B}\delta^2M_H\right), \text{ and} \\
        \IP[X\leq M_L/(1+\delta)]&\leq \exp\left(-\frac{1}{12B}\delta^2M_L\right).
    \end{align*}
\end{lemma}
\begin{lemma}[Chebyshev bound]
    \label{lemma:chebyshev}
    Given any random variable $X\geq 0$, with $M_L\leq \Expectation{X}\leq M_H$. Then for all $\delta\in(0,\frac{1}{2})$,
    \begin{align*}
        \IP\left[X \geq \frac{1}{1 - \delta} M_H \right] &\leq \frac{4\Var{X}}{\delta^2 M_H^2}, \text{ and} \\
        \IP\left[X \leq \frac{1}{1 + \delta} M_L \right] &\leq \frac{4 \Var{X}}{\delta^2 \Expectation{X}^2}.
    \end{align*}
\end{lemma}

\subsection{Node Oracles}
\label{sec:oraclenode}

We now design an oracle which, when queried on a labeled node $(v,\kappa,\ell)$ during \cref{algo:oracle}, outputs 
\begin{center}
    $\oracle(v,\kappa,\ell)=\heavy$ if and only if $t(v,\kappa,\ell)>s_3\theta_{(v,\kappa,\ell)}$
\end{center}
with a large constant probability.
Our oracle uses node sampling to detect four-cycles in $\squares_v$, the set of four-cycles containing $v$, and then employs heaviness oracles for edges and wedges to decide whether a found four-cycle counts towards $t(v,\kappa,\ell)$. For the rest of this section, we condition on the event that the edge and wedge oracles are \emph{perfect}, i.e., they answer correctly all queries.

Next, note that the LLM oracle only calls the node oracle for nodes that are part of realized configurations of \cref{algo:count}. For each such realized configuration, the node oracle is called inside \isLLM in \cref{algo:oracle} at most $1+\log(1/\delta^2)$ times.
This is crucial, as it implies that we only need to correctly answer $4\log(1/\delta^2)|\FC|p^2=\tildeO(1)$ many queries in expectation and thus at most 
\begin{align}
\label{eq:upper-bound-queries}
    L:=100\cdot4\log(1/\delta^2)T\log(T)p^2
\end{align}
queries with probability $99/100$.

\begin{lemma}
    \label{lemma:oraclenode}
    Assuming access to perfect edge heaviness and wedge heaviness oracles,
    there is a node heaviness oracle which correctly answers all at most $L=\tildeO(1)$ queries by \cref{algo:oracle} with probability at least $81/100$. The oracle requires $\tildeO(m/\sqrt{T})$ space.
\end{lemma}
\smallskip
\paragraph{Oracle Construction and Integration into \Cref{algo:count}.}
Next, we describe how the oracle is constructed.
The oracle receives all node sets sampled by \cref{algo:count} as input before the first stream pass. After the third pass, it receives a set of queries $Z$ where each query is a labeled node for which heaviness needs to be checked during \cref{algo:oracle}. 

The oracle samples, similarly to the main algorithm, node sets  $S_{1,\kappa}',R_{1a,\kappa}',R_{1b,\kappa}',S_{2,\kappa}',R_{2a,\kappa}',R_{2b,\kappa}'$ with probabilities 
$$p_{1,\kappa}'=c_2(600L)^2p_{1,\kappa},\quad p_{2,\kappa}'=c_2(600L)^2p_{2,\kappa},$$
respectively, where $c_2=\tildeO(1)$ will be set later. 

During the first pass, while the specific set of queries $Z$ is unknown, the oracle collects all edges in $E[S_{1,\kappa'},S_{2,\kappa}']$, $E[S_{2,\kappa'},S_{1,\kappa}']$ and $E[S_{2,\kappa}',S_{1,\kappa}']$ for all $\kappa,\kappa'\in\CI$ with $\delta^2 \kappa'\leq \kappa\leq \kappa'$. These edge sets are used to answer queries on nodes with labels $S_1$ or $S_2$. Analogously, during the first stream pass it also collects
$E[R_{2a,\kappa'}\cup R_{2a,\kappa}',R_{2b,\kappa}'\cup R_{1b,\kappa}']$, 
$E[R_{2b,\kappa'}\cup R_{2b,\kappa}',R_{2a,\kappa}'\cup R_{1a,\kappa}']$,
$E[R_{1a,\kappa'}\cup R_{1a,\kappa}',R_{2b,\kappa}']$, 
$E[R_{1b,\kappa'}\cup R_{1b,\kappa}',R_{2a,\kappa}']$ for all $\kappa,\kappa'\in\CI$ with $\delta^2 \kappa'\leq \kappa\leq \kappa'$.

During the second stream pass, the edges in $E[R_{1a,\kappa'},R_{1b,\kappa}']$, $E[R_{1a,\kappa}',R_{1b,\kappa}']$ or $E[R_{1b,\kappa'},R_{1b,\kappa}']$ which complete four-cycles with one node from $R_{2a,\kappa'}$ and one from $R_{2b,\kappa}'$ are collected in the second stream pass. This is done exactly as in \Cref{algo:count}.
Here, we note that collecting the edges for the node oracle only requires two passes. The third pass in \cref{algo:count} is necessary for building the edge oracles.

After the third stream pass, the oracle is queried for the set $Z$ of labeled nodes.
For each query $(v,\kappa,\ell)\in Z$ it estimates $t(v,\kappa,\ell)$ by checking whether any four-cycle using nodes sampled by the oracle, corresponds to a configuration which counts towards $t(v,\kappa,\ell)$. We give the details for the cases $\ell(v)=S_1$ and $\ell(v)=R_{2a}$, and note that all other cases follow analogously.

\emph{Case 1: $\ell(v)=S_1$.}
The oracle considers all four-cycles $(v,a,b,d)$ where $b\in S_{1,\kappa}'$ and $a,d\in S_{2,\kappa}'$.
Each such four-cycle corresponds to a configuration $c=(A,\kappa,v,b)$ with $(v,\kappa,\ell)\in c$. We say the oracle \emph{realized} the configuration\footnote{We note that, formally speaking, our definitions for labeled substructures and realized configurations were only defined for the sets $S_1$, $S_2$, etc.\ and not for the sets $S_1'$, $S_2'$, etc.\ However, we could define the same for $S_1'$, $S_2'$, etc.\ analogously and thus continue to use the notation based on $S_1$, $S_2$, etc.} $c'$.
Observe that, conditioned on $v\in S_{1,\kappa}$, this is precisely how configurations containing $(v,\kappa,\ell)$ are realized in \cref{algo:count} up to lower order terms in the sampling probabilities.

For each such configuration $c=(A,\kappa,v,b)$ realized by the oracle, we check whether $c$ counts towards $t(v,\kappa,\ell)$, i.e., whether edges and wedges in $c$ are light. 
For this, we assume access to perfect edge and wedge heaviness oracles. 
Let $q(v)$ denote the number of configurations realized by the oracle that count towards $t(v,\kappa,\ell)$.
We then define the oracle answer as
\begin{equation*}
    \oracle(v,\kappa,\ell)=\begin{cases}
        \light & \text{ if } q(v)\leq p_{1,\kappa}'(p_{2,\kappa}')^2s_3\theta_{(v,\kappa,\ell)},\\
        \heavy & \text{ if } q(v)> p_{1,\kappa}'(p_{2,\kappa}')^2s_3\theta_{(v,\kappa,\ell)},
    \end{cases}
\end{equation*}
where $\theta_{(v,\kappa,\ell)}=\frac{\kappa\sqrt{T}}{\delta^{1.5}}$ is the heaviness threshold for $(v,\kappa,\ell)$.

\emph{Case 2: $\ell(v)=R_{2a}$.}
The oracle considers all four-cycles $A=(v,a,b,d)$ where $a\in R_{2b}'$, $b\in R_{1a}'$ and $d\in R_{1b}'$. Each such four-cycle corresponds to a configuration $c=(A,\kappa,b,d)$ with $(v,\kappa,\ell)\in c$ which is realized by the oracle.
Again, observe that, conditioned on $v\in R_{2a,\kappa}$, this is exactly how configurations containing $(v,\kappa,\ell)$ are realized in \cref{algo:count}.

For each such configuration $c=(A,\kappa,b,d)$ realized by the oracle, we check whether $c$ counts towards $t(v,\kappa,\ell)$, i.e., whether ist valid and its edges and wedges in $c$ are light. 
To check whether edges and wedges are light, we assume perfect edge and wedge oracles for now. To check whether $c$ is valid, we use an edge oracle to check the heaviness of $(b,d)$ as well.
Let $q(v)$ denote the number of such configurations realized by the oracle that count towards $t(v,\kappa,\ell)$.
We then define the oracle answer as
\begin{equation*}
    \oracle(v,\kappa,\ell)=\begin{cases}
        \light & \text{ if } q(v)\leq p_{2,\kappa}'(p_{1,\kappa}')^2s_3\theta_{(v,\kappa,\ell)},\\
        \heavy & \text{ if } q(v)> p_{2,\kappa}'(p_{1,\kappa}')^2s_3\theta_{(v,\kappa,\ell)},
    \end{cases}
\end{equation*}
where $\theta_v=\frac{T}{\delta^{1.5}\kappa}$ is the heaviness threshold for $(v,\kappa,\ell)$.

\emph{Other cases:}
The oracle answers for all other node labelings are computed analogously.

\bigskip
We note that for any $(v,\kappa,\ell)\in Z$ we know that $v$ was sampled by \cref{algo:count}.
Thus the edges that the oracle collects allow it to answer (1)~the heaviness query for any node node $(v,\kappa',\ell)$ sampled by \cref{algo:count} and (2)~all heaviness queries $(v,\kappa,\ell)$ for $\delta^2\kappa'\leq \kappa<\kappa'$ required in $\isLLM(\cdot)$. 

We stress that while the node oracle collects edges to be able to answer heaviness queries for all nodes sampled by \cref{algo:count}, our accuracy analysis will rely on the fact that the oracle is only queried for $|Z|\leq L\in\tildeO(1)$ nodes.

We further note that for nodes with label $S_1$ or $S_2$, all relevant edges are collected by the oracle in the first pass, and thus the oracle can answer such queries after Pass 1.

\smallskip
\paragraph{Analysis of the Oracle.}
Next, we provide our guarantees for the oracle. 
The proof again depends on the label $\ell(v)$ of the query $(v,\kappa,\ell)$. 
We give the details for the case $\ell(v)=S_1$,
and note that all other cases follow analogously.

Suppose $\ell(v)=S_1$.
First, observe that $\Expectation{q(v)}=t(v,\kappa,\ell) p_{1,\kappa}'(p_{2,\kappa}')^2$ as any configuration counting towards $t(v,\kappa,\ell)$ is realized by the oracle with probability $p_{1,\kappa}'(p_{2,\kappa}')^2$.

Let $\epsoracleV=1/(600L)$. We obtain the following accuracy guarantee for the node oracle.
\begin{lemma}
    \label{lemma:oraclenodesaccuracyS1}
    With probability $96/100$, the node oracle classifies all nodes $(v,\kappa,\ell)$ with $\ell(v)=S_1$ and $t(v,\kappa,\ell)\notin[(1-\epsoracleV)s_3\theta_{(v,\kappa,\ell)},(1+\epsoracleV) s_3\theta_{(v,\kappa,\ell)}]$ correctly.
\end{lemma}
\begin{proof}
    We apply a variance analysis similar to \cref{sec:variance_counting}.
    Fix a labeled node $(v,\kappa,\ell)$. We condition on the event that $\oracle(v,\kappa,\ell)$ is called.
    Let $\FC_{(v,\kappa,\ell)}$ be the set of configurations which count towards $t(v,\kappa,\ell)$.
    Let $Y_c$ be the binary random variable indicating whether the oracle realized the configuration $c\in \FC_{(v,\kappa,\ell)}$.
    Then $q(v)=\sum_{c\in \FC_{(v,\kappa,\ell)}} Y_c$.
    
    Note that for any two configurations $c,d\in\FC_{(v,\kappa,\ell)}$, the intersection $c\cap d$ is always at least the node $(v,\kappa,\ell)$. Furthermore, $Y_c$ and $Y_d$ are dependent only if $c\cap d$ is an edge or wedge, as we condition on $\oracle(v,\kappa,\ell)$ being called.
    To account for this conditioning, 
    we also define $p'_{(\substr,\kappa,\ell_\substr)}$ as the probability that the substructure $(\substr,\kappa,\ell_\substr)$ is \emph{realized by the oracle}. For example, an edge $((v,u),\kappa,\ell_{(v,u)})$ with $\ell(v)=S_1$, $\ell(u)=S_2$ is realized by the oracle if $u\in S_{2,\kappa}'$ and thus with probability $p'_{((v,u),\kappa,\ell_{(v,u)})}=p_{2,\kappa}'$.
    
    This implies $p_{(\substr,\kappa,\ell_\substr)}\leq \epsoracleV^2p'_{(\substr,\kappa,\ell_\substr)}p_{1,\kappa}/c_2$ and thus \begin{equation}
        \theta_{(\substr,\kappa,\ell_\substr)}\leq T p_{(\substr,\kappa,\ell_\substr)}/c_1\leq \epsoracleV^2T p'_{(\substr,\kappa,\ell_\substr)}p_{1,\kappa}/(c_1c_2)\leq\epsoracleV^2\theta_{(v,\kappa,\ell)} p'_{(\substr,\kappa,\ell_\substr)}/c_2 \label{eq:nodeoracle}
    \end{equation}
    by \cref{observ:heaviness2}.
    
    Next, we have
    \begin{align*}
        \Var{q(v)} \leq \sum_{c\in \FC_{(v,\kappa,\ell)}}\Var{X_c} + \sum_{c,d\in \FC_{(v,\kappa,\ell)}} \Cov{X_c, X_d},
    \end{align*}
    and
    $$\sum_{c\in\FC_{(v,\kappa,\ell)}}\Var{Y_c}\leq \Expectation{q(v)}\leq \frac{\epsoracleV^2}{c_2}\Expectation{q(v)}^2$$ 
    as $\Expectation{q(v)}>c_2/\epsoracleV^2$. We bound the covariances analogously to \cref{sec:variance_counting}, but using the adapted definition of $p'_{(\substr,\kappa,\ell_\substr)}$:

    \begin{align*}
        \sum_{c,d\in \FC_{(v,\kappa,\ell)}} \Cov{Y_c, Y_d} 
        &\leq  \sum_{\substack{c\in\FC_{(v,\kappa,\ell)}}}\Expectation{Y_c}\sum_{\substack{(\substr,\kappa,\ell_\substr)\in c\\x\neq v}}\sum_{\substack{d\in\FC_{(v,\kappa,\ell)}\\(\substr,\kappa,\ell_\substr)=c\cap d}} \Expectation{Y_d\mid Y_c=1}\\
        &\leq  \sum_{\substack{c\in\FC_{(v,\kappa,\ell)}}} p_{1,\kappa}'(p_{2,\kappa}')^2 \sum_{\substack{(\substr,\kappa,\ell_\substr)\in c\\x\neq v}}\sum_{\substack{d\in\FC_{(v,\kappa,\ell)}\\(\substr,\kappa,\ell_\substr)=c\cap d}} p_{1,\kappa}'(p_{2,\kappa}')^2/p'_{(\substr,\kappa,\ell_\substr)}\\
        &\leq \sum_{c\in\FC_{(v,\kappa,\ell)}} p_{1,\kappa}'(p_{2,\kappa}')^2 \sum_{(\substr,\kappa,\ell_\substr)\in c} t(\substr,\kappa,\ell_\substr) p_{1,\kappa}'(p_{2,\kappa}')^2/p'_{(\substr,\kappa,\ell_\substr)}\\
        &\leq \sum_{c\in\FC_{(v,\kappa,\ell)}} p_{1,\kappa}'(p_{2,\kappa}')^2 \sum_{(\substr,\kappa,\ell_\substr)\in c} s_1^2s_2s_3\theta_{(\substr,\kappa,\ell_\substr)} p_{1,\kappa}'(p_{2,\kappa}')^2/p'_{(\substr,\kappa,\ell_\substr)}\\
        &\leq t(v,\kappa,\ell) (p'_{1,\kappa})^2(p'_{2,\kappa})^4   \cdot 10 s_1^2s_2s_3\epsoracleV^2\theta_{(v,\kappa,\ell)}/c_2,
    \end{align*}
    where we used $|\FC_{(v,\kappa,\ell)}|=t(v,\kappa,\ell)$ and \Cref{eq:nodeoracle} in last inequality.

    We now apply a case distinction to obtain the lemma.
    
    First, assume $t(v,\kappa,\ell)>(1+\epsoracleV)s_3\theta_{(v,\kappa,\ell)}$. Then
    \begin{align*}
        \Var{q(v)} &\leq \frac{\epsoracleV^2}{c_2}\Expectation{q(v)}^2 +  10\epsoracleV^2 t(v,\kappa,\ell)^2 (p'_{1,\kappa})^2(p'_{2,\kappa})^4/c_2 \\
        &\leq \frac{\epsoracleV^2}{c_2}\Expectation{q(v)}^2+\frac{\epsoracleV^2}{c_2}\Expectation{q(v)}^2 
        = \frac{\epsoracleV^2}{100L}\Expectation{q(v)}^2
    \end{align*} 
    for large enough $c_2>200L$. By applying \cref{lemma:chebyshev} we get
    \begin{align*}
        \Prob{q(v)<p_{1,\kappa}'(p_{2,\kappa}')^2s_3\theta_{(v,\kappa,\ell)}} &\leq \Prob{q(v)<p_{1,\kappa}'(p_{2,\kappa}')^2s_3\frac{1}{1+\epsoracleV}t(v,\kappa,\ell)} \\
        &= \Prob{q(v)<\frac{1}{1+\epsoracleV}\Expectation{q(v)}}\\
        &\leq \frac{4}{\epsoracleV^2} \frac{\Var{q(v)}}{\Expectation{q(v)}^2}\leq \frac{4}{100L}.
    \end{align*}

    Second, assume $t(v,\kappa,\ell)<(1-\epsoracleV)s_3\theta_{(v,\kappa,\ell)}$. Then $\Expectation{q(v)}<(1-\epsoracleV)s_3\theta_{(v,\kappa,\ell)}p_{1,\kappa}'(p_{2,\kappa}')^2=:M_H$, and thus
    \begin{align*}
        \Var{q(v)} &\leq \frac{\epsoracleV^2}{c_2}\Expectation{q(v)}^2 +  10\epsoracleV^2 s_3\theta_{(v,\kappa,\ell)}^2 (p'_{1,\kappa})^2(p'_{2,\kappa})^4/c_2 \\
        &\leq \frac{\epsoracleV^2}{c_2}M_H^2+\frac{\epsoracleV^2}{c_2}M_H^2 
        = \frac{\epsoracleV^2}{100L}M_H^2
    \end{align*} 
    for large enough $c_2>200L$. By applying \cref{lemma:chebyshev} we get
    \begin{align*}
        \Prob{q(v)>p_{1,\kappa}'(p_{2,\kappa}')^2s_3\theta_{(v,\kappa,\ell)}} 
        &\leq \Prob{q(v)>\frac{1}{1-\epsoracleV}M_H}\\
        &\leq \frac{4}{\epsoracleV^2} \frac{\Var{q(v)}}{M_H^2}\leq \frac{4}{100L}.
    \end{align*}
    Taking a union bound over all up to $L$ oracle calls shows that all oracle answers are correct with probability at least $\frac{96}{100}$, which proves the lemma.
\end{proof}

The accuracy guarantee for all other cases $\ell(v)\neq S_1$ follows completely analogous to \cref{lemma:oraclenodesaccuracyS1}.
We get the following intermediate result by deriving a lemma similar to \cref{lemma:oraclenodesaccuracyS1} for all four possible labels $\ell(v)$ and then using that all of these four lemmas are satisfied simultaneously with probability at least $(\frac{96}{100})^4>\frac{84}{100}$.
\begin{lemma}
    \label{lemma:oraclenodesaccuracy}
    With probability $84/100$ the node oracle correctly classifies all nodes $(v,\kappa,\ell)$ with $t(v,\kappa,\ell)\notin[(1-\epsoracleV)s_3\theta_{(v,\kappa,\ell)},(1+\epsoracleV) s_3\theta_{(v,\kappa,\ell)}]$.
\end{lemma} 
\smallskip
\paragraph{Shifting Thresholds.}
The analysis above shows that the node oracle is accurate for nodes for
which the heaviness is not too close to the heaviness threshold. 
Now we call a node $(v,\kappa,\ell)$ \emph{bad} if $t(v,\kappa,\ell)\in[(1-\epsoracleV)s_3\theta_{(v,\kappa,\ell)},(1+\epsoracleV) s_3\theta_{(v,\kappa,\ell)}]$ and the oracle is called on $(v,\kappa,\ell)$.
Next, we show that by randomly selecting the shift $s_3$, with probability $99/100$, the oracle is not called on any bad nodes.

Recall that \cref{algo:count} randomly selects one of approximately $100L$ different values for $s_3$ to shift the heaviness thresholds by a small multiplicative factor. Here, recall that $L$ was the upper bound on the number of queries we obtain from \Cref{eq:upper-bound-queries} and that $s_3\in\CS_3:=\{(1+\frac{1}{200L})^i\mid 0\leq i <\log_{1+\frac{1}{200L}}(2)\}$.
We set the oracle accuracy $\epsoracleV=1/(600L)$ such that
$$(1+\epsoracleV)\left(1+\frac{1}{200L}\right)^i
<(1-\epsoracleV)\left(1+\frac{1}{200L}\right)^{i+1}$$
implying that all the possible oracle uncertainty intervals $[(1-\epsoracleV)s\theta_{(\wedge,\kappa,\ell)},(1+\epsoracleV)s\theta_{(\wedge,\kappa,\ell)}]$ for $s\in\CS_3$ are disjoint. 
As the node heaviness $t(v,\kappa,\ell)$ is independent of $s_3$, any node can be bad for at most one shift. 
Thus, at most $L$ shifts admit a bad node and the probability of picking a shift admitting a bad node is at most $$\frac{L}{|\CS_3|}\leq \frac{L}{100L} \leq\frac{1}{100},$$
where we used $|\CS|\geq\log_{1+\frac{1}{200L}}(2)\geq \ln(2) \cdot 200L\geq 100L$.

Thus, the node oracle correctly answers all queries with probability at least $\frac{83}{100}$ by \cref{lemma:oraclenodesaccuracy}.
\smallskip
\paragraph{Space.}
It remains to show the space complexity of the oracle.

As before, all node sets are stored implicitly using hash functions.
Note that $q':=p_{1,\kappa}'p_{2,\kappa}'=\tildeO(1/\sqrt{T})$ and $q_{\delta^2}:=p_{1,\kappa}'p_{2,\delta^2\kappa}=\tildeO(1/\sqrt{T})$ are independent of $\kappa$. 
Each edge in the first pass is collected by the oracle with probability at most $\max\{q',q_{\delta^2}\}$.
In the second pass, the oracle realizes at most $100 T\log(T) q'q_{\delta^2}\in\tildeO(1)$ configurations with probability $99/100$. 
Thus, a Markov bound implies that the space usage of the node oracle is $\tildeO(m\max\{q',q_{\delta^2}\})=\tildeO(m/\sqrt{T})$ with probability at least $1/100$.

Now we can prove \Cref{lemma:oraclenode}.

\begin{proof}[Proof of \cref{lemma:oraclenode}]
With probability $99/100$, the node oracle is called at most $L$ times. With probability $99/100$ none of the $L$ calls is for a bad node, i.e., a node $(v,\kappa,\ell)$ with $t(v,\kappa,\ell)\in[(1-\epsoracleV)s_3\theta_{(v,\kappa,\ell)},(1+\epsoracleV) s_3\theta_{(v,\kappa,\ell)}]$. Thus, by \cref{lemma:oraclenodesaccuracy}, all oracle calls are correctly answered with probability $84/100$.

The statement follows from a union bound and the above space considerations.
\end{proof}

\subsection{Edge Oracles}
\label{sec:oracleedge}
We now design edge oracles which, when queried on a labeled node $(e, \kappa, \ell)$ during \cref{algo:oracle} or inside the node oracle, output
\begin{center}
    $\oracle(e,\kappa,\ell)=\heavy$ if and only if $t(e,\kappa,\ell)>s_2\theta_{(e,\kappa,\ell)}$
\end{center}
with a large constant probability.

To implement the oracle, recall the definition of refined heaviness for edges:
If $p_{(e,\kappa,\ell)} = p_{2,\kappa}^2$, then we defined
\begin{align*}
    t(e,\kappa,\ell)=|\{c\in \FC\mid (e,\kappa,\ell)\in c,\; t(y,\kappa,\ell')\leq s_1\theta_{(y,\kappa,\ell')} \text{ for all wedges } (y,\kappa,\ell')\in c\}|,
\end{align*}
and otherwise we set $t(e,\kappa,\ell)=t(e)$.

We next build two different oracles for these two different cases. In the case of $t(e,\kappa,\ell)=t(e)$, we can build a simple oracle that is based on edge sampling. Otherwise, we again resort to node sampling to estimate $t(e,\kappa,\ell)$.

To understand why we need two different types of edge oracles, let us consider edges $(e,\kappa,\ell)$ with $t(e,\kappa,\ell)=t(e)$. Here, we use an edge-sampling based oracles to estimate the true heaviness $t(e)$. Concretely, if $p_{(e,\kappa,\ell)}=p_{1,\kappa}^2$, we sample a set of edges~$Q$ which contains each edge with probability $p_{2,\kappa}^2\log(n)$. Then in the next stream pass, we collect all edges with one endpoint incident to $e$ and the other endpoint incident to a sampled edge in~$Q$ (see below for the details). This requires $\tildeO(L'mp_2^2)=\tildeO(mp_2^2)$ space. As all four-cycles containing $e$ are sampled independently with this procedure, we get an accurate oracle that is easy to analyze via Chernoff bounds.
However, this edge sampling approach breaks for edges  $(e,\kappa,\ell)$ with $p(e,\kappa,\ell)=p_{2,\kappa}^2$ as then we need to sample the edges in $Q$ with probability $p_{1,\kappa}^2$, giving us a space usage $\tildeO(mp_{1,\kappa}^2)$ which exceeds our budget of $\tildeO(m/\sqrt{T})$ for $\kappa\gg T^{1/4}$.
Thus, we instead use a node-sampling based oracle for this case, which works similarly to the node heaviness oracle.

Additionally, we need to define the oracle $\oracle_{\text{valid}}(e)$ for an edge $e$ which is used in \cref{algo:oracle} to check whether a configuration containing $e$ is valid.
Note that this oracle only checks whether $t(e)\leq\sqrt{T}/\delta^2$ and that $\sqrt{T}/\delta^2$ is also the threshold for edges $(e',\kappa',\ell')$ with $p_{(e',\kappa',\ell')}=p_{1,\kappa'}p_{2,\kappa'}$.
Thus, we can reuse the edge-sampling based oracle for this purpose.

As before, we observe that the number $L'$ of calls to the edge heaviness oracle is $L'=\tildeO(1)$: 
The edge oracle is called at most $L$ times during \cref{algo:oracle}, where $L=\tildeO(1)$ is as in \Cref{eq:upper-bound-queries}, and since
the edge oracle is called at most four times for each configuration realized by the node oracle.

Next, we prove the following guarantee for our edge heaviness oracle.
\begin{lemma}
\label{lemma:oracleedgemain}
    Assuming access to perfect wedge heaviness oracles, there is an edge heaviness oracle which correctly answers all queries by \cref{algo:oracle} and the node oracles with probability $9/10$. The oracle requires $\tildeO(m/\sqrt{T})$ space.
\end{lemma}

\smallskip
\paragraph{Edge-Sampling Based Edge Oracles.}
We first consider the case of edges $(e,\kappa,\ell)$ with $p(e,\kappa,\ell)=p_{1,\kappa}p_{2,\kappa}$.
Let $Z$ be the set of such labeled edges for which the oracle is called during \cref{algo:count,algo:oracle}.
Note that $|Z|\leq L'\in\tildeO(1)$ as we argued above. 
Let $\epsoracleE=\frac{1}{600L'}$ be the edge oracle accuracy.

The oracle works as follows. During the first stream pass, we sample an edge set $Q^{\text{edge}}$ by sampling each edge with probability $p'=\log(n)p_1p_2/\epsoracleE^2$. 
In the second stream pass, we do nothing, and in the third stream pass, we collect all edges connecting edges in $Q^{\text{edge}}$ to edges in $Z$, i.e., all edges in $E[Q^{\text{edge}},Z]$.

Note that the reason why we only collect the edges of the induced subgraph in third stream pass is that we only know $Z$ after the second stream pass finished, since some of the configurations yielding the set $Z$ are only realized during Pass~2 of \cref{algo:count} and Pass~2 of the node oracle. 

For each edge $e\in Z$, we then count the number $q(e)$ of four-cycles formed using $e$ and one disjoint edge in the sample $Q^{\text{edge}}$.
The heaviness oracle is then defined as 
\begin{equation*}
    \oracle(e,\kappa,\ell)=\begin{cases}
        \light & \text{ if } q(e)\leq p's_2\theta_e,\\
        \heavy & \text{ if } q(e)> p's_2\theta_e,
    \end{cases}
\end{equation*}
where $\theta_e=\frac{\sqrt{T}}{\delta^2}$ is the heaviness threshold for $e$.

The oracle for edges $(e,\kappa,\ell)$ with $p(e,\kappa,\ell)=p_{1,\kappa}^2$ works analogously by sampling with $p_2^2\log(n)/\epsoracleE^2$ instead of $p'$.
We now obtain the following guarantee.
\begin{lemma}
    \label{lemma:oracleedgeChernoff}
    With high probability the edge oracle correctly classifies all edges $(e,\kappa,\ell)$ with $p(e,\kappa,\ell)\in \{p_{1,\kappa}p_{2,\kappa},p_{1,\kappa}^2\}$ and $t(e)\notin[(1-\epsoracleE)s_2\theta_e,(1+\epsoracleE)s_2\theta_\wedge]$.
    The oracle's space usage is $\tildeO(m/\sqrt{T})$.
\end{lemma}
\begin{proof}
    We prove the statement for $p(e,\kappa,\ell)=p_{1,\kappa}p_{2,\kappa}$. The other case follows analogously.
    
    Observe that $\Expectation{q(e)}=t(e)p'$.
    Let $Y_f\in\{0,1,2\}$, $f\in Q^{\text{edge}}$, be the random variable indicating the number of four-cycles detected by the oracle using $e$ and $f$. Note that $Y_f$ and $Y_{f'}$ are independent for any $f\neq f'$.
    We have $q(e)=\sum_f Y_f$ and get a high accuracy oracle using Chernoff bounds and a case distinction.
    
    \emph{Case 1: $t(e)>(1+\epsoracleE)s_2\theta_{(e,\kappa,\ell)}$.} Then by \cref{lemma:chebyshev} (with $B=2$) we get
    \begin{align*}
        \Prob{q(e)< p'\theta_{(e,\kappa,\ell)}} &\leq \Prob{q(e)\leq \frac{1}{1+\epsoracleE}\Expectation{q(e)}} \leq \exp\left(-\frac{\epsoracleE^2}{24}s_2\theta_{(e,\kappa,\ell)}p'\right)\leq n^{-2},
    \end{align*}
    since $p'\geq p_{1,\kappa}p_{2,\kappa}\log(n)/\epsoracleW^2\geq 2\cdot12\log(n)c_1/(s_2\theta_{(\wedge,\kappa,\ell)}\epsoracleW^2)$ for large enough $c_1$.

    \emph{Case 2: $t(e)<(1-\epsoracleE)s_2\theta_{(e,\kappa,\ell)}$.}
    Then
    \begin{align*}
        \Prob{q(e)> p'\theta_{(e,\kappa,\ell)}} &\leq \Prob{q(e)\geq \frac{1}{(1-\epsoracleE)}\Expectation{q(e)}} \leq \exp\left(-\frac{\epsoracleE^2}{24}s_2\theta_{(e,\kappa,\ell)}p'\right)\leq n^{-2}.
    \end{align*}
    
    Finally, observe that the expected space required of the oracle is $O(|Z||Q^{edge}|)=O(|Z|mp')=\tildeO(m/\sqrt{T})$.
\end{proof}

\paragraph{Node-Sampling Based Edge Oracle.}
Next, we present our edge heaviness oracle for edges $(e,\kappa,\ell)$ with $p_{(e,\kappa,\ell)}=p_{2,\kappa}^2$, for which we want to estimate the refined heaviness $t(e,\kappa,\ell)$. It is based on node sampling and detecting four-cycles containing $e$ by sampling nodes and collecting edges (similar to the node oracle in \cref{sec:oraclenode}).

The oracle receives the node sets $R_{2a,\kappa},R_{2b,\kappa}$ and $R_{2a,\kappa}',R_{2b,\kappa}'$ for all $\kappa$ as input. Here, $R_{2a,\kappa}',R_{2b,\kappa}'$ are node sets sampled by the node oracle.
After the second stream pass, it may be queried for any edge $(e,\kappa',\ell)$ such that $(e,\kappa,\ell)$ in $E[R_{2a,\kappa},R_{2b,\kappa}]\cup E[R_{2a,\kappa}',R_{2b,\kappa}']$ any pair $\kappa,\kappa'\in\CI$ with $\delta^2\kappa'\leq \kappa\leq \kappa'$.
It is implemented as follows.

Before the first stream pass, the oracle samples each node with probability $p'_{1,\kappa}=c_3p_{1,\kappa}/\epsoracleE^2$, where $c_3=\tildeO(1)$ is set below, into sets $Q_{1a,\kappa}$ and $Q_{1b,\kappa}$.
During the first stream pass, the oracle collects the edges $E[R_{2a,\kappa'}\cup R_{2a,\kappa'}',Q_{1b,\kappa}]\cup E[R_{2b,\kappa'}\cup R_{2b,\kappa'}',Q_{1a,\kappa}]$ for all $\delta^2\kappa'\leq \kappa\leq \kappa'$. 
During the second stream pass, the oracle collects the edges in $E[Q_{1a,\kappa},Q_{1b,\kappa}]$ which complete a four-cycle with any edge $e\in E[R_{2a,\kappa'},R_{2b,\kappa'}]\cup E[R_{2a,\kappa'}',R_{2b,\kappa'}'] $ and one edge from $E[Q_{1a,\kappa},R_{2b,\kappa'}\cup R_{2b,\kappa'}']$ and one edge from $E[Q_{1b,\kappa},R_{2a,\kappa'}\cup R_{2a,\kappa'}']$ for any pair $\kappa,\kappa'$ with $\delta^2\kappa'\leq \kappa\leq \kappa'$.

Note that in the algorithm above it is crucial that $R_{2a,\kappa},R_{2a,\kappa}'$ and $R_{2b,\kappa},R_{2b,\kappa}'$ are already known before the first stream pass, and that after the first stream pass we also stored all edges in $E[R_{2a,\kappa},R_{2b,\kappa}]$ and  $E[R_{2a,\kappa}',R_{2b,\kappa}']$ in \Cref{algo:count} and the node oracle.

To compute an answer to the query $(e,\kappa,\ell)$, the oracle considers all four-cycles $(u,v,a,b)$ where $e=(u,v)$, $a\in Q_{1a,\kappa},b\in Q_{1b,\kappa}$, since each such four-cycle corresponds to a configuration $c=(A,\kappa,a,b)$ with $(e,\kappa,\ell)\in c$.
For each such configuration $(A,\kappa,a,b)$ realized by the oracle, we check whether $c$ counts towards $(e,\kappa,\ell)$, that is, whether the wedges in $c$ are light. For this, we assume access to a perfect wedge oracle.

Let $q(e)$ be the number of four-cycles found that count towards $t(e,\kappa,\ell)$.
The heaviness oracle is then defined as 
\begin{equation*}
    \oracle(e,\kappa,\ell)=\begin{cases}
        \light & \text{ if } q(e)\leq (p_1')^2s_2\theta_e,\\
        \heavy & \text{ if } q(e)> (p_1')^2s_2\theta_e,
    \end{cases}
\end{equation*}
where $\theta_e=\frac{T}{\delta^2\kappa^2}$ is the heaviness threshold for $e$.

Next, we summarize the guarantees for our oracle.
\begin{lemma}
    \label{lemma:oracleedgeR2R2}
    With probability $96/100$, the edge oracle correctly classifies all edges $(e,\kappa,\ell)$ with $p_{(e,\kappa,\ell)}=p_{2,\kappa}^2$ and $t(e,\kappa,\ell)\notin[(1-\epsoracleE)s_2\theta_e,(1+\epsoracleE)s_2\theta_e]$.
    The oracle uses space $\tildeO(m/\sqrt{T})$.
\end{lemma}
\begin{proof}
    Let $Y_c$ be the binary random variable whether $c$ was realized by the edge oracle.
    Let $\FC_{(e,\kappa,\ell)}$ denote the set of configurations containing $(e,\kappa,\ell)$.
    
    We have $\Expectation{q(e)}=t(e,\kappa,\ell)(p_1')^2$ and 
    \begin{align*}
        \Var{q(e)}
        &\leq\sum_{\substack{c\in\FC_{(v,\kappa,\ell)}}}\Var{Y_c^2} + \sum_{\substack{c,d\in\FC_{(e,\kappa,\ell)}\\V(c\cap d)=3}} \Expectation{Y_dY_c}\\
        &\leq \Expectation{q(e)} + t(e,\kappa,\ell)\cdot 2s_1\theta_{c\cap d} (p_1')^3 \\
        &\leq \Expectation{q(e)} + 2s_1t(e,\kappa,\ell) \sqrt{T}/(\delta \kappa)(p_1')^3\\
        &\leq \frac{\epsoracleE^2}{c_3}\Expectation{q(e)}^2 + \frac{2s_1\epsoracleE^2}{c_3}t(e,\kappa,\ell)\theta_{(e,\kappa,\ell)}(p_1')^4 
    \end{align*}
    for large enough $c_3>200L'$. Here, we used that $\theta_{c\cap d}=\sqrt{T}/(\delta \kappa)$ and $\theta_{(e,\kappa,\ell)}=T/(\delta^2 \kappa^2)$. Note that this argument is the same as \Cref{eq:nodeoracle}.

    As before, we finish the proof by distinguishing two cases and applying a union bound.
    
    \emph{Case 1: $t(e,\kappa,\ell)\geq (1+\epsoracleE) s_2\theta_{(e,\kappa,\ell)}$.}
    Then $$\Var{q(e)}\leq \frac{\epsoracleE^2}{c_3}\Expectation{q(e)}^2 +\frac{4\epsoracleE^2}{c_3}\Expectation{q(e)}^2 = \frac{\epsoracleE^2}{100L'}\Expectation{q(e)}^2$$
    for large enough $c_3$. By applying \cref{lemma:chebyshev} we get
    \begin{align*}
        \Prob{q(e) < \theta_{(e,\kappa,\ell)}s_2(p_1')^2} 
        &< \Prob{q(e) < \frac{1}{(1+\epsoracleE)} t(e,\kappa,\ell)(p'_1)^2} = \Prob{q(e) < \frac{1}{(1+\epsoracleE)}\Expectation{q(e)}} \\
        &\leq \frac{4}{\epsoracleE^2} \frac{\Var{q(e)}}{\Expectation{q(e)}^2}
        \leq \frac{4}{\delta^2} \frac{\epsoracleE^2}{100L'}=\frac{4}{100L'}.
    \end{align*}

    \emph{Case 2: $t(e,\kappa,\ell)\leq (1-\epsoracleE) \theta_{(e,\kappa,\ell)}$.}
    Then $\Expectation{q(e)}\leq (1-\epsoracleE)s_2\theta_{(e,\kappa,\ell)}(p_1')^2=:M_H$, and
    $$\Var{q(e)}\leq \frac{\epsoracleE^2}{c_3}M_H^2 +\frac{4\epsoracleE^2}{c_3}M_H^2 = \frac{\epsoracleE^2}{100L'}\Expectation{q(e)}^2$$
    for large enough $c_3$. By applying Chebyshev's inequality (\cref{lemma:chebyshev}) we get that
    \begin{align*}
        \Prob{q(e) \geq \theta_{(e,\kappa,\ell)}(p_1')^2} 
        &= \Prob{q(e) \geq \frac{1}{1-\epsoracleE}M_H} \\
        &\leq \frac{4}{\epsoracleE^2} \frac{\Var{q(e)}}{\Expectation{q(e)}^2}
        \leq \frac{4}{\epsoracleE^2} \frac{\epsoracleE^2}{100L'} = \frac{4}{100}.
    \end{align*}

    Taking a union bound over all up to $L'$ oracle calls shows that all oracle answers are correct with probability at least $96/100$.
\end{proof}

\paragraph{Shifting Thresholds.}
Analogously to \cref{sec:oraclenode}, we now show that with probability $99/100$ none of the $L'$ edges $(e,\kappa,\ell)$ for which the oracle is called fulfills $t(e,\kappa,\ell)\in[(1-\epsoracleE)s_2\theta_e,(1+\epsoracleE)s_2\theta_e]$.
The algorithm randomly selects the shift $s_2$ as one out of $100L'$ different values in $\CS_2:=\{(1+\frac{1}{200L'})^i\mid 0\leq i < \log_{1+\frac{1}{200L'}}(2)\}$. By setting $\epsoracleE=\frac{1}{600L'}$, we ensure that at most $L'$ out of $|\CS_2|$ shifts admit edges with  $t(e,\kappa,\ell)\in[(1-\epsoracleE)s_2\theta_e,(1+\epsoracleE)s_2\theta_e]$.
Thus, by randomly selecting the shift $s_3$, with probability $99/100$, the oracle is not called on any bad nodes.
We refer to \cref{sec:oraclenode} for the details.

\begin{proof}[Proof of \cref{lemma:oracleedgemain}]
    With probability $99/100$, none of the $L'$ oracle calls regard a bad node due to the shift $s_2$. Conditioned on this,
    the edge-sampling oracle calls are correct with high probability by \cref{lemma:oracleedgeChernoff} and the node-sampling oracle calls are correct with probability $96/100$ by \cref{lemma:oracleedgeR2R2}.

    The oracle uses $\tildeO(m/\sqrt{T})$ space in the first pass and with probability $99/100$ stores at most $\tildeO(1)$ additional edges in $E[Q_{1a,\kappa},Q_{2a,\kappa}]$ in the second pass. 
    In total, we get a success probability of at least $\frac{9}{10}$ by a union bound.
\end{proof}

\subsection{Wedge Oracles}
\label{sec:oraclewedge}
We now show how to build a wedge heaviness oracle which, when queried on a labeled wedge $(\wedge,\kappa,\ell)$ during \cref{algo:oracle} or inside the node or edge oracle, outputs
\begin{center}
    $\oracle(\wedge,\kappa,\ell)=\heavy$ if and only if $t(\wedge,\kappa,\ell)>s_2\theta_{(\wedge,\kappa,\ell)}$
\end{center}
with a high probability.
Recall that refined heaviness $t(\wedge,\kappa,\ell)=t(\wedge)$ equals true heaviness for wedges. However, the threshold $\theta_{(\wedge,\kappa,\ell)}$ depends on the labeling.

We use node sampling to build the wedge oracles. Intuitively, given a wedge $\wedge$ with heaviness threshold $\theta$, we sample nodes with probability $p>1/\theta$ and count how many of them form a four-cycle with $\wedge$ by collecting the required edge sets. 
Observe that for wedges $(\wedge,\kappa,\ell)$ with $p_{(\wedge,\kappa,\ell)}=p_{1,\kappa}^2p_{2,\kappa}$ we have $\theta_{(\wedge,\kappa,\ell)}=\kappa/\delta$ and thus $p_{2,\kappa}>1/\theta_{(\wedge,\kappa,\ell)}$. Similarly, $p_{1,\kappa}>1/\theta_{(\wedge,\kappa,\ell')}$ for wedges $(\wedge,\kappa,\ell')$ with $p_{(\wedge,\kappa,\ell')}=p_{1,\kappa}p_{2,\kappa}^2$.

We now state the oracle's guarantees in the following lemma.
\begin{lemma}
    \label{lemma:oraclewedgemain}
    There is a wedge heaviness oracle which correctly answers all queries by \cref{algo:oracle} and by the node and edge oracles with probability $98/100$. The oracle requires $\tildeO(m/\sqrt{T})$ space.
\end{lemma}

Let $L''$ be the number of oracle calls made to the wedge heaviness oracle by \cref{algo:oracle}, the node oracles and the edge oracles together. We have $L''\in \tildeO(1)$ as only $\tildeO(1)$ configurations and four-cycles are sampled in each of these algorithms. Let $\epsoracleW=\frac{1}{600L''}$ denote the accuracy of our oracle. 

The oracle receives all node sets sampled by \cref{algo:count}, the node oracles and the edge oracles as input. Additionally, it samples two node sets $Q_{1,\kappa}^{\text{wedge}}$ and $Q^{\text{wedge}}_{2,\kappa}$ with probabilities
$$p_{1,\kappa}'=p_{1,\kappa}\log(n)/\epsoracleW^2\quad\text{ and }\quad p_{2,\kappa}'=p_{2,\kappa}\log(n)/\epsoracleW^2,$$ respectively before the first stream pass. 

Now let $V_{1,\kappa}$ (resp. $V_{2,\kappa}$) denote the union of all node sets which were sampled in the node oracle, edge oracle or \cref{algo:count} with probability $\tilde{\Theta}(p_{1,\kappa})$ (resp. $\tilde{\Theta}(p_{2,\kappa})$), that is 
\begin{align*}
    V_{1,\kappa}&=S_{1,\kappa}\cup S_{1,\kappa}'\cup R_{1a,\kappa}\cup R_{1b,\kappa}\cup R_{1a,\kappa}'\cup R_{1b,\kappa}'\cup Q_{1a}\cup Q_{1b}\\
    V_{2,\kappa}&=S_{2,\kappa}\cup S_{2,\kappa}'\cup R_{2a,\kappa}\cup R_{2b,\kappa}\cup R_{2a,\kappa}'\cup R_{2b,\kappa}'
\end{align*}
Curcially, observe that for all queries $(\wedge,\kappa',\ell)$ to the wedge oracle, either $\wedge$ contains two nodes from $V_{1,\kappa}$ and one from $V_{2,\kappa}$ or two nodes from $V_{2,\kappa}$ and one from $V_{1,\kappa}$ for any $\delta^2\kappa\leq \kappa'\leq \kappa$.

Thus, during the first two passes we collect edges between $V_{1,\kappa'}, V_{2,\kappa'}$ and $Q_{1,\kappa}^{\text{wedge}}, Q^{\text{wedge}}_{2,\kappa}$ as in the previous oracles. 
In particular, in the first pass we collect the edge sets $E[V_{1,\kappa'},Q^{\text{wedge}}_{2,\kappa}]$, $E[V_{2,\kappa'},Q^{\text{wedge}}_{1,\kappa}]$ and $E[V_{2,\kappa'},Q^{\text{wedge}}_{2,\kappa}]$ for all $\delta^2\kappa\leq \kappa'\leq \kappa$.
In the second pass, we collect all edges in $E[V_{1,\kappa'},Q^{\text{wedge}}_{1,\kappa}]$ which complete a four-cycle with one edge in $E[V_{1,\kappa'},V_{2,\kappa'}]$, one edge in $E[V_{2,\kappa'},V_{2,\kappa'}]$ and one edge in $E[V_{2,\kappa'},Q^{\text{wedge}}_{2,\kappa}]$ for any $\delta^2\kappa\leq \kappa'\leq \kappa$.

When the oracle is queried for some wedge $(\wedge,\kappa,\ell)$, we need to distinguish whether $p_{(\wedge,\kappa,\ell)}=p_{1,\kappa}p_{2,\kappa}^2$ or $p_{(\wedge,\kappa,\ell)}=p_{1,\kappa}^2p_{2,\kappa}$. We only provide details for the latter case here. The other case follows analogously.

Let $(\wedge,\kappa,\ell)$ be any wedge for which the oracle is queried with $p_{(\wedge,\kappa,\ell)}=p_{1,\kappa}^2p_{2,\kappa}$. 
The oracle then counts the number $q(\wedge)$ of four-cycles containing $\wedge$ and one disjoint node in $Q^{\text{wedge}}_2$.
The heaviness oracle is then defined as 
\begin{equation*}
    \oracle(\wedge,\kappa,\ell)=\begin{cases}
        \light & \text{ if } q(\wedge)\leq p'_2s_1\theta_{(\wedge,\kappa,\ell)},\\
        \heavy & \text{ if } q(\wedge)> p'_2s_1\theta_{(\wedge,\kappa,\ell)},
    \end{cases}
\end{equation*}
where $\theta_{(\wedge,\kappa,\ell)}=\frac{\kappa}{\delta}$ is the heaviness threshold for $\wedge$.

As all the four-cycles counting towards $q(\wedge)$ are found independently (conditioned on having sampled $\wedge$), we get a high accuracy oracle using Chernoff bounds:

\begin{lemma}
    \label{lemma:oraclewedge}
    With high probability the wedge oracle correctly classifies all wedges $(\wedge,\kappa,\ell)$ with $t(\wedge)\notin[(1-\epsoracleW)s_1\theta_{(\wedge,\kappa,\ell)},(1+\epsoracleW)s_1\theta_{(\wedge,\kappa,\ell)}]$.
\end{lemma}
\begin{proof}
    We prove the oracle accuracy using \cref{lemma:chernoff}. We only proof the lemma for wedges $(\wedge,\kappa,\ell)$ with $p_{(\wedge,\kappa,\ell)}=p_{1,\kappa}^2p_{2,\kappa}$. The other case follows analogously.
    
    Observe that $\Expectation{q(\wedge)}=t(\wedge)p_{2,\kappa}'$ as each node $v$ which forms a four-cycle with $\wedge$ is sampled with $p_{2,\kappa}$ into $Q_{2,\kappa}^\text{wedge}$ and if sampled, the edges between $v$ and the endpoints of $\wedge$ are collected by the oracle.

    To bound the variance,
    first assume $t(\wedge)>(1+\epsoracleW)\theta_{(\wedge,\kappa,\ell)}$.
    Then
    \begin{align*}
        \Prob{q(\wedge)< p_{2,\kappa}'\theta_{(\wedge,\kappa,\ell)}} &\leq \Prob{q(\wedge)\leq \frac{1}{1+\epsoracleW}\Expectation{q(\wedge)}} \leq \exp\left(-\frac{\epsoracleW^2}{12}s_1\theta_{(\wedge,\kappa,\ell)}p_{2,\kappa}'\right)\leq n^{-4},
    \end{align*}
    as $p_{2,\kappa}'\geq p_{2,\kappa}\log(n)/\epsoracleW^2\geq \log(n)c_1/(\theta_{(\wedge,\kappa,\ell)}\epsoracleW^2)\geq 4\cdot12\log(n)/(s_1\theta_{(\wedge,\kappa,\ell)}\epsoracleW^2)$ for large enough $c_1$.

    Now assume $t(\wedge)<(1-\epsoracleW)\theta_{(\wedge,\kappa,\ell)}$.
    Then
    \begin{align*}
        \Prob{q(\wedge)> p_2'\theta_{(\wedge,\kappa,\ell)}} &\leq \Prob{q(\wedge)\geq \frac{1}{(1-\epsoracleW)}\Expectation{q(\wedge)}} \leq \exp\left(-\frac{\epsoracleW^2}{24}s_1\theta_{(\wedge,\kappa,\ell)}p_2'\right)\leq n^{-4},
    \end{align*}
    where the last inequality again assumes large enough $c_1$.
\end{proof}
This shows that the described wedge oracle is accurate with high probability for wedges, for which the heaviness is not too close to the heaviness threshold.

\begin{proof}[Proof of \Cref{lemma:oraclewedgemain}]
    Analogously to the shifting analysis in \cref{sec:oraclenode} we obtain that by randomly selecting the shift $s_1\in\CS_2:=\{(1+\frac{1}{200L''}^i)\mid 0\leq i < \log_{1+\frac{1}{200L''}}(2)\}$ during \cref{algo:count}, all of the $L''$ wedges $(\wedge,\kappa,\ell)$ for which the oracle is called fulfill $t(\wedge)\notin[(1-\epsoracleW)s_1\theta_{(\wedge,\kappa,\ell)},(1+\epsoracleW)s_1\theta_{(\wedge,\kappa,\ell)}]$ with probability $99/100$.
    
    Conditioned on this event, all oracle queries are answered correctly with high probability by \cref{lemma:oraclewedge}.

    The wedge oracle uses $\tildeO(m/\sqrt{T})$ space in the first pass and $\tildeO(1)$ space in the second pass with probability $99/100$.
\end{proof}

We note that the wedge oracle only needs a single stream pass to be able to answer all queries on wedges using nodes from $S_{1,\kappa},S_{2,\kappa},S_{1,\kappa}',S_{2,\kappa}'$ during \cref{algo:oracle} and inside the node oracle. 

\subsection{Proof of \texorpdfstring{\cref{thm:oracles}}{Proposition 6.6}}
We now combine the different oracles to prove \cref{thm:oracles}. As mentioned, we call each oracle at most $\tildeO(1)$ times, which is crucial for the oracle accuracy and space usage.

\begin{proof}[Proof of \cref{thm:oracles}]
    By \cref{lemma:oraclenode,lemma:oracleedgemain,lemma:oraclewedgemain} and a union bound, all heaviness oracle calls by \cref{algo:oracle} are correct with probability at least $\frac{2}{3}$.
    By \cref{observ:LLMoracle}, the oracle $\IsLLM(\cdot)$ answers all queries correctly if all heaviness oracle answers are correct.
    As the node oracle, edge oracle and wedge oracle each use $\tildeO(m/\sqrt{T})$ space, the total space usage of the \isLLM oracle is $\tildeO(m/\sqrt{T})$.
\end{proof}

\section*{Acknowledgments}
We are grateful to the anonymous reviewers for their helpful comments on restructuring our paper.
This research has been funded by the Vienna Science and Technology Fund (WWTF) [Grant ID: 10.47379/VRG23013].
Supported in part by NSFC Grant 62272431 and the Innovation Program for Quantum Science and Technology (Grant No. 2021ZD0302901).

\bibliographystyle{plainnat}
\bibliography{references} 

\end{document}